\PassOptionsToPackage{numbers,sort&compress}{natbib}  
\documentclass[a4paper,fleqn]{cas-sc}

\usepackage[numbers]{natbib}
\usepackage[toc,page]{appendix}
\usepackage[dvipsnames]{xcolor}
\usepackage{tikz}
\usepackage{amsthm}
\usepackage{subcaption}
\usepackage{algpseudocode}
\usepackage{multirow}
\usepackage{mathrsfs}
\usepackage{algorithm}
\usepackage{float}
\usetikzlibrary{arrows.meta, calc}

\newtheorem{prop}{Proposition} 
\newtheorem{lem}{Lemma}
\newtheorem{remark}{Remark}

\def\tsc#1{\csdef{#1}{\textsc{\lowercase{#1}}\xspace}}
\tsc{WGM}
\tsc{QE}

\begin{document}
	
	\newcommand{\rhostar}{$\frac{\rho^\mathrm{f}}{\rho^\mathrm{s}}$}
	\newcommand{\hshf}{$\frac{h^\mathrm{f}}{h^\mathrm{s}}$}
	\newcommand{\kt}{\tilde{k}}
	\newcommand{\rf}{\rho_{\mathrm{f}}}
	\newcommand{\rs}{\rho_{\mathrm{s}}}
	\newcommand{\hs}{h_{\mathrm{s}}}

	\let\WriteBookmarks\relax
	\def\floatpagepagefraction{1}
	\def\textpagefraction{.001}
	
	\shorttitle{}    
	
	\shortauthors{}  
	
	\title [mode = title]{A Momentum-weighted Interface Stabilization Approach for
		Fluid-Structure Interaction with Strong Added-Mass Effects}  
	
	
	
	%
	
	
	\author[1]{Nihar B. Darbhamulla}[orcid=0000-0002-0546-7641]
    \ead{nihar.darbhamulla@ubc.ca}
	
	\credit{Writing - original draft, Conceptualization, Methodology, Software, Visualization, Validation, Investigation, Formal analysis }
	
	\affiliation[1]{organization={Dept. of Computer Science, The University of
			British Columbia},
		city={Vancouver},
		postcode={V6T 1Z4}, 
		state={BC},
		country={Canada}}
	
	\author[2]{Xiaoyu Mao}
	
	
	\ead{dgcsmaoxiaoyu@gmail.com}
	
	
	\credit{ Writing - review \& editing, Conceptualization, Methodology, Software, Visualization, Investigation, Formal analysis }
	

	\author[2]{Aditya Narkhede}[orcid=0009-0004-4281-611X]
	
	\credit{Writing - review \& editing, Software, Investigation, Visualization, Validation,  Formal analysis }
    \ead{aditya.narkhede@ubc.ca}
	
	
	

	\author[3]{David Brown}
    \credit{Writing - review \& editing, Conceptualization, Resources, Project administration, Funding acquisition}

    \author[2]{Rajeev K. Jaiman}[orcid=0000-0002-8346-3486]
	\affiliation[2]{organization={Dept. of Mechanical Engineering, The University
			of British Columbia},
		city={Vancouver},
		postcode={V6T 1Z4}, 
		state={BC},
		country={Canada}}
    \cortext[1]{Corresponding author}
    \cormark[1]
    \ead{rjaiman@mech.ubc.ca}

    \affiliation[3]{organization={Synopsys Inc.},
		city={Waterloo},
		state={ON},
		country={Canada}}
        

	\credit{Writing - review \& editing, Conceptualization, Software, Methodology, Supervision, Formal analysis, Resources, Project administration, Funding acquisition}


	
	
	\begin{abstract}
Partitioned fluid-structure interaction (FSI) schemes permit the reuse of optimized standalone solvers but can become unstable under strong fluid added-mass effects, particularly when structural inertia is small relative to fluid inertia. This work traces the instability to the discrepancy between the interface pressure Schur complements of the monolithic and Dirichlet-Neumann partitioned systems, as revealed by a fixed-point analysis of the linearized FSI problem. Motivated by this analysis, a momentum-weighted interface stabilization (MWIS) is proposed, in which the interface pressure block of the fluid subproblem is augmented by an operator constructed from the fluid and structural interface momentum operators. The augmentation is consistent by construction, vanishing at convergence so that the monolithic solution is recovered. A sparse diagonal approximation of the operator, derived from element metrics, limits intersolver data exchange to quantities comparable in size to standard interface coupling data while requiring minimal modifications to the standalone solvers. Numerical experiments, spanning a linear model and established nonlinear FSI benchmarks, demonstrate improved coupling robustness in regimes of strong added-mass where the standard Dirichlet-Neumann scheme diverges.
	\end{abstract}
	
	
	\begin{highlights}
		\item A momentum-weighted stabilization is developed for partitioned FSI coupling.
		\item Schur-complement analysis establishes stability and convergence conditions.
        \item A sparse interface augmentation preserves consistency with monolithic FSI.
		\item Numerical benchmarks demonstrate robustness under strong added-mass effects.
	\end{highlights}

	
	\maketitle
	
	\section{Introduction} \label{sec:Introduction}
Fluid-structure interaction (FSI) is a fundamental multiphysics phenomenon in which fluid motion and structural deformation are dynamically coupled through the exchange of momentum and the enforcement of kinematic compatibility at their common interface. Such interactions are ubiquitous in natural and engineered systems, ranging from cardiovascular flows and flexible biological tissues to offshore structures, marine propulsion systems, and flow-induced vibrations. Accurate and computationally efficient prediction of these coupled responses is therefore essential to understand their underlying physics and to assess the performance, reliability, and stability of engineering systems. Despite substantial advances in computational mechanics, the numerical simulation of FSI remains challenging because the fluid and structural governing equations must be solved subject to strongly coupled interface conditions. These difficulties are compounded by nonlinear fluid dynamics, potentially large structural deformations, evolving fluid domains, and the need to accurately transfer interface tractions and kinematic information between distinct discretizations. In particular, the relative magnitudes of fluid and structural inertia strongly influence the stability and convergence of the coupled numerical system, especially when the structural response is sensitive to fluid loads. Consequently, developing FSI formulations that simultaneously achieve accuracy, numerical robustness, and computational efficiency remains an important objective in computational mechanics
\citep{TEZDUYAR20062002,FERNANDEZ2005127,wall2007strong,DETTMER20065754,2008_quaini,
		BADIA20087027,idelsohn_2009,
		LIU2014687,JAIMAN2016187, li_2018, Delaisse2023,10.1063/5.0182801}.
	
	Accurate and robust strongly coupled schemes have been proposed for FSI
	problems, wherein the solution at each timestep involves solving the fully
	coupled nonlinear fluid and structural equations. The underlying geometric
	nonlinearity associated with the deforming domains is either handled implicitly
	(Newton) \citep{TEZDUYAR20062002} or explicitly (quasi-Newton)
	\citep{TEZDUYAR20062002, DETTMER20065754, LIU2014687}. The nonlinear interface
	coupling is achieved by formulating a linearized FSI problem at each step and
	solving the monolithic block-partitioned linear system \citep{TEZDUYAR20062002,
		LIU2014687} or adopting a block-partitioned approach wherein the Jacobian
	information is cross-utilized within the fluid and structure block sub-problems
	\citep{FERNANDEZ2005127, DETTMER20065754}. Both classes of methods are
	exceptionally robust, and their performance has been demonstrated for
	challenging problems involving fluids and structures of comparable inertia,
	commonly known in the literature as added-mass effects
	\citep{2005_Causin_FSI}. However, these methods are computationally expensive, in which the former class of methods involves the solution to ill-conditioned linear systems, while the latter class of methods requires multiple Jacobian constructions within a single timestep. Furthermore, this class of methods restricts the use of highly optimized standalone solvers developed for the purpose of elucidating fluid and structural dynamics. 
	
Alternately, the computational cost can be partially mitigated through the use of semi-implicit coupling schemes. Typically, block-partitioned schemes  \citep{TEZDUYAR20062002, wall2007strong} are designed where the solution to one system lags the other. However, the use of lagged operators induces added-mass instability, wherein only structural inertia is accounted for in the motion of the interface, inducing destabilizing perturbations in the fluid flow.  This instability is particularly pronounced in incompressible flows, where pressure enforces mass conservation and strongly couples the fluid response to interface motion. Within the scope of partitioned schemes, additional stabilization approaches are necessary to recover the robustness of the strongly-coupled schemes. Some approaches adopted
include under-relaxation of interface coupling conditions
\citep{2005_Causin_FSI}, use of Robin-Robin conditions at the interface
\citep{BADIA20087027}, and algebraic fractional step methods \citep{2008_quaini}.

Furthermore, these schemes often require individual subsolvers to access information beyond the solution variables of interest (e.g., discretization, linearization, and iterative schemes), thereby restricting the use of existing highly optimized solvers designed for individual systems \citep{FERNANDEZ2005127}. Alternatively, computationally efficient semi-implicit coupling schemes have been proposed, although they may impose time-step restrictions \citep{2008_quaini} or depend on tuning parameters that control the interface coupling \citep{wall2007strong,BADIA20087027, JAIMAN2016187}. However, these schemes restore the flexibility of using highly optimized standalone solvers, while introducing challenges in ensuring robust convergence at the interface. Similarly, heuristic stabilization methods have been developed \citep{idelsohn_2009}, which use local fluid and structural properties at the interface. Such methods are minimally invasive to the fluid solver, but may raise consistency concerns depending on the stabilization formulation. In parallel, variations of quasi-Newton approaches based on interface data have been used to accelerate and stabilize fixed-point iterations \citep{Delaisse2023}. However, their effectiveness may be limited during initialization, while their performance and computational cost depend on the available iteration history and its management. More recently, surrogate modeling approaches aimed at reproducing the convergence characteristics of monolithic solvers have been developed \citep{10.1063/5.0182801}, although the interpretability of the underlying models may be limited.

In this paper, a novel momentum-weighted interface stabilization (MWIS) approach is proposed for Dirichlet-Neumann partitioned FSI. By construction, the method is independent of user-defined interface relaxation and boundary-condition tuning parameters, does not require heuristic history windows, and preserves consistency with the monolithic FSI formulation at convergence. The method relies on augmenting the interface pressure operator in the fluid subproblem to mitigate the added-mass instability. A fixed-point analysis of the linearized FSI problem in monolithic and partitioned frameworks identifies the interface coupling information omitted in the partitioned fluid solve. By augmenting the interface pressure Schur complement with an operator constructed from the fluid and structural interface momentum operators, the missing coupling contribution is incorporated into the partitioned formulation. Using approximations of the pressure Schur complement \citep{2014_Elman_FE} and sparse approximations of the interfacial momentum operators, a minimally invasive pressure augmentation operator is constructed. The additional interface storage scales as $O(n^{1-\frac{1}{d}})$, where $n$ denotes the number of bulk degrees of freedom and $d$ the spatial dimension, assuming comparable mesh resolution. This scaling is of the same order as the interface data exchanged in a conventional partitioned FSI scheme. The MWIS scheme demonstrates robust coupling for established FSI benchmarks at structure-to-fluid mass ratios as low as $m^\star = 0.1$ for benchmark problems established in the literature \citep{2006_turek_FSIBench, LIU2014687}.
	
The remainder of this article is organized as follows. Section \ref{sec:simplifiedFSI} examines the added-mass instability in a Dirichlet-Neumann partitioned scheme using a simplified linear FSI problem \citep{2005_Causin_FSI}. An augmentation operator is proposed and a framework for fixed-point convergence analysis is presented. Section \ref{sec:MWIS} extends the augmentation strategy to a general linearized FSI problem and establishes the spectral criteria that guide the design of an interfacial operator to mitigate added \textit{mass} instability. Subsequently, a sparse discrete momentum-weighted interface stabilization (MWIS) operator is constructed. Section \ref{sec:results} evaluates the consistency and numerical stability of the augmented partitioned FSI scheme through established FSI benchmark problems. Finally, Section \ref{sec:conclusion} summarizes the main findings and provides directions for future work.
	
    \section{Added-mass instability and the principle of interface stabilization}
    \label{sec:simplifiedFSI}
    In the current section, a model linear fluid-structure interaction problem is considered, and the source of the added-mass instability is identified through the fixed-point analysis. Subsequently, an augmented partitioned scheme is presented and the criteria governing convergence are evaluated. Using standard conventions in continuum mechanics, it is assumed that the fluid occupies a domain $\Omega^\mathrm{f} \subset \mathbb{R}^{d}$ and is confined within the boundaries $\partial \Omega^\mathrm{f}$. Similarly, the structure is assumed to occupy the domain $\Omega^\mathrm{s}$ and located within the boundaries $\partial \Omega^\mathrm{s}$. Correspondingly, the fluid-structure interface is defined by $\Gamma^\mathrm{fs} = \partial\Omega^\mathrm{f}\cap \partial \Omega^\mathrm{s}$. 
	
	\subsection{Model fluid-structure problem}

\def\domL{20}         
\def\domR{5}          
\def\wallT{0.2}       
\def\xNorm{18}        
\def\xSig{6}          
\def\bInset{0.2}      

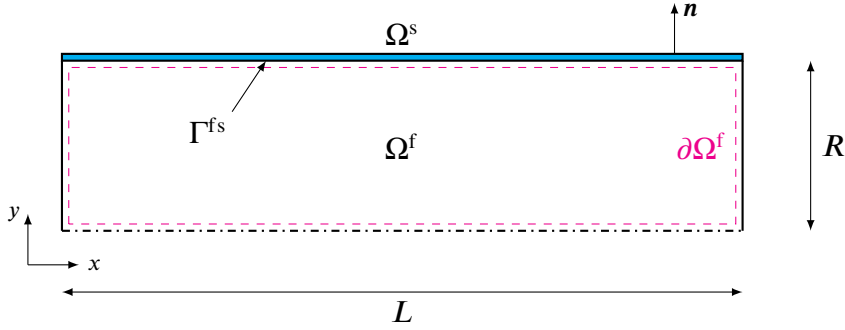
\begin{figure}[pos=H]
	\centering
    \begin{tikzpicture}[scale=0.45] 
    	\coordinate (A)  at (0,0);
    	\coordinate (B)  at (\domL,0);
    	\coordinate (C)  at (\domL,\domR);
    	\coordinate (Ff) at (0,\domR);
     
    	\draw[thick, dash dot] (A) -- (B);
    	\draw[thick] (A) -- (Ff);
    	\draw[thick] (B) -- (C);
     
    	\node at (\domL/2, \domR/2) {\large $\Omega^\mathrm{f}$};
        \draw[dashed, thin, magenta] (\bInset,\bInset) rectangle (\domL-\bInset,\domR-\bInset);
        \node[left, magenta] at (\domL-\bInset, \domR/2) {\large $\partial\Omega^\mathrm{f}$};
     
    	\coordinate (P1) at (0,\domR);
    	\coordinate (P2) at (\domL,\domR);
    	\coordinate (P3) at (\domL,\domR+\wallT);
    	\coordinate (P4) at (0,\domR+\wallT);
    	\fill[cyan] (P1) -- (P2) -- (P3) -- (P4) -- cycle;
    	\draw[thick] (P1) -- (P2) -- (P3) -- (P4) -- cycle;
    	\node[above] at (\domL/2, \domR+\wallT) {\large $\Omega^\mathrm{s}$};
     
    	\draw[latex-] (\xSig,\domR) -- (\xSig-1.0,\domR-1.5);
    	\node[below left] at (\xSig-0.9,\domR-1.4) {\large $\Gamma^\mathrm{fs}$};
     
    	\draw[-latex] (\xNorm,\domR+\wallT) -- (\xNorm,\domR+\wallT+1.5);
    	\node[right] at (\xNorm, \domR+\wallT+1.3) {\normalsize $\boldsymbol{n}$};
     
    	\draw[-latex] (-1.0,-1.0) -- (0.5,-1.0);
    	\draw[-latex] (-1.0,-1.0) -- (-1.0,0.5);
    	\node[right] at (0.5,-1.0) {$x$};
    	\node[left]  at (-1.0,0.5) {$y$};
     
    	\draw[latex-latex] (\domL+2.0,0) -- (\domL+2.0,\domR);
    	\node[right] at (\domL+2.05, \domR/2) {\large $R$};
     
    	\draw[latex-latex] (0,-1.8) -- (\domL,-1.8);
    	\node[below] at (\domL/2, -1.8) {\large $L$};
     
    \end{tikzpicture}
\caption{Representative domain for the linear fluid-structure interaction problem.}
\label{fig:causin}
\end{figure}
    
	For the model problem, the fluid occupies a rectangular domain $\Omega^\mathrm{f} \subset \mathbb{R}^2$. To simplify the analysis, nonlinear convective effects are neglected, and the fluid is assumed to be inviscid and incompressible. The structure is taken to be a one dimensional entity (i.e. $\Omega^\mathrm{s} \subset \mathbb{R}$), whose material coordinates coincide with the initial configuration of the fluid-structure interface $\Gamma^\mathrm{fs}$, with the outward unit normal denoted by $\boldsymbol{n}$. A generalized string model is adopted to model the structural dynamics, wherein the deformation occurs normal to $\Gamma^\mathrm{fs}$ in the initial configuration. Under these conditions, the fluid governing equations
	take the form
	\begin{subequations}
	\begin{align}
			\rho^\mathrm{f}\frac{\partial \boldsymbol{u}^\mathrm{f}}{\partial t} + \nabla p &= 0 ~ &&\text{in} \quad (0,T)\times\Omega^\mathrm{f}, \label{eq:mom1D}\\
			\nabla\cdot\boldsymbol{u}^\mathrm{f} &= 0 ~ &&\text{in} \quad (0,T)\times\Omega^\mathrm{f}, \label{eq:div1D}\\
			\boldsymbol{u}^\mathrm{f}\cdot\boldsymbol{n} &= \frac{\partial \eta^\mathrm{s}}{\partial t} ~ &&\text{on} \quad (0,T)\times\Gamma^\mathrm{fs},
	\end{align}	\label{eq:linearFSI}	
	\end{subequations}
	and the structural governing equation is
	\begin{equation}
		\rho^\mathrm{s}t^\mathrm{s}\frac{\partial^2 \eta^\mathrm{s}}{\partial t^2} +
		a\eta^\mathrm{s} - b\frac{\partial^2 \eta^\mathrm{s}}{\partial x^2} = p \quad
		\text{in} \quad (0,T)\times\Omega^\mathrm{s}.
		\label{eq:structure}
	\end{equation}
	Here, $\rho^\mathrm{f}$, $\boldsymbol{u}^\mathrm{f}$ and $p$ denote the fluid density, velocity and pressure, while $\rho^\mathrm{s}$, $t^\mathrm{s}$ and
	$\eta^\mathrm{s}$ denote the density, thickness, and normal displacement of the structure. The coefficients of the structural model are $a =
	Et^\mathrm{s}/R^2(1-\nu^2)$ and $b = \kappa_T G t^\mathrm{s}$, where $E$, $\nu$	and $G$ are the Young's modulus, Poisson's ratio and shear modulus of the structural material. $R$ is the initial width of the fluid domain, and $\kappa_T$ is the Timoshenko shear correction factor. 
	
	
	Taking the divergence of \eqref{eq:mom1D} and combining with \eqref{eq:div1D} allows the fluid velocity to be eliminated and leads to the Poisson equation for pressure,
	\begin{equation}
		\Delta p = 0 \quad \text{in} \quad (0,T)\times\Omega^\mathrm{f},
		\label{eq:pressurePoisson}
	\end{equation}
	where the fluid state is entirely determined by the pressure and conditions 
imposed on the boundary of the fluid domain $\partial\Omega^\mathrm{f}$. Similarly, the kinematic
	boundary condition at the fluid-structure interface $\Gamma^\mathrm{fs}$ becomes
	\begin{equation}
		\frac{\partial p}{\partial \boldsymbol{n}} = -\rho^\mathrm{f}\frac{\partial^2
			\eta^\mathrm{s}}{\partial t^2} \quad \text{on} \quad (0,T)\times\Gamma^\mathrm{fs}.
		\label{eq:pressureBC}
	\end{equation}
	
	\subsection{Discretized monolithic system and its exact solution}
	The fluid and structural governing equations are discretized in space using the 
	standard Galerkin finite element approach. To advance the structure in time, a
	Newmark time integration scheme is employed, which expresses the structural
	governing equation in terms of its acceleration $\ddot{\eta}^\mathrm{s}$. 
	With the appropriate choice of function spaces, the weak forms corresponding to the governing equations for the fluid \eqref{eq:pressurePoisson} and structure \eqref{eq:structure} are
	\begin{subequations}
		\begin{align}
			\frac{1}{\rho^\mathrm{f}}(\nabla q, \nabla p)_{\Omega^\mathrm{f}} +
			(q,\ddot{u}_\Sigma)_\Sigma &= 0,\\
			(v_\mathrm{s},
			\rho^\mathrm{s}t^\mathrm{s}\ddot{\eta}^\mathrm{s})_{\Omega^\mathrm{s}} +
			ac_s\Delta t^2(v_\mathrm{s}, \ddot{\eta}^\mathrm{s})_{\Omega^\mathrm{s}} +
			bc_s\Delta t^2 \left(\frac{\partial v_\mathrm{s}}{\partial x},\frac{\partial
				\eta^\mathrm{s}}{\partial x}\right)_{\Omega^\mathrm{s}} -
			(v_\mathrm{s},p_\Sigma)_\Sigma &= 0,
		\end{align}  \label{eq:varlinearFSI}  
	\end{subequations}
	respectively. Here, $v_\mathrm{s}$ and $q$ denote the test functions used to
	construct the variational form, $(\cdot,\cdot)_\Omega$ denotes the inner product $L^2$ on a domain $\Omega$, and $c_\mathrm{s}$ is a constant arising from the
	time-integration scheme for the structure.
	
	Assembling \eqref{eq:varlinearFSI} on the bases of finite elements and
	partitioning the pressure degrees of freedom into interior sets (denoted by the
	subscript $I$) and interface sets (denoted by the subscript $\Sigma$), the
	discrete monolithic system is constructed as follows:
	\begin{align}
		\mathcal{A}u = \begin{bmatrix}
			K_\mathrm{s} & 0 & -B_\Sigma^T\\
			0 & A_I & A_{I\Sigma}\\
			B_\Sigma & A_{I\Sigma}^T & A_{\Sigma}
		\end{bmatrix} \begin{bmatrix}
			\ddot{\eta}^\mathrm{s}\\
			p_I\\
			p_\Sigma
		\end{bmatrix} = \begin{bmatrix}
			0\\
			f_I\\
			0
		\end{bmatrix} = f, \label{eq:monolithic}
	\end{align}
	with $f_I$ representing the contribution of the prescribed boundary conditions to $\partial \Omega^\mathrm{f}\Gamma^\mathrm{fs}$. Let
	$\{\phi_i\}$ and $\{\psi_j\}$ be the bases of the structural and pressure
	spaces. The individual blocks are then given by
	\begin{subequations}
		\begin{align*}
			(K_\mathrm{s})_{ij} &= \left(\phi_i, \left(\rho^\mathrm{s}t^\mathrm{s} +
			ac_\mathrm{s}\Delta t^2\right)\phi_j\right)_{\Omega^\mathrm{s}} +
			bc_\mathrm{s}\Delta t^2 \left(\frac{\partial \phi_i}{\partial x},\frac{\partial
				\phi_j}{\partial x}\right)_{\Omega^\mathrm{s}}, \\
			(A_I)_{ij} &= \frac{1}{\rho^\mathrm{f}}\left(\nabla \psi_i, \nabla
			\psi_j\right)_{\Omega^\mathrm{f}}, \\
			(A_{I\Sigma})_{ij} &= \frac{1}{\rho^\mathrm{f}}\left(\nabla \psi_i, \nabla
			(\mathcal{R}\psi_j)\right)_{\Omega^\mathrm{f}}, \\
			(A_{\Sigma})_{ij} &= \frac{1}{\rho^\mathrm{f}}\left(\nabla
			(\mathcal{R}\psi_i), \nabla (\mathcal{R}\psi_j)\right)_{\Omega^\mathrm{f}},
			\\
			(B_\Sigma)_{ij} &= \left(\mathcal{R}\psi_i, \phi_j\right)_\Sigma, \nonumber
		\end{align*}
	\end{subequations}
	where $\mathcal{R}$ denotes the restriction of the pressure space to the
	interface.
	
	The exact solution of the coupled problem \eqref{eq:monolithic} is obtained
by the block $LU$ decomposition of $\mathcal{A}$, expressed as
	\begin{equation}
		\mathcal{A} = \begin{bmatrix}
			K_\mathrm{s} & 0 & -B_\Sigma^T\\
			0 & A_I & A_{I\Sigma}\\
			B_\Sigma & A_{I\Sigma}^T & A_{\Sigma}
		\end{bmatrix} = \begin{bmatrix}
			K_\mathrm{s} & 0 & 0\\
			0 & A_I & 0\\
			B_\Sigma & A_{I\Sigma}^T & S_2\\
		\end{bmatrix}\begin{bmatrix}
			I &0 & -K_\mathrm{s}^{-1}B_\Sigma^T\\
			0 & I & A_I^{-1}A_{I\Sigma}\\
			0 & 0 & I\\
		\end{bmatrix}, \label{eq:decomposition}
	\end{equation} where
	\begin{align}
		S_2 &= \underbrace{A_\Sigma - A_{I\Sigma}^TA_I^{-1}A_{I\Sigma}}_{S_\Sigma} +
		\underbrace{B_\Sigma K_\mathrm{s}^{-1}B_\Sigma^T}_{C},
	\end{align}
	is the Schur interface complement, which naturally decomposes into two
	components, $S_\Sigma$ and $C$, representing the effects of the bulk of the fluid and 
	of the structural motion on the interface. The term
	$\rho^\mathrm{f}S_\Sigma^{-1}$ is commonly referred to in the literature as the
	added-mass operator.
	Substituting \eqref{eq:decomposition} into \eqref{eq:monolithic}, the exact
solution is obtained over a time step as
	\begin{subequations}
		\begin{align}
			p_\Sigma &= -(S_\Sigma + C)^{-1}\left(A_{I\Sigma}^TA_I^{-1}f_I\right),
			\label{eq:monolithicItf}\\
			p_I &= A_I^{-1}f_I + A_I^{-1}A_{I\Sigma}\left(S_\Sigma +
			C\right)^{-1}A_{I\Sigma}^TA_I^{-1}f_I,\\
			\ddot{\eta}^\mathrm{s} &= K_\mathrm{s}^{-1}B_\Sigma^Tp_\Sigma.
		\end{align}    \label{eq:monolithicSol}
	\end{subequations} 
	
	\subsection{Partitioned problem and analysis of convergence}\label{sec:partitioned}
	Alternatively, the solution to the FSI problem \eqref{eq:varlinearFSI} can be constructed using a Dirichlet-Neumann partitioning scheme, in which structural and fluid problems are solved in a sequence. However, under partitioning, the convergence is conditional, and it is necessary to establish conditions under which the iterates converge to the monolithic solution \eqref{eq:monolithicSol}. 
	
	Let the fixed-point iterations be performed such that the $(k+1)^\mathrm{th}$
	iteration begins with the interface and interior pressures $(p_\Sigma^k, p_I^k)$
	from the previous iterate. The structure is advanced first, subject to the
	interface pressure held fixed at $p_\Sigma^k$, which gives
	\begin{equation}
		\ddot{\eta}^{k+1} = K_\mathrm{s}^{-1}B_\Sigma^T p_\Sigma^{k}.
		\label{eq:accStruct}
	\end{equation}
	
	The resulting interface acceleration is then imposed on the fluid through the
	interface condition \eqref{eq:pressureBC}. Consequently, the discretized fluid equations
	take the form
	\begin{equation}
		\begin{bmatrix}
			A_I & A_{I\Sigma}\\
			A_{I\Sigma}^T & A_\Sigma
		\end{bmatrix} \begin{bmatrix}
			p_I^{k+1}\\
			p_\Sigma^{k+1}
		\end{bmatrix} = \begin{bmatrix}
			f_I\\
			-B_\Sigma\ddot{\eta}^{k+1}
		\end{bmatrix}, \label{eq:fluidPart}
	\end{equation}
	and eliminating $p_I^{k+1}$ yields
	\begin{subequations}
		\begin{align}
			p_\Sigma^{k+1} &= -S_\Sigma^{-1}Cp_\Sigma^{k} -
			S_\Sigma^{-1}\left(A_{I\Sigma}^TA_I^{-1}f_I\right),\label{eq:itfpresPart}\\
			p_I^{k+1} &= A_I^{-1}f_I + A_I^{-1}A_{I\Sigma}S_\Sigma^{-1}\left(Cp_\Sigma^k
			+ A_{I\Sigma}^TA_I^{-1}f_I\right).
		\end{align}\label{eq:partitionedIterates}
	\end{subequations}
	
	Comparing the partitioned pressure iterates \eqref{eq:partitionedIterates} with
	the monolithic solution \eqref{eq:monolithicSol} reveals that: (1) the
	monolithic solution inverts $S_\Sigma + C$ in a single step, while the
	partitioned scheme inverts $S_\Sigma$ alone with $C$ lag to the previous
	iterate, and (2) the operator $C$ represents the response of the structure to a
	perturbation in the pressure at the interface, projected back into the fluid pressure
	space. It is this effect which the fluid sub-problem no longer accounts for once the partitioned solution process is adopted.
	
	Since interior pressure is entirely determined by boundary data $f_I$
and interface pressure $p_\Sigma$, convergence is governed by interface 
pressure alone. This implies that the convergence of iterations is governed by
	the maximal eigenvalue of the iteration matrix $S_\Sigma^{-1}C$. As shown below,
	the operator $C$ reduces to a scalar multiple of the interface mass matrix $M_\Sigma$, which motivates the characterization of the spectrum of $S_\Sigma$ relative to $M_\Sigma$.
	\begin{lem}
		The eigenvalues of $S_\Sigma^{-1}M_\Sigma$ are $\rho_\mathrm{f}\mu_i$ with
		eigenvectors $v_i$, and the maximum eigenvalue is
		$\rho_\mathrm{f}\mu_\mathrm{max}$. Equivalently, the eigenvalue problem
		$S_\Sigma v = \sigma M_\Sigma v$ has the spectrum $\sigma_i = (\rho_\mathrm{f}
		\mu_i)^{-1}$\label{lem:eigsSM}
	\end{lem}
	\begin{proof}
		Since $S_\Sigma$ is symmetric positive definite (SPD) \citep{boyd2004convex}, there exists a square root
		matrix $S_\Sigma^\frac12$. Therefore, the matrix $S_\Sigma^{-1}M_\Sigma$ can be
		expressed as
		\begin{equation}
			S_\Sigma^{-1}M_\Sigma = S_\Sigma^{-\frac12}(S_\Sigma^{-\frac12} M_\Sigma
			S_\Sigma^{-\frac12})S_\Sigma^\frac12,
		\end{equation} 
		
		Let $\tilde{M}_\Sigma = S_\Sigma^{-\frac12} M_\Sigma S_\Sigma^{-\frac12}$, and
		let $(\rho^\mathrm{f}\mu_i, w_i)$ be its eigenpairs, with $v_i =
		S_\Sigma^{-\frac12}w_i$. Since $M_\Sigma$ is SPD, $\tilde{M}_\Sigma$ is SPD by
		congruence, and therefore the eigenvalues $\rho^\mathrm{f}\mu_i$ are real and
		positive. Now consider the following
		\begin{align}
			S_\Sigma^{-1}M_\Sigma v_i &= S_\Sigma^{-\frac12} \tilde{M}_\Sigma
			S_\Sigma^{\frac12}v_i = S_\Sigma^{-\frac12} \tilde{M}_\Sigma w_i,\nonumber\\
			&= \rho^\mathrm{f}\mu_i S_\Sigma^{-\frac12}w_i,\nonumber\\
			&= \rho^\mathrm{f}\mu_i v_i. \label{eq:eigenvalueSM}
		\end{align} From \eqref{eq:eigenvalueSM}, it follows that
		$(\rho^\mathrm{f}\mu_i, v_i)$ are the eigenpairs of $S_\Sigma^{-1}M_\Sigma$.
		Correspondingly, the maximal eigenvalue of $S_\Sigma^{-1}M_\Sigma$ is
		$\rho^\mathrm{f}\mu_\mathrm{max}$. It trivially follows that $\sigma_i =
		(\rho_\mathrm{f} \mu_i)^{-1}$
	\end{proof} 
	
	With regard to the partitioned scheme, the following result can be shown to
	hold true:
	\begin{prop}
		Let $\rho^\mathrm{f}\mu_\mathrm{max}$ be the largest eigenvalue of
		$S_\Sigma^{-1}M_\Sigma$, and $b = 0$. Then the iterates
		\eqref{eq:partitionedIterates} converge to the monolithic solution
		\eqref{eq:monolithicSol} if 
		\begin{equation}
			\frac{\rho^\mathrm{f}\mu_\mathrm{max}}{\rho^\mathrm{s}t^\mathrm{s} +
				ac_\mathrm{s}\Delta t^2} <1 \label{eq:DNineq}
		\end{equation}
	\end{prop}
	\begin{proof}
		Using Eqs.\eqref{eq:monolithicItf} and \eqref{eq:itfpresPart},
		$p_\Sigma^{k+1}$ can be expressed as
		\begin{align}
			p_\Sigma^{k+1} - p_\Sigma &= -S_\Sigma^{-1}C(p_\Sigma^k - p_\Sigma).
			\label{eq:fixedPiteration}
		\end{align} Defining the iteration matrix $T = -S_\Sigma^{-1}C$, a necessary
		and sufficient condition for the iterates to converge to the monolithic solution
		is 
		\begin{equation}
			\rho(T) = \mathrm{max}_\lambda ~|\lambda(T)| < 1, \label{eq:Spectralcond}
		\end{equation} where $\rho(T)$ is the spectral radius of the iteration matrix.
		For $b = 0$, $C = \frac{1}{\rho^\mathrm{s}t^\mathrm{s} + ac_s\Delta t^2
		}M_\Sigma$, where $M_\Sigma$ is the matrix of mass of the interface. Using Lemma
		\ref{lem:eigsSM}, the spectral radius of the iteration matrix is given by
		\begin{align}
			\rho(T) &= \mathrm{max}_\lambda~|\lambda(-S_\Sigma^{-1}C)| ,\nonumber\\
			&=\frac{1}{\rho^\mathrm{s}t^\mathrm{s} + ac_s\Delta t^2} \mathrm{max}_\lambda
			~|\lambda(-S_\Sigma^{-1}M_\Sigma)|,\nonumber\\
			&= \frac{\rho^\mathrm{f}\mu_\mathrm{max}}{\rho^\mathrm{s}t^\mathrm{s} +
				ac_s\Delta t^2}
		\end{align} and invoking the inequality \eqref{eq:Spectralcond} yields the
		condition \eqref{eq:DNineq}.
	\end{proof} 
	
	The above analysis illustrates that the convergence of the Dirichlet-Neumann
partition scheme is strictly dependent on the eigenvalue $\mu_\mathrm{max}$,
	which is governed by the geometric properties of the domain. Therefore, in the
	absence of under-relaxation, the scheme is conditionally convergent. These 
	findings are in agreement with those presented by Causin \textit{et al.} in 
	\citep{2005_Causin_FSI}.
	
	\subsection{Augmenting the partitioned scheme for enhancing convergence}\label{sec:aug_partitioned}
	
	To restore the structural response within the fluid solve, an interfacial
	operator $\tilde{C}$ is introduced. Adding $\tilde{C}$ to the interface pressure block of 
	\eqref{eq:fluidPart}, and balancing it on the right-hand side, gives
	\begin{equation}
		\begin{bmatrix}
			A_I & A_{I\Sigma}\\
			A_{I\Sigma}^T & A_\Sigma + \tilde{C}
		\end{bmatrix} \begin{bmatrix}
			p_I^{k+1}\\
			p_\Sigma^{k+1}
		\end{bmatrix} = \begin{bmatrix}
			f_I\\
			-B_\Sigma\ddot{\eta}^{k+1} + \tilde{C}p_\Sigma^k
		\end{bmatrix}. \label{eq:augmented}
	\end{equation}
	The augmentation is consistent by construction, since at a fixed point
	$p_\Sigma^{k+1} = p_\Sigma^k$ the two occurrences of $\tilde{C}$ cancel and
	\eqref{eq:augmented} reduces to \eqref{eq:fluidPart}. Combining
	\eqref{eq:augmented} with the structural solve \eqref{eq:accStruct} and
	eliminating $p_I^{k+1}$, the interface pressure is obtained as
	\begin{equation}
		p_\Sigma^{k+1} = (S_\Sigma + \tilde{C})^{-1}(\tilde{C} - C)p_\Sigma^k -
		(S_\Sigma + \tilde{C})^{-1}A_{I\Sigma}^TA_I^{-1}f_I, \label{eq:itfPresaug}
	\end{equation}
	in which $\tilde{C}$ enters only through the difference $\tilde{C} - C$, so the
	scheme contracts to the extent that $\tilde{C}$ captures the lagged structural
	response.

	To establish the conditions under which the augmented iterations converge to the
	monolithic solution, structural inertia is defined as
	\begin{equation}
		\zeta = \rho^\mathrm{s}t^\mathrm{s} + ac_\mathrm{s}\Delta t^2,
	\end{equation}
	together with the ratio,
	\begin{equation}
		\xi = \frac{bc_\mathrm{s}\Delta t^2}{\zeta},
	\end{equation}
	and denote by $\tilde{k}_i$ and $w_i$ the eigenvalues and corresponding
	eigenvectors of the generalized eigenvalue problem (GEVP)\begin{equation}
		K_\Sigma w = \tilde{k}M_\Sigma w,
	\end{equation}
	where $K_\Sigma$ is the structural stiffness matrix and $M_\Sigma$ the interface
	mass matrix. The following result then holds for the structural response
	operator $C$.
	\begin{lem}
		Given $b\neq 0$, the eigenvalues $\theta_i$ of the GEVP, $Cz = \theta M_\Sigma
		z$ are given by
		\begin{equation}
			\theta_i = \frac{1}{\zeta}\left(\frac{1}{1 + \xi \tilde{k}_i}\right)
			\label{eq:eigGEVPC}
		\end{equation}
		\label{lem:eigsC}
	\end{lem} 
	\begin{proof}
		For the linear FSI problem under consideration, the matrices $B_\Sigma =
		M_\Sigma$ and $K_\mathrm{s} = \zeta M_\Sigma + \xi\zeta K_\Sigma$. It follows that the GEVP takes
		the form
		\begin{equation}
			\frac{1}{\zeta}\left(I + \xi M_\Sigma^{-\frac12}K_\Sigma
			M_\Sigma^{-\frac12}\right)^{-1}\tilde{z} = \theta\tilde{z}, \label{eq:GEVPC}
		\end{equation} where $\tilde{z} = M_\Sigma^\frac12z$. Using congruence
		properties, the eigenvalues of $M_\Sigma^{-\frac12}K_\Sigma M_\Sigma^{-\frac12}$
		are $\tilde{k}_i$. Consequently, it follows that the eigenvalues of
		\eqref{eq:GEVPC} are given by \eqref{eq:eigGEVPC}.
	\end{proof}
	
	\begin{prop}
		Let the matrix $\tilde{C} = \gamma M_\Sigma$. Then the iterates given by
		\eqref{eq:itfPresaug} are convergent under the condition
		\begin{equation}
			\gamma > \frac12\left(\theta_\mathrm{max} -
			\frac{1}{\rho_\mathrm{f}\mu_\mathrm{max}}\right),
		\end{equation}
		and the spectral radius of the corresponding iteration matrix has an upper
		bound when $\gamma$ takes the value
		\begin{equation}
			\gamma^\star = \frac12(\theta_\mathrm{max}+ \theta_\mathrm{min})
		\end{equation}
	\end{prop}
	\begin{proof}
		The iterations given by \eqref{eq:itfPresaug} can be rearranged as
		\begin{equation}
			p_\Sigma^{k+1} - p_\Sigma = -(S_\Sigma + \tilde{C})^{-1}(C -
			\tilde{C})(p_\Sigma^k - p_\Sigma),
		\end{equation} which leads to the definition of the iteration matrix
		$T_\mathrm{D} = -(S_\Sigma  + \tilde{C})^{-1}(C - \tilde{C})$. The iterates
		converge to the monolithic solution under the contractive condition on the
		spectral radius given by \eqref{eq:Spectralcond}. The generalized eigenvalue
		problem corresponding to the iteration matrix $T_\mathrm{D}$ is
		\begin{equation}
			(C- \gamma M_\Sigma)u = -\lambda(S_\Sigma + \gamma M_\Sigma)u,
		\end{equation} where $(\lambda, u)$ forms an eigenpair. Defining the Rayleigh
		coefficients $R_\mathrm{S}(x)$ and $R_\mathrm{C}(x)$ as
		\begin{equation*}
			R_C(x) = \frac{x^TCx}{x^T M_\Sigma x}, ~R_\mathrm{S}(x) = \frac{x^TS_\Sigma
				x}{x^TM_\Sigma x},
		\end{equation*} which satisfy the bounds $\theta_\mathrm{min} \leq
		R_\mathrm{C}(x) \leq \theta_\mathrm{max}$ and $\sigma_\mathrm{min} \leq
		R_\mathrm{S}(x) \leq \sigma_\mathrm{max}$. Using the results of Lemmas
		\ref{lem:eigsSM} and \ref{lem:eigsC}, the eigenvalues $\lambda(u_i)$ can be
		expressed as
		\begin{equation}
			\lambda_i = \frac{\gamma - R_\mathrm{C}(u_i)}{\gamma + R_\mathrm{S}(u_i)},
		\end{equation} leading to the contraction condition
		\begin{equation}
			\rho(T_\mathrm{D}) \leq \sup_{x\neq0} ~ \left|\frac{\gamma -
				R_\mathrm{C}(x)}{\gamma + R_\mathrm{S}(x)} \right| < 1,
		\end{equation} which leads to the sufficient condition
		\begin{align}
			\gamma &> \frac12\left(\theta_\mathrm{max} - \sigma_\mathrm{min}\right) =
			\frac12\left(\theta_\mathrm{max}
			-\frac{1}{\rho_\mathrm{f}\mu_\mathrm{max}}\right) \label{eq:gamma_req}
		\end{align} Using the min-max condition on the spectral radius gives rise to
		the condition
		\begin{align}
			\rho(T_\mathrm{D}) = \left|\frac{\gamma -R_\mathrm{C}(x)}{\gamma +
				R_\mathrm{S}(x)}\right|\leq
			\frac{\mathrm{max}\{|\theta_\mathrm{min}-\gamma|,|\theta_\mathrm{max} -
				\gamma|\}}{\gamma + \sigma_\mathrm{min}} \leq \frac{\theta_\mathrm{max} -
				\theta_\mathrm{min}}{(\theta_\mathrm{max} + \theta_\mathrm{min}) +
				2\sigma_\mathrm{min}}, \label{eq:ineqSpec}
		\end{align} wherein the inequality holds when $\gamma^\star =
		\frac12(\theta_\mathrm{max}+\theta_\mathrm{min})$.
	\end{proof}
	
	\begin{remark}[Regimes of stabilized scheme] \label{rem:regimes}
		The quantity $\theta_\mathrm{max}$ is governed by the structrural properties
, while $\sigma_\mathrm{min}$ is governed by the discretization of the domain and interface. The relative magnitude of the two quantities separates two regimes, where 
\begin{equation}
			\frac{\theta_\mathrm{max}}{\sigma_\mathrm{min}} = \frac{\rf
				\mu_\mathrm{max}}{\zeta(1+\xi \kt_\mathrm{min})} =
			\frac{1}{1+\xi\kt_\mathrm{min}}\rho(T),
		\end{equation} which for $b = 0$ reduces to $\rho(T) > 1$ in the added-mass
		regime. For $\rho(T) < 1$,
		\begin{equation}
			\frac{\theta_\mathrm{max}}{\sigma_\mathrm{min}} < 1 \implies \inf~\gamma < 0,
		\end{equation} which indicates $\gamma = 0$ suffices for high mass-ratio
		regimes, thereby recovering the unaugmented partitioned scheme.
	\end{remark}
	
	From \eqref{eq:ineqSpec}, the spectral radius is governed by
	$\sigma_\mathrm{min}$ in the denominator, which is the smallest eigenvalue
	corresponding to the operator $S_\Sigma$ projected onto the mass matrix at the
	interface. Given that $\sigma_\mathrm{min} =
	(\rho_\mathrm{f}\mu_\mathrm{max})^{-1}$, the spectral radius of $T_\mathrm{D}$
	remains bounded in the limiting case of $\mu_\mathrm{max} \to \infty$.
	Furthermore, combining the eigenvalues given by \eqref{eq:eigGEVPC}, the bound
	on the spectral radius of the iteration matrix given by \eqref{eq:ineqSpec}
leads to the following result:
	\begin{equation}
		\rho(T_\mathrm{D}(\theta,\sigma)) \leq
		\frac{\xi\tilde{k}_\mathrm{max}}{2\zeta\sigma_\mathrm{min}(1+\xi
			\tilde{k}_\mathrm{max}) + (2+\xi\tilde{k}_\mathrm{max})}, \label{eq:ineqspec}
	\end{equation} and 
	\begin{equation}
		\gamma^\star =
		\frac{1}{\zeta}\frac{\left(1+\frac{\xi\tilde{k}_\mathrm{max}}{2}\right)}{1+\xi\tilde{k}_\mathrm{max}}.\label{eq:gammastar}
	\end{equation} From inequality \eqref{eq:ineqspec}, $\sigma_\mathrm{min}$ only
	governs the convergence rate of the partitioned scheme and has no bearing on the
	parameter $\gamma^\star$ used in the augmentation. In the limiting case of
	$\sigma_\mathrm{min} \to 0$, the convergence rate is bounded by the ratio
	\begin{equation}
		\rho \leq \frac{\xi\tilde{k}_\mathrm{max}}{2+\xi\tilde{k}_\mathrm{max}},
	\end{equation} and is governed by the stiffest structural modes, and remains
	convergent for all possible values $(\rho_\mathrm{s}h_\mathrm{s}, a, b, \Delta
	t^2)$. 
	\begin{remark}[Choice of $\gamma = \zeta^{-1}$]
		For $\gamma = \zeta^{-1} \geq \theta_\mathrm{max}$ satisfies inequality
		\eqref{eq:gamma_req} unconditionally. Also, $R_\mathrm{C}(x) \leq
		\theta_\mathrm{max} \leq \gamma$ holds for all $x$, so the eigenvalues
		$\lambda_i$ of the iteration matrix $T_\mathrm{D}$ are positive, and therefore
		lead to a monotonic convergence of the augmented scheme. This suboptimal choice
		of $\gamma$ circumvents the need to estimate $\tilde{k}_\mathrm{max}$.
	\end{remark}

	\subsection{Numerical assessment}
	
	The convergence characteristics of the augmented partitioned scheme,
	\eqref{eq:accStruct} and \eqref{eq:itfPresaug}, are assessed through two sets of
	experiments. The first varies the stiffness of the structural $b$ and the density of the structural
 $\rho^\mathrm{s}$ with the length of the fluid domain set to $L = 6$, isolating
	the role of the structural parameters. The second varies the mesh resolution
	$\Delta_\mathrm{s}$ and $\rho^\mathrm{s}$ with the fluid domain length set to
	$L = 60$, which yields a large added-mass eigenvalue $\mu_\mathrm{max}$ and
	thereby separates the influence of discretization from that of the
	structure. The fluid domain is discretized uniformly using a grid size $\Delta_\mathrm{f} = 0.15$,
	with the matched resolution $\Delta_\mathrm{s} = \Delta_\mathrm{f}$ for the structure. A 
	time step of $\Delta t = 10^{-3}$ and a fluid density $\rho^\mathrm{f} = 1$ is used throughout the
	experiments conducted in this section. The fixed-point iterations are
	terminated when
	\begin{equation}
		\|p_\Sigma^{k+1} - p_\Sigma^k\|_2 < 10^{-8},
		\label{eq:convCriterion}
	\end{equation}
	with the maximum number of iterations capped at $2000$. The augmentation
	parameter is set to
	\begin{equation}
		\gamma = \zeta^{-1}
		\label{eq:gammaChoice}
	\end{equation}
	rather than the optimal value $\gamma^\star$. This choice is deliberate, as
	the optimal value minimizes the bound \eqref{eq:ineqSpec} on the spectral radius
	but requires an estimate of $\tilde{k}_\mathrm{max}$ that is not readily
	available in a practical implementation. In contrast, $\zeta^{-1}$ follows
	directly from the structural parameters and satisfies the convergence condition
	\eqref{eq:gamma_req} unconditionally, so the results below reflect the
	performance of the scheme as it would be implemented in practice rather than its
	best attainable rate. In each case, the spectral radius and iteration counts of
	the augmented scheme are compared with those of the 
Dirichlet-Neumann partitioning scheme of the baseline from Sect.~\ref{sec:partitioned}.
	
	Table~\ref{tab:kappaGhsrho} shows the results of the first numerical experiment
	obtained with $\Delta_\mathrm{s} = \Delta_\mathrm{f} = 0.15$ and the structural 
	thickness $t^\mathrm{s} = 0.1$. When $b = 0$, the structure is governed purely by its inertia and by Lemma~\ref{lem:eigsC},
	\begin{equation}
		\theta_i = \zeta^{-1},
	\end{equation}
	form the eigenmodes.  Consequently, the structural response operator 
	is then $C = \zeta^{-1}M_\Sigma$, which coincides with the augmented interfacial
	operator $\tilde{C} = \gamma M_\Sigma$ chosen here. The iteration matrix 
	subsequently vanishes and the augmented partitioned scheme converges in
	a single iteration, irrespective of the mass ratio $m^\star = \rho^\mathrm{s}/\rho^\mathrm{f}$, which is
	a measure of the induced added-mass effect. The baseline scheme, in contrast,
	exhibits a sharp threshold in the same range, failing once $m^\star < 39.6$.
	
	\begin{table}[pos=H]
		\centering
		\begin{tabular}{ c c c c c c }
			\toprule
			\multirow{2}{*}{$b$} & \multirow{2}{*}{$m^\star$} &
			\multicolumn{2}{c}{Baseline partitioning} & \multicolumn{2}{c}{Augmented
				partitioning} \\
			\cmidrule(lr){3-6}
			&  & spectral radius  & iterations & spectral radius & iterations \\
			\midrule
			\multirow{7}{*}{$0$} 
			& 45   & 0.877514 & 161.3  & 0 & 1 \\
			& 42   & 0.939823 & 339.7  & 0 & 1 \\
			& 39.6  & 0.996930 & 2000& 0 & 1 \\
			& 39.4  & 1.001451 & --   & 0 & 1 \\
			& 10   & 3.873904 & --   & 0 & 1 \\
			& 1.0   & 31.766012& --   & 0 & 1 \\
			& $10^{-3}$   & 158.103601&  --   & 0 & 1 \\
			\midrule
			\multirow{7}{*}{$2.5 \times 10^4$} 
			& 45   & 0.877481  & 161.3  & 0.00273 & 4 \\
			& 42   & 0.939785  & 339.5  & 0.000312 & 4 \\
			& 39.6  & 0.996382  & 2000 & 0.000349 & 4 \\
			& 39.2  & 1.006484  & --   & 0.000356 & 4 \\
			& 10   & 3.873256  & --   & 0.004425 & 5 \\
			& 1.0   & 31.722506 & --   & 0.128471 & 12 \\
			& $10^{-3}$ & 157.12410 & -- & 0.538197 & 35.7 \\
			\midrule
			\multirow{7}{*}{$2.5 \times 10^6$} 
			& 45   & 0.874202  & 160.7  & 0.010388 & 6 \\
			& 42   & 0.936025  & 326.7  & 0.011588 & 6 \\
			& 39.6  & 0.992156  & 2000 & 0.012707 & 6 \\
			& 39.2  & 1.002173  & --   & 0.012909 & 6 \\
			& 10   & 3.810177  & --   & 0.092134 & 10 \\
			& 1.0   & 27.934801 & --   & 0.590752 & 42.5 \\
			& $10^{-3}$ & 93.996658 & -- & 0.888134 & 175.9 \\
			\bottomrule
		\end{tabular}
		\caption{Iteration counts for $\Delta t = 10^{-3}$, $L=6$, $t^\mathrm{s} = 0.1$, varying $b=\kappa_T
			Gt^\mathrm{s}$ and the mass ratio $m^\star$. $-$ indicates the scheme
			failed to converge within the 2000 iterations. }
		\label{tab:kappaGhsrho}
	\end{table} 
	
	When the structural stiffness is not negligible, meaning that $\xi > 0$ and the
	eigenvalues $\theta_i$ are spread over a finite interval, the choice of $\gamma$ given by \eqref{eq:gammaChoice} can no longer replicate $C$ exactly. The residual spread
	is set by $\xi\tilde{k}_\mathrm{max}$, and through \eqref{eq:ineqSpec} governs the
	convergence rate. Table~\ref{tab:kappaGhsrho} indicates that as $b$ grows,
	the spectral radius of the augmented partitioned scheme remains smaller than 1, for
	even the smallest values of $m^\star$. This implies the augmented
	interfacial operator $\tilde{C}$ effectively mitigates the added-mass effect, at the 
	cost of a slowly degrading convergence rate as the dynamics begins to be governed by 
	the structural stiffness. This is in contrast to the baseline partitioned scheme, which
	has a rapidly diverging spectral radius when the added-mass effect becomes prominent. 
	
	The second numerical experiment explores the role of $\mu_\mathrm{max}$ in governing
	the spectral radius of the baseline and augmented partitioning schemes. $\mu_\mathrm{max}$
	is a purely geometric quantity which depends on the slenderness of the fluid domain 
	and the spatial resolution. In order to evaluate this effect, the domain is elongated such that
	$L = 60$, and three different mesh resolutions are considered to
	evaluate the variation in $\mu_\mathrm{max}$. From Table~\ref{tab:iteration3}, it is observed 
	that $\mu_\mathrm{max} \propto (\Delta_\mathrm{s})^{-1}$. For the baseline partitioned scheme,
	$\mu_\mathrm{max}$ enters the numerator of \eqref{eq:DNineq}, so an increasing
	$\mu_\mathrm{max}$ raises the spectral radius, which exceeds one across the
	range of structural densities considered, restricting convergence. In contrast, $\mu_\mathrm{max}$
	impacts the augmented partitiond scheme through $\sigma_\mathrm{min}$ in the denominator
	of \eqref{eq:ineqSpec}, where it governs the rate of convergence without affecting $\gamma^\star$. 
	Refinement therefore leaves the spectral radius for the augmented scheme largely unchanged,
	and in the limiting case of $\sigma_\mathrm{min} \approx 0$, the rate remains bounded by the 
	stiffest structural modes, as established following \eqref{eq:ineqspec}. Convergence of the augmented
	scheme is consequently only challenged once the mass ratio  $O(10^{-2})$ and below, 
	where $\xi$ is large.
	 
	\begin{table}[pos=H]
	\centering
	\begin{tabular}{ c c c c c c c }
		\toprule
		\multirow{2}{*}{$\Delta_\mathrm{s}$} & \multirow{2}{*}{$\mu_\mathrm{max}$} &
		\multirow{2}{*}{$m^\star$} &
		\multicolumn{2}{c}{Baseline partitioning} & \multicolumn{2}{c}{Augmented
			partitioning} \\
		\cmidrule(lr){4-7}
		&  & &
		spectral radius  & iterations & spectral radius & iterations \\
		\midrule
		\multirow{4}{*}{$0.6$} & \multirow{4}{*}{608.498}
		& $10^{3}$  & 3.649468   & --  & 0.000016 & 3 \\
		& & $10^{1}$  & 356.075705   & --   & 0.06604  & 9 \\
		& & $10^{-1}$ & 10378.859402  & --  & 0.825102 & 105 \\
		& & $10^{-3}$ & 14444.754646  & --  & 0.869572 & 142.3 \\
		\midrule
		\multirow{4}{*}{$0.3$} & \multirow{4}{*}{1216.959}
		& $10^{3}$  & 3.649809     & --   & 0.000026 & 3 \\
		& & $10^{1}$  & 356.109000    & --  & 0.079057 & 10 \\
		& & $10^{-1}$ & 10379.832907   & -- & 0.839383 & 118.7 \\
		& & $10^{-3}$ & 14446.635327  & --  & 0.880589 & 160.9 \\
		\midrule
		\multirow{4}{*}{$0.15$} & \multirow{4}{*}{2433.93}
		& $10^{3}$  & 3.649941     & --   & 0.000048 & 3 \\
		& & $10^{1}$  & 356.121883   & --   & 0.092134 & 10 \\
		& & $10^{-1}$ & 10380.209183  & --  & 0.849389 & 130.3 \\
		& & $10^{-3}$ & 14446.635327  & --  & 0.888243 & 176.8 \\
		\bottomrule
		\end{tabular}
		\caption{Spectral radius and augmented iteration count for varying mesh size
			$\Delta_\mathrm{s}$ and $\rho^\mathrm{s}t^\mathrm{s}$, at $\Delta t=10^{-3}$,
			$b=2.5\times10^5$,
			$L=60$ and $t^\mathrm{s} = 0.1$. }
		\label{tab:iteration3}
	\end{table} 
	
	Based on the performance of the augmented partitioned scheme applied to the
	linear FSI problem, a general partitioned FSI problem can be
	stabilized through the construction of an interface operator $\tilde{C}$
	following the principle established by Lemma~\ref{lem:eigsC}, wherein the
	dominant effect of the structural response is accounted for within the fluid
	solve.

    The analysis of the model problem establishes three results that carry forward to the general FSI system. 
    \begin{itemize}
        \item The added-mass instability of the Dirichlet-Neumann scheme arises because the fluid subproblem inverts the fluid interface pressure operator $S_\Sigma$ alone, while the structural response $C$ is lagged to the previous iterate \eqref{eq:itfpresPart}.
        \item Augmenting the interface pressure block with an operator $\tilde{C}$ that is balanced on the right-hand side \eqref{eq:augmented} leaves the fixed point unchanged, and enables convergence to the monolithic solution, while the lagged contribution to the iteration becomes $\tilde{C} - C$ \eqref{eq:itfPresaug}.
        \item the convergence of the augmented scheme is characterized through the spectral radius of the error iteration matrix, which for $\tilde{C} = \gamma M_\Sigma$ yields the explicit condition \eqref{eq:gamma_req} and the practical choice $\gamma = \zeta^{-1}$.
    \end{itemize}  In the model problem, the error iteration acts on the interface pressure alone \eqref{eq:fixedPiteration}, whereas in the the general linearized FSI system, the interface velocity is an additional unknown, and the errors in the interface velocity and pressure are coupled. The augmentation and the fixed-point analysis are retained in Sec.~\ref{sec:MWIS}, where they are used to establish the conditions governing the design of $\tilde{C}$ for the general system.
    
	
	\section{Momentum-weighted interface stabilization for general FSI systems}\label{sec:MWIS}
    
	In this section, the approach developed in Sec. \ref{sec:simplifiedFSI} is
	extended to a general nonlinear FSI problem partitioned across the
	fluid-structure interface $\Gamma^\mathrm{fs}$. The methodology is built around
	the discrete linear system obtained by linearizing the governing nonlinear
	fluid-structure equations, followed by consistent spatio-temporal
	discretization, wherein the structural equations are discretized in the
	Lagrangian frame, the fluid equations in an Arbitrary Lagrangian-Eulerian frame,
	and the resulting geometric nonlinearity is treated explicitly
	\citep{LIU2014687}. Further details regarding the linearization and
	discretization are presented in Appendix \ref{app1}. 
	
	The monolithic solution to
	the linearized FSI problem is considered, after which the standard
	Dirichlet-Neumann partitioned scheme is discussed and the source of added-mass
	instability is identified. The conditions necessary for the construction of  a
    consistent momentum-weighted linear operator acting at the fluid-structure
    interface are presented. Based on the analysis of Sec. \ref{sec:simplifiedFSI},
    the conditions are derived in the inviscid, pressure-driven limit. Subsequently, 
    the approximations for a practical low-memory footprint operator are presented. 

	\subsection{Monolithic solution to the linearized system}
	\label{sec:monolithic}

	Consider the general discrete, linearized monolithic FSI problem, which is expressed as
	\begin{equation}
		\mathcal{A}\delta\boldsymbol{u} =
		\left[\begin{array}{c c c | c}
			A_\mathrm{S}         & A_{\Sigma\mathrm{S}}^T & 0                    & 0 \\[2pt]
			A_{\Sigma\mathrm{S}} & A_\Sigma               & A_{\Sigma\mathrm{F}} & B_\Sigma^T \\[2pt]
			0                    & A_{\mathrm{F}\Sigma}   & A_\mathrm{F}         & B_\mathrm{F}^T \\[2pt]
			\hline
			\noalign{\vspace{2pt}}
			0                    & B_\Sigma             & B_\mathrm{F}      & 0
		\end{array}\right]
		\left[\begin{array}{c}
			\delta\boldsymbol{u}_\mathrm{S} \\[2pt]
			\delta\boldsymbol{u}_\Sigma \\[2pt]
			\delta\boldsymbol{u}_\mathrm{F} \\[2pt]
			\hline
			\noalign{\vspace{2pt}}
			\delta p
		\end{array}\right]
		=
		\left[\begin{array}{c}
			R_\mathrm{S} \\[2pt]
			R_{\mathrm{m,\Sigma}} \\[2pt]
			R_{\mathrm{m,F}} \\[2pt]
			\hline
			\noalign{\vspace{2pt}}
			R_\mathrm{c}
		\end{array}\right] = \mathcal{R},
		\label{eq:discBlockMonolithic}
	\end{equation}
	where $\delta\boldsymbol{u}_\mathrm{S}$, $\delta\boldsymbol{u}_\Sigma$ and
	$\delta\boldsymbol{u}_\mathrm{F}$ denote the increments of the velocity
	DOFs corresponding to the interior of the structure, the
	interface, and the interior of the fluid, respectively, while $\delta p$ is
	the increment of the pressure field within the fluid. The right-hand
	side vectors $R_\mathrm{S}$ and $R_{\mathrm{m,F}}$
	represent the residuals of the structural and fluid momentum equations
	respectively, $R_{\mathrm{m,\Sigma}}$ denotes the dynamic imbalance at the
	interface, and $R_\mathrm{c}$ denotes the residual of the continuity
	equation. The tangent matrices for the structure, interface, and fluid are
	denoted by $A_\mathrm{S}$, $A_\Sigma$, and $A_\mathrm{F}$, respectively. Since the 
	interface DOFs are shared between the fluid and the structure, the interface 
	block carries contributions from both the momentum equations, i.e.,
	\begin{equation}
		A_\Sigma = A_\Sigma^\mathrm{f} + A_\Sigma^\mathrm{s},
		\label{eq:itfMonolithic}
	\end{equation}
	where $A_\Sigma^\mathrm{f}$ and $A_\Sigma^\mathrm{s}$ denote the tangent momentum operators of the fluid
and structural interface, respectively.
	
	The off-diagonal blocks in \eqref{eq:discBlockMonolithic} represent the
	coupling between the velocity DOFs and the velocity-pressure coupling.
$A_{\Sigma\mathrm{S}}$ couples interfacial DOFs with
 interior structural DOFs. Similarly, the matrices $A_{\Sigma\mathrm{F}}$ and
	$B_\Sigma^T$ map the force exerted by the fluid onto the fluid-structure
	interface. The matrix $A_{\mathrm{F}\Sigma}$ accounts for the dynamic effects of the
	moving interface on the fluid motion. The discrete divergence operators $B_\Sigma$ and
	$B_\mathrm{F}$ act in conjunction to enforce the incompressibility constraint on the 
	interface and interior fluid velocities, respectively.
	
	The solution vector $\delta\boldsymbol{u}$ and the right-hand side residual vector 
	$\mathcal{R}$ are partitioned around the pressure DOFs, i.e.,
	\begin{equation}
		\delta\boldsymbol{u} = \begin{bmatrix}
			\delta{\boldsymbol{y}}\\[2pt]
			\delta p
		\end{bmatrix}
		= \left[\begin{array}{c}
			\delta\boldsymbol{u}_\mathrm{S}\\[2pt]
			\delta\boldsymbol{u}_\Sigma\\[2pt]
			\delta\boldsymbol{u}_\mathrm{F}\\[2pt]
			\hline
			\noalign{\vspace{2pt}}
			\delta p
		\end{array}\right], \qquad
		\mathcal{R} = \begin{bmatrix}
			R_\mathrm{m}\\[2pt]
			R_\mathrm{c}
		\end{bmatrix}
		= \left[\begin{array}{c}
			R_\mathrm{S}\\[2pt]
			R_{\mathrm{m,\Sigma}}\\[2pt]
			R_{\mathrm{m,F}}\\[2pt]
			\hline
			\noalign{\vspace{2pt}}
			R_\mathrm{c}
		\end{array}\right],
		\label{eq:partitionVectors}
	\end{equation}
	where $\delta{\boldsymbol{y}}$ collects the velocity increments of all
	three node sets and ${{R}}_\mathrm{m}$ is the corresponding collection of
	momentum residuals. The discrete system \eqref{eq:discBlockMonolithic} then 
	takes the form
	\begin{equation}
		\mathcal{A}\delta\boldsymbol{u} = \begin{bmatrix}
			\mathcal{A}_\mathrm{m} & \mathcal{B}^T\\[2pt]
			\mathcal{B} & 0
		\end{bmatrix}\begin{bmatrix}
			\delta{\boldsymbol{y}}\\[2pt]
			\delta p
		\end{bmatrix} = \begin{bmatrix}
			R_\mathrm{m}\\[2pt]
			R_\mathrm{c}
		\end{bmatrix},
		\label{eq:discBlock2}
	\end{equation}
	where $\mathcal{A}_\mathrm{m}$ is the monolithic momentum block and 
	$\mathcal{B} = [0 \; B_\Sigma \; B_\mathrm{F}]$.
	
	The block partitioning in \eqref{eq:discBlock2} lends itself naturally to
	pressure projection schemes, in which the pressure is solved first and the
	velocities are subsequently recovered. Eliminating the velocity increments 
	through a block $LU$ decomposition, as in Sec.~\ref{sec:simplifiedFSI}, yields the
	solution
	\begin{subequations}
		\begin{align}
			\delta p &= -\mathcal{S}^{-1}\left(\mathcal{B}\mathcal{A}_\mathrm{m}^{-1}{{R}}_\mathrm{m} - R_\mathrm{c}\right),\label{eq:solutionMonolithic_u}\\
			\delta{\boldsymbol{y}} &= \mathcal{A}_\mathrm{m}^{-1}\left({{R}}_\mathrm{m}- \mathcal{B}^T\delta p\right),\label{eq:solutionMonolithic_p}
		\end{align}\label{eq:solutionMonolithic}
	\end{subequations}
	where
	\begin{equation}
		\mathcal{S} = -\mathcal{B}\mathcal{A}_\mathrm{m}^{-1}\mathcal{B}^T
		\label{eq:pressureSchur}
	\end{equation}
	is the pressure Schur complement.
	
	The monolithic solution \eqref{eq:solutionMonolithic_u} and \eqref{eq:solutionMonolithic_p} requires the inverse of
	the velocity block $\mathcal{A}_\mathrm{m}$, whose structure is exposed through the 
	$LDU$ factorization
	\begin{equation}
		\mathcal{A}_\mathrm{m} = \begin{bmatrix}
				I & 0 & 0\\
				A_{\Sigma\mathrm{S}}A_\mathrm{S}^{-1} & I & 0\\
				0 & A_{\mathrm{F}\Sigma}S_1^{-1} & I
		\end{bmatrix}\begin{bmatrix}
				A_\mathrm{S} & 0 & 0\\
				0 & S_1 & 0\\
				0 & 0 & S_2
		\end{bmatrix}\begin{bmatrix}
				I & A_\mathrm{S}^{-1}A_{\Sigma\mathrm{S}}^T & 0\\
				0 & I & S_1^{-1}A_{\Sigma\mathrm{F}}\\
				0 & 0 & I
		\end{bmatrix},
		\label{eq:LDU}
	\end{equation}
	where
	\begin{subequations}
		\begin{align}
			S_1 &= A_\Sigma - A_{\Sigma\mathrm{S}}A_\mathrm{S}^{-1}A_{\Sigma\mathrm{S}}^T,\label{eq:schurmonolithic_1}\\
			S_2 &= A_\mathrm{F} - A_{\mathrm{F}\Sigma}S_1^{-1}A_{\Sigma\mathrm{F}}\label{eq:schurmonolithic_2},
		\end{align}
		\label{eq:schurmonolithic}
	\end{subequations}
	are Schur complements which capture the structural interaction with the 
	interface, and the interaction of the interface with the fluid interior, 
	respectively. 
	
	Static condensation of the velocity increments out of \eqref{eq:discBlock2},
	with $\mathcal{A}_\mathrm{m}^{-1}$ expressed through the factorization
	\eqref{eq:LDU}, gives the pressure Schur complement $\mathcal{S}$ in terms of
	the individual blocks (for details, refer to Appendix~\ref{app2}). In the inviscid,
	pressure-driven limit, the  couplings $A_{\Sigma\mathrm{F}}$ and
	$A_{\mathrm{F}\Sigma}$ vanish, reducing $S_2$ to $A_\mathrm{F}$ in
	\eqref{eq:schurmonolithic}. The velocity block $\mathcal{A}_\mathrm{m}$ therefore
	decouples the interface and interior fluid velocities, which makes the structure
	of the added-mass coupling apparent in the pressure Schur complement,
	\begin{equation}
		\mathcal{S} = -\left(\mathcal{B}_\Sigma S_1^{-1}\mathcal{B}_\Sigma^T
		+ \mathcal{B}_\mathrm{F} A_\mathrm{F}^{-1}\mathcal{B}_\mathrm{F}^T\right).
		\label{eq:Schurinviscid}
	\end{equation}
	which leads to the monolithic pressure field increment $\delta p$ to
	satisfy
	\begin{equation}
		\left(B_\Sigma S_1^{-1}B_\Sigma^T +
		B_\mathrm{F}A_\mathrm{F}^{-1}B_\mathrm{F}^T\right)\delta p =
		-B_\Sigma S_1^{-1}A_{\Sigma S}A_\mathrm{S}^{-1}{R}_\mathrm{S} +
		B_\Sigma S_1^{-1}{R}_\mathrm{m,\Sigma} +
		B_\mathrm{F}A_\mathrm{F}^{-1}{R}_\mathrm{m,F}. 
		\label{eq:monolithicPressure}
	\end{equation}
    The left-hand side operator combines the pressure Schur complement
	$B_\mathrm{F}A_\mathrm{F}^{-1}B_\mathrm{F}^T$ corresponding to
	the bulk fluid interior with the interfacial contribution $B_\Sigma
	S_1^{-1}B_\Sigma^T$, where the structural inertia contributes
	to the pressure through $S_1$, see \eqref{eq:itfMonolithic} and
	\eqref{eq:schurmonolithic}. This construction is the general 
	counterpart to the inversion of $S_\Sigma + C$ presented in
	Sec.~\ref{sec:simplifiedFSI}. As highlighted in Sec.~\ref{sec:Introduction},
	while a monolithic formulation can be used to construct robust solutions to FSI
	problems \citep{LIU2014687}, incorporating existing fluid 
	and structural solvers remains challenging. In the next section, the
	augmented Dirichlet-Neumann partitioned scheme is developed for the general
	construction presented above, which recovers robust features offered by
	monolithic solution methods.

	\subsection{Partitioned solve and the interface pressure Schur complement}
	\label{sec:pressureSchur}
	
	The monolithic system \eqref{eq:discBlockMonolithic} is now partitioned across
	the fluid-structure interface, $\Gamma^\mathrm{fs}$, with the fluid and the structure being solved in
	a staggered manner. The index $k$ represents the coupling iterations performed
	within each time step. The structure is solved first, using the fluid state from
	the $k^\mathrm{th}$ iterate, i.e.,
	\begin{equation}
		\begin{bmatrix}
			A_\mathrm{S} & A_{\Sigma\mathrm{S}}^T\\[2pt]
			A_{\Sigma\mathrm{S}} & A_\Sigma^\mathrm{s}
		\end{bmatrix} \begin{bmatrix}
			\delta\boldsymbol{u}_\mathrm{S}^{k+1}\\[2pt]
			\delta\boldsymbol{u}_\Sigma^{k+1}
		\end{bmatrix} = \begin{bmatrix}
			{R}_\mathrm{S}\\[2pt]
			{R}_{\mathrm{m,\Sigma}} - A_{\Sigma\mathrm{F}}\delta\boldsymbol{u}_\mathrm{F}^k -
			B_\Sigma^T \delta p^k - A_\Sigma^\mathrm{f}\delta\boldsymbol{u}_\Sigma^k
		\end{bmatrix},
		\label{eq:solidSolve}
	\end{equation}
	from which the interface velocity increment follows as
	\begin{equation}
		\delta\boldsymbol{u}_\Sigma^{k+1} =
		-S_\mathrm{S}^{-1}A_{\Sigma\mathrm{S}}A_\mathrm{S}^{-1}{R}_\mathrm{S} +
		S_\mathrm{S}^{-1}\left({R}_{\mathrm{m,\Sigma}}  -
		A_{\Sigma\mathrm{F}}\delta\boldsymbol{u}_\mathrm{F}^k - B_\Sigma^T \delta p^k -
		A_\Sigma^\mathrm{f}\delta\boldsymbol{u}_\Sigma^k\right), 
		\label{eq:itfVel}
	\end{equation}
	where
	\begin{equation}
		S_\mathrm{S} = A_\Sigma^\mathrm{s} - A_{\Sigma\mathrm{S}}A_\mathrm{S}^{-1}A_{\Sigma\mathrm{S}}^T
		\label{eq:structSchur}
	\end{equation}
	is the structural interface Schur complement.
	
	The interface velocity increment \eqref{eq:itfVel} is then imposed on the fluid as a
	Dirichlet boundary condition, yielding
	\begin{equation}
		\begin{bmatrix}
			A_\mathrm{F} & B_\mathrm{F}^T\\[2pt]
			B_\mathrm{F} & 0
		\end{bmatrix}\begin{bmatrix}
			\delta\boldsymbol{u}_\mathrm{F}^{k+1}\\[2pt] \delta p^{k+1}
		\end{bmatrix} = \begin{bmatrix}
			{R}_{\mathrm{m,F}} - A_{\mathrm{F}\Sigma}\delta\boldsymbol{u}_\Sigma^{k+1}\\[2pt]
			{R}_\mathrm{c} - B_\Sigma\delta\boldsymbol{u}_\Sigma^{k+1}
		\end{bmatrix}. 
		\label{eq:fluidSolve}
	\end{equation}
	
	Eliminating $\delta\boldsymbol{u}_\mathrm{F}^{k+1}$ and substituting
	\eqref{eq:itfVel}, the pressure update in the inviscid limit
	$A_{\mathrm{F}\Sigma} = A_{\Sigma\mathrm{F}} = 0$ is
	\begin{equation}
		\begin{split}	
			B_\mathrm{F}A_\mathrm{F}^{-1}B_\mathrm{F}^T\delta p^{k+1} =
			& ~ -B_\Sigma S_\mathrm{S}^{-1}B_\Sigma^T\delta p^k -
			B_\Sigma S_\mathrm{S}^{-1}A_\Sigma^\mathrm{f}\delta\boldsymbol{u}_\Sigma^k \\
			&~+ B_\mathrm{F}A_\mathrm{F}^{-1}{R}_{\mathrm{m,F}} -
			B_\Sigma S_\mathrm{S}^{-1}A_{\Sigma\mathrm{S}}A_\mathrm{S}^{-1}{R}_\mathrm{S} +
			B_\Sigma S_\mathrm{S}^{-1}{R}_{\mathrm{m,\Sigma}}  - {R}_\mathrm{c},
		\end{split}
		\label{eq:ppeUnaug}
	\end{equation} 
	with the intermediate steps presented in Appendix \ref{app2}. 
	At the fixed point, the operator assembled by
    \eqref{eq:ppeUnaug} is
	\begin{equation}
		B_\Sigma S_\mathrm{S}^{-1}B_\Sigma^T +
		B_\mathrm{F}A_\mathrm{F}^{-1}B_\mathrm{F}^T,
	\end{equation}
	which differs from the monolithic Schur complement \eqref{eq:Schurinviscid} only 
	in $S_1$ being replaced by $S_\mathrm{S}$. From equations
	\eqref{eq:itfMonolithic}, \eqref{eq:schurmonolithic_1} and
	\eqref{eq:structSchur},\begin{equation}
		S_1 = S_\mathrm{S} + A_\Sigma^\mathrm{f}.
		\label{eq:S1split}
	\end{equation} 
	
	The implication here is the fluid interfacial contribution
	$A_\Sigma^\mathrm{f}$ carried by the monolithic system is omitted within the
	partitioned scheme, and therefore, the iteration driven by $B_\Sigma
	S_\mathrm{S}^{-1}B_\Sigma^T$ in \eqref{eq:ppeUnaug} is the origin of the
	added-mass instability and is conditionally convergent \citep{BADIA20087027, wall2007strong}.
	
	Following the principles outlined in Sec.~\ref{sec:aug_partitioned}, the
	incompressibility constraint is augmented by an interfacial pressure operator
	$\tilde{C}$, and the fluid subproblem \eqref{eq:fluidSolve} becomes
	\begin{equation}
		\begin{bmatrix}
			A_\mathrm{F} & B_\mathrm{F}^T\\[2pt]
			B_\mathrm{F} & -\tilde{C}
		\end{bmatrix}\begin{bmatrix}
			\delta\boldsymbol{u}_\mathrm{F}^{k+1}\\[2pt] \delta p^{k+1}
		\end{bmatrix} = \begin{bmatrix}
			{R}_{\mathrm{m,F}} - A_{\mathrm{F}\Sigma}\delta\boldsymbol{u}_\Sigma^{k+1}\\[2pt]
			{R}_\mathrm{c} - B_\Sigma\delta\boldsymbol{u}_\Sigma^{k+1} - \tilde{C}\delta p^k
		\end{bmatrix},
		\label{eq:fluidSolveAug}
	\end{equation}
	leading to the augmented pressure update
	\begin{equation}
		\begin{split}
			\left(\tilde{C} + B_\mathrm{F}A_\mathrm{F}^{-1}B_\mathrm{F}^T\right)\delta p^{k+1} =
			& ~ \left(\tilde{C} - B_\Sigma S_\mathrm{S}^{-1}B_\Sigma^T\right)\delta p^k -
			B_\Sigma S_\mathrm{S}^{-1}A_\Sigma^\mathrm{f}\delta\boldsymbol{u}_\Sigma^k \\
			&+ B_\mathrm{F}A_\mathrm{F}^{-1}{R}_{\mathrm{m,F}} -
			B_\Sigma S_\mathrm{S}^{-1}A_{\Sigma\mathrm{S}}A_\mathrm{S}^{-1}{R}_\mathrm{S} +
			B_\Sigma S_\mathrm{S}^{-1}{R}_{\mathrm{m,\Sigma}}  - {R}_\mathrm{c}.
		\end{split}
		\label{eq:ppeFSI}
	\end{equation} 
	
	By construction, the effect of the operator $\tilde{C}$ cancels at convergence,
	ensuring the augmentation aids only the convergence characteristics of the
	iterations, and leaves the solution unchanged.

	\subsection{Choice of the interface operator}\label{sec:itfoperator}

	In Sec.~\ref{sec:simplifiedFSI}, within the inertia-dominant regime, the structural 
	response at the interface was observed to be 
	proportional to the mass matrix, reducing the augmentation to a single scaling
	parameter $\gamma$. Such simplifications are generally not available for 
    nonlinear FSI problems. For a robust and cost-effective parititioned method,
    the design of a practically admissible $\tilde{C}$ is guided by two requirements,
    namely that it acts on the pressure DOFs located at the fluid-structure interface
    and maintains a small computational footprint relative to the overall size of the
    problem. To drive the convergence of the partitioned iterates to the monolithic solution, the interface operator is developed under these constraints.

	The errors in the pressure field and interface velocity with respect to the
	monolithic solution given by \eqref{eq:solutionMonolithic} are defined as
	\begin{subequations}
		\begin{align}
			e_p^k &= \delta p^k - \delta p,\label{eq:errorDef1}\\
			\boldsymbol{e}_{u_\Sigma}^k &= \delta\boldsymbol{u}_\Sigma^k -
			\delta\boldsymbol{u}_{\Sigma}\label{eq:errorDef2}.
		\end{align}
	\end{subequations} 
	respectively, and their evolution over the coupling iterations is given by the
	following result.
	
	\begin{prop}
		The errors in the pressure field and the interface velocity of the iterates
		\eqref{eq:itfVel} and \eqref{eq:ppeFSI}, measured relative to the monolithic 	solution, evolve as
		\begin{equation}
			\begin{bmatrix}
				S_\mathrm{S} & 0\\[2pt]
				0 & \tilde{C} + B_\mathrm{F}A_\mathrm{F}^{-1}B_\mathrm{F}^T
			\end{bmatrix}\begin{bmatrix}
				\boldsymbol{e}_{u_\Sigma}\\[2pt] e_p
			\end{bmatrix}^{k+1} = \begin{bmatrix}
				-A_\Sigma^\mathrm{f} & -B_\Sigma^T\\[2pt]
				-B_\Sigma S_\mathrm{S}^{-1}A_\Sigma^\mathrm{f} & \tilde{C} -
				B_\Sigma S_\mathrm{S}^{-1}B_\Sigma^T
			\end{bmatrix}\begin{bmatrix}
				\boldsymbol{e}_{u_\Sigma}\\[2pt] e_p
			\end{bmatrix}^{k}.
			\label{eq:errBlock}
		\end{equation}
		\label{prop:contraction}
	\end{prop}
	\begin{proof}
		See Appendix \ref{app3}.
	\end{proof}
	
	The iterations \eqref{eq:errBlock} converge to the monolithic solution if and
	only if the spectral radius of the corresponding iteration matrix is less than
	unity. Denoting by $(\boldsymbol{v}, q)$ the eigenvector associated with the
	eigenvalue $\lambda$, the generalized eigenvalue problem takes the form
	\begin{equation}
		\begin{bmatrix}
			-A_\Sigma^\mathrm{f} & -B_\Sigma^T\\[2pt]
			-B_\Sigma S_\mathrm{S}^{-1}A_\Sigma^\mathrm{f} &
			\tilde{C} - B_\Sigma S_\mathrm{S}^{-1}B_\Sigma^T
		\end{bmatrix}
		\begin{bmatrix}
			\boldsymbol{v}\vphantom{B_\Sigma^T}\\[2pt]
			q\vphantom{S_\mathrm{S}^{-1}}
		\end{bmatrix}
		= \lambda
		\begin{bmatrix}
			S_\mathrm{S}\vphantom{A_\Sigma^\mathrm{f}} & 0\vphantom{B_\Sigma^T}\\[2pt]
			0\vphantom{S_\mathrm{S}^{-1}} &
			\tilde{C} + B_\mathrm{F}A_\mathrm{F}^{-1}B_\mathrm{F}^T
		\end{bmatrix}
		\begin{bmatrix}
			\boldsymbol{v}\vphantom{B_\Sigma^T}\\[2pt]
			q\vphantom{S_\mathrm{S}^{-1}}
		\end{bmatrix}.
		\label{eq:errEig}
	\end{equation}

	Consider first the mode for which the pressure error vanishes, i.e., $q = 0$.
	This reduces \eqref{eq:errEig} to
	\begin{equation}
		-A_\Sigma^\mathrm{f}\boldsymbol{v} = \lambda S_\mathrm{S} \boldsymbol{v},
		\label{eq:errVelocity}
	\end{equation}
	under the constraint
	$S_\mathrm{S}^{-1}A_\Sigma^\mathrm{f}\boldsymbol{v}\in \ker(B_\Sigma)$. The
	interface velocity error must therefore lie in the null space of the divergence
	constraint restricted to the interface, and the rate of convergence of such
	modes is governed by the spectral radius of
	$S_\mathrm{S}^{-1}A_\Sigma^\mathrm{f}$. The operator $\tilde{C}$ drops out of
	\eqref{eq:errVelocity} entirely and is therefore ineffective in driving their
	convergence.
	
	Alternately, in the case where $q\neq 0$, the elimination of $\boldsymbol{v}$
	leads to the nonlinear GEVP,
	\begin{equation}
		\tilde{C}q = \frac{\lambda}{1-\lambda}\left(B_\mathrm{F}A_\mathrm{F}^{-1}B_\mathrm{F}^T +
		B_\Sigma\left(A_\Sigma^\mathrm{f}+\lambda
		S_\mathrm{S}\right)^{-1}B_\Sigma^T\right)q. \label{eq:gevpLambda}
	\end{equation}

    {The structure of \eqref{eq:gevpLambda} implies two components govern the design of the operator $\tilde{C}$, i.e., the pressure Schur complement corresponding to the interior of the fluid, $B_\mathrm{F}A_\mathrm{F}^{-1}B_\mathrm{F}^T$, and the interfacial contribution $B_\Sigma(A_\Sigma^\mathrm{f} + \lambda S_\mathrm{S})^{-1}B_\Sigma^T$, weighted by the combined fluid and structural momentum operators at the interface. Since the operator $\tilde{C}$ acts only on the interfacial pressure DOFs, the right-hand side condensed onto these DOFs governs the nonzero eigenvalues $\lambda$. Denoting the interfacial pressure mode by $q_\Sigma$, and the condensation of the operator $B_\mathrm{F}A_\mathrm{F}^{-1}B_\mathrm{F}^T$ onto the interfacial pressure DOFs by $\hat{S}_\mathrm{F}$, the eigenvalue problem \eqref{eq:gevpLambda} reduces to 
    \begin{equation}
        \tilde{C}_\Sigma q_\Sigma = \frac{\lambda}{1-\lambda}\left(\hat{S}_{F} +B_{\Sigma\Sigma}\left(A_\Sigma^\mathrm{f} + \lambda S_\mathrm{S}\right)^{-1}B_{\Sigma\Sigma}^T\right)q_\Sigma, \label{eq:tildeItf}
    \end{equation} 
    where $B_{\Sigma\Sigma}$ and $\tilde{C}_\Sigma$ are the divergence and augmentation operators restricted to the interfacial pressure DOFs. Physically, the operator on the RHS of \eqref{eq:tildeItf} measures the divergence induced in the fluid adjacent to the interface due to an interfacial pressure mode. 

     }
	
	Defining $\mu = \lambda/(1-\lambda)$, and applying the Woodbury
	matrix identity on $(A_\Sigma^\mathrm{f} + \lambda S_\mathrm{S})^{-1}$, the interfacial
	GEVP \eqref{eq:tildeItf} can be rewritten as,
    \begin{equation}
		\tilde{C}_\Sigma q_\Sigma = \mu\left(S_\mathrm{FP} - S_\Sigma^\star(\mu)\right)q_\Sigma,
		\label{eq:nlevp}
	\end{equation}
	where
    \begin{subequations}
		\begin{align}
			S_\mathrm{FP} &= B_{\Sigma\Sigma} {A_\Sigma^\mathrm{f}}^{\,-1}B_{\Sigma\Sigma}^T + \hat{S}_\mathrm{F},
            \label{eq:Sfp}\\
			S_\Sigma^\star(\mu) &= B_{\Sigma\Sigma}\left(A_\Sigma^\mathrm{f} +
			\left(1+\frac{1}{\mu}\right)A_\Sigma^\mathrm{f}S_\mathrm{S}^{-1}A_\Sigma^\mathrm{f}\right)^{-1}B_{\Sigma\Sigma}^T.
			\label{eq:Sstar}
		\end{align}
	\end{subequations}

    {The operator $S_\mathrm{FP}$ measures the divergence induced in the fluid adjacent to the interface by an interfacial pressure mode when the interface moves under the fluid momentum alone, through $\hat{S}_\mathrm{F}$ from the fluid interior and through $B_{\Sigma\Sigma}{A_\Sigma^\mathrm{f}}^{-1}B_{\Sigma\Sigma}^T$ from the interface motion. $S_\Sigma^\star$ modifies the divergence induced by the interface motion to account for the lagged structural response, and the combination of these effects govern the design of the operator $\tilde{C}_\Sigma$ (and therein $\tilde{C}$).} Furthermore, the definition of $\mu$ establishes the necessary
	condition for convergence 
	\begin{equation}
		\mathscr{R}(\mu) > -\frac12, \label{eq:muCondition}
	\end{equation} 
	where $\mathscr{R}(\cdot)$ denotes the real part of a complex number.  
    
    { Within \eqref{eq:Sstar}, the structural response enters through the ratio $A_\Sigma^\mathrm{f}S_\mathrm{S}^{-1}$ of the fluid and the structural momentum operators at the interface, scaled by the factor $(1+1/\mu)$. The condition \eqref{eq:muCondition} is equivalent to $|1+1/\mu|>1$, i.e., the ratio is amplified for every convergent mode. Eigenvalues satisfying
    \begin{equation}
        \left|\mu+\frac12\right| < \frac12 \label{eq:osccond}
    \end{equation}
    additionally satisfy $\mathscr{R}(1+1/\mu)<0$, with the implication of an amplified ratio entering $S_\Sigma^\star$ with a negative sign, which numerically translates to rendering the interface motion more compliant to pressure modes, provided the fluid and structural inertia are comparable. Given the pressure error mode $q$, the penalty on the divergence-free constraint induced by the operator $S_\mathrm{FP} - S_\Sigma^\star$ exceeds the penalty arising from the operator $S_\mathrm{FP}$. The reinforcement through the inverted sign of $S_\Sigma^\star$ implies the lag in the structural response opposes the fluid motion, resulting in pressure corrections overshooting to accommodate the additional divergence penalty. Furthermore, since $\mathscr{R}(1+1/\mu) < 0$, the convergence is non-monotone. 

    \begin{figure}
        \centering
        \includegraphics[width=0.7\linewidth]{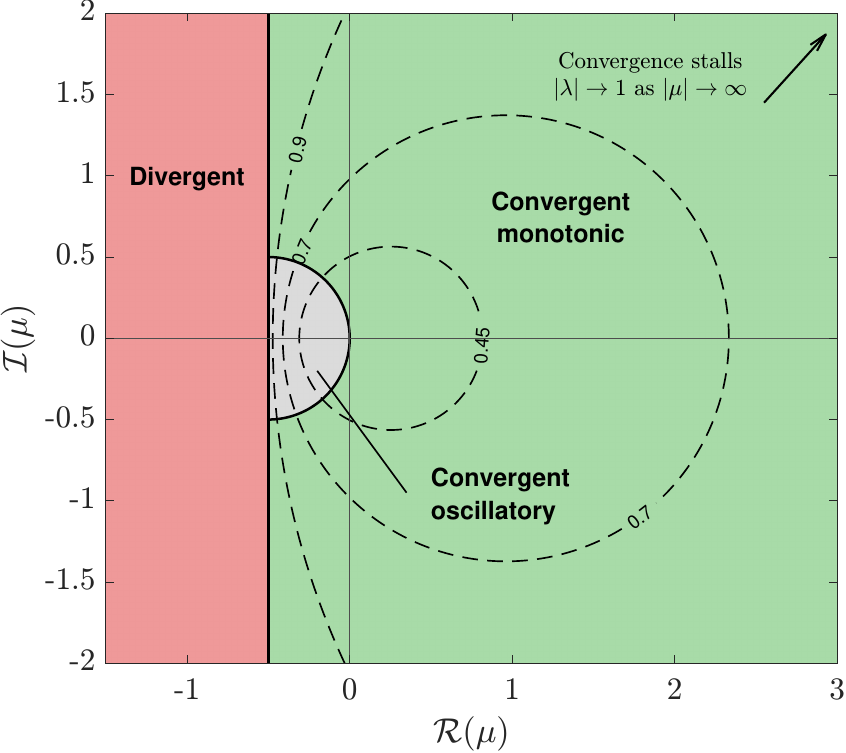}
        \caption{Delineation of regions based on $\mu$, where the augmentation operator can aid convergence. The dashed lines indicate $|\lambda|$, which is a measure of the convergence rate of the augmented iterations. }
        \label{fig:muregions}
    \end{figure}
    
    For $\mathcal{R}(\mu) < -\frac12$, $|1+1/\mu|<1$, and the structural ratio is both attenuated and reversed. In the limit $\mu\to-1$, $S_\Sigma^\star$ cancels the interfacial contribution to $S_\mathrm{FP}$, and the fluid responds to the pressure mode as if the interface were prescribed, recovering the added-mass instability of the Dirichlet-Neumann scheme. Alternately, in the scenario where $\mathcal{R}(\mu) > 0$, $\mathscr{R}(1+1/\mu) > 1$, which implies that the interface is numerically stiffened in the presence of fluid motion. In contrast to the case discussed above, the penalty on the divergence-free constraint is lower relative to that induced by $S_\mathrm{FP}$. The convergence is now rendered monotone, with pressure corrections aiding the convergence direction. Given the two conditions above, the rate of convergence is governed by 
    \begin{equation}
        \rho_c = \max_i \left|\frac{\mu_i}{1+\mu_i}\right|. \label{eq:condSpec}
    \end{equation} In addition to the condition established by \eqref{eq:muCondition}, a second constraint on convergence is governed by the limit as $|\mu|\to\infty$. Under this condition, the operator on the right-hand side approaches the monolithic interface operator, and the convergence of the pressure error slows down. } The interfacial operator $\tilde{C}$ can therefore be designed using the condition established by \eqref{eq:nlevp}, along with the conditions \eqref{eq:muCondition}, \eqref{eq:osccond} and \eqref{eq:condSpec} on the eigenvalues $\mu$. Fig.~\ref{fig:muregions} illustrates the convergence behaviour observed under the action of $\tilde{C}$ based on the location of the eigenvalues $\mu$.

	\subsection{The momentum-weighted interface operator}

	Guided by the role of the fluid and structure momentum operators in
	\eqref{eq:Sfp} and \eqref{eq:Sstar}, the interface operator is defined as
	\begin{equation}
		\tilde{C} = B_\mathrm{\Sigma}\left({A_\mathrm{\Sigma}^\mathrm{f}}^{\,-1} +
		{A_\mathrm{\Sigma}^\mathrm{s}}^{\,-1}\right)B_\mathrm{\Sigma}^T. 
		\label{eq:Ctilde}
	\end{equation}
	 
	The operator takes contributions from the fluid and the structural momentum
	operators at the interface, and the resulting method is hence referred to as
	momentum-weighted interface stabilization (MWIS). 
	
	
	Constructing $\tilde{C}$ as defined in \eqref{eq:Ctilde} would require inverting
	$A_\Sigma^\mathrm{f}$ and $A_\Sigma^\mathrm{s}$, which for large problems can be
	computationally expensive. Therefore, sparse approximations of the two operators
	are necessary to restrict the computational footprint. From the weak form of the 
	governing equations for the fluid and structure, the interfacial matrices 
	$A_\Sigma^\mathrm{f}$ and $A_\Sigma^\mathrm{s}$ can be viewed as
	\begin{equation}
		A_\Sigma^\mathrm{f} = \frac{\rho^\mathrm{f}}{\Delta t} M_\Sigma^\mathrm{f} +
		\rho^\mathrm{f}K_\Sigma^\mathrm{conv} + \mu^\mathrm{f}K_\Sigma^\mathrm{diff}, 
		\label{eq:itfF}
	\end{equation}
	and
	\begin{equation}
		A_\Sigma^\mathrm{s} = \frac{\rho^\mathrm{s}}{\Delta t} M_\Sigma^\mathrm{s} +
		\left((\lambda^\mathrm{s}+2\mu^\mathrm{s})K_\Sigma^\mathrm{vol} +
		\mu^\mathrm{s}K_\Sigma^\mathrm{dev}\right)\Delta t, 
		\label{eq:itfS}
	\end{equation} 
	where $M_\Sigma^\mathrm{f}$ and $M_\Sigma^\mathrm{s}$ are the mass matrices
	coupling the fluid and structural interfacial degrees of freedom, respectively,
	$K_\Sigma^\mathrm{conv}$ is the discrete convection operator and
	$K_\Sigma^\mathrm{diff}$ the discrete diffusion operator at the interface, and
	$K_\Sigma^\mathrm{vol}$ and $K_\Sigma^\mathrm{dev}$ represent the dilatational and
	shear components of the discrete structural stiffness matrix. Here
	$\mu^\mathrm{f}$ is the dynamic viscosity of the fluid, and $\lambda^\mathrm{s}$
	and $\mu^\mathrm{s}$ are the structure's material parameters.

	The interface operators \eqref{eq:itfF} and \eqref{eq:itfS} can be factored as
	\begin{align}
		A_\Sigma^\mathrm{f} &= \frac{\rho^\mathrm{f}}{\Delta t}M_\Sigma^\mathrm{f}
		\left(I + \frac{\Delta t}{\rho^\mathrm{f}}{M_\Sigma^\mathrm{f}}^{\,-1}
		\left(\rho^\mathrm{f}K_\Sigma^\mathrm{conv} + \mu^\mathrm{f}K_\Sigma^\mathrm{diff}\right)\right),
		\label{eq:itfFfac}\\
		A_\Sigma^\mathrm{s} &= \frac{\rho^\mathrm{s}}{\Delta t}M_\Sigma^\mathrm{s}
		\left(I + \frac{\Delta t^2}{\rho^\mathrm{s}}{M_\Sigma^\mathrm{s}}^{\,-1}
		\left((\lambda^\mathrm{s} + 2\mu^\mathrm{s})K_\Sigma^\mathrm{vol} +
		\mu^\mathrm{s}K_\Sigma^\mathrm{dev}\right)\right),
		\label{eq:itfSfac}
	\end{align}
	in which the bracketed terms carry the convective and viscous response of the
	fluid and the elastic response of the structure. These terms are replaced with
	their corresponding diagonal approximations $\tilde{D}_\tau^\mathrm{f}$ and
	$\tilde{D}_\tau^\mathrm{s}$, whose nodal entries are constructed from the
	element metrics as
	\begin{align}
		\tilde{D}_{\tau,k}^\mathrm{f} &= \rho^\mathrm{f}\,\frac{\sum_{e\in\mathcal{T}^\mathrm{f}}
			\int_{\Omega_e}N_k\left(\sqrt{(\boldsymbol{u}^\mathrm{f}-\boldsymbol{u}_m)\cdot
				\boldsymbol{G}(\boldsymbol{u}^\mathrm{f}-\boldsymbol{u}_m)} +
			\nu^\mathrm{f}\sqrt{C_I\boldsymbol{G}:\boldsymbol{G}}\right)d\Omega}
		{\sum_{e\in\mathcal{T}^\mathrm{f}}\int_{\Omega_e}N_k\,d\Omega}, \label{eq:Dtauf}\\
		\tilde{D}_{\tau,k}^\mathrm{s} &= \rho^\mathrm{s}\,\frac{\sum_{e\in\mathcal{T}^\mathrm{s}}
			\int_{\Omega_e}N_k\left(c_l^2\Delta t\sqrt{C_I\boldsymbol{G}:\boldsymbol{G}}\right)d\Omega}
		{\sum_{e\in\mathcal{T}^\mathrm{s}}\int_{\Omega_e}N_k\,d\Omega}, \label{eq:Dtaus}
	\end{align}
	where $\mathcal{T}^\mathrm{f}$ and $\mathcal{T}^\mathrm{s}$ denote the fluid and
	structural triangulations, $N_k$ the basis function corresponding to
	node $k$ in an element, $c_l$ the longitudinal wave speed within the structure, $C_I$ the constant 
	arising from the element-wise inverse estimate, and $\boldsymbol{G}$ is the element
	contravariant metric tensor
	\begin{equation}
		\boldsymbol{G} = \frac{\partial \xi^T}{\partial \boldsymbol{x}}
		\frac{\partial \xi}{\partial \boldsymbol{x}},
		\label{eq:metricTensor}
	\end{equation}
	with $\boldsymbol{x}$ and $\xi$ the physical and parametric coordinates,
	respectively. The resulting approximations to 
	the interface operators are
	\begin{subequations}
		\begin{align}
			\tilde{A}_\Sigma^\mathrm{f} &= \frac{\rho^\mathrm{f}}{\Delta t}
			M_\Sigma^\mathrm{f}\underbrace{\left(I + \frac{\Delta t}{\rho^\mathrm{f}}
				\tilde{D}^\mathrm{f}_\tau\right)}_{\hat{D}_\tau^\mathrm{f}}
			\label{eq:Aif},\\
			\tilde{A}_\Sigma^\mathrm{s} &= \frac{\rho^\mathrm{s}}{\Delta t}
			M_\Sigma^\mathrm{s}\underbrace{\left(I + \frac{\Delta t}{\rho^\mathrm{s}}
				\tilde{D}^\mathrm{s}_\tau\right)}_{\hat{D}_\tau^\mathrm{s}},
			\label{eq:Ais}
		\end{align}
	\end{subequations} and used only within the stabilization process.

	Consequently, the interface operator \eqref{eq:Ctilde} for inertially dominated
	flows is approximated as
	\begin{equation}
		\begin{split}
			\tilde{C} &= B_\Sigma{A_\Sigma^\mathrm{s}}^{\,-1}B_\Sigma^T +
			B_\Sigma{A_\Sigma^\mathrm{f}}^{\,-1}B_\Sigma^T,\\
			&= B_\Sigma{A_\Sigma^\mathrm{s}}^{\,-1}\left(A_\Sigma^\mathrm{f}{A_\Sigma^\mathrm{f}}^{\,-1}\right)B_\Sigma^T
			+ B_\Sigma{A_\Sigma^\mathrm{f}}^{\,-1}B_\Sigma^T,\\
			&\approx B_\Sigma\underbrace{\left(\frac{\rho^\mathrm{f}}{\rho^\mathrm{s}}
				{\hat{D}_\tau^\mathrm{s}{}}^{\,-1}{M_\Sigma^\mathrm{s}}^{\,-1}M_\Sigma^\mathrm{f}
				\hat{D}_\tau^\mathrm{f}\right)}_{D^\mathrm{fs}}{A_\Sigma^\mathrm{f}}^{\,-1}B_\Sigma^T
			+ B_\Sigma{A_\Sigma^\mathrm{f}}^{\,-1}B_\Sigma^T.
		\end{split}
		\label{eq:approxCtilde}
	\end{equation}
	
	Using a lumped mass approximation for the interface, $D^\mathrm{fs}$ can be 
	commutated with $B_\Sigma$ using a least-squares approach \citep{2014_Elman_FE}. 
	Defining the commutation as $B_\Sigma D^\mathrm{fs} \approx
	\tilde{D}^\mathrm{fs}B_\Sigma$, the approximation \eqref{eq:approxCtilde}
	further simplifies to
	\begin{equation}
		\tilde{C} \approx (I + \tilde{D}^\mathrm{fs})
		B_\Sigma {A_\Sigma^\mathrm{f}}^{\,-1} B_\Sigma^T. 
		\label{eq:tildeCfApprox}
	\end{equation}
    With this interfacial operator, the structural contribution enters only through
    the diagonal factor $\tilde{D}^\mathrm{fs}$, which acts only on the  pressure degrees of freedom restricted to the interface. The augmentation therefore does not change the sparsity of the pressure Schur complement.
    	
    \section{Numerical experiments} \label{sec:results}
    The preceding sections establish the theoretical foundation of the proposed momentum-weighted interface stabilization (MWIS). The simplified FSI model in Section 2 provides insight into the stabilization mechanism, while Section 3 extends this principle to general discrete FSI systems and develops a practical interface operator based on fluid and structural momentum contributions. The numerical performance of the resulting formulation is evaluated, with particular emphasis on consistency, convergence characteristics, and robustness under strong added-mass effects. The numerical studies proceed from a model FSI problem, which allows a systematic investigation of the influence of mass ratio and interface mesh resolution, to established FSI benchmarks that assess the accuracy and effectiveness of the proposed method in more representative coupled configurations.

	\subsection{Elastic solid under a gravity-driven fluid column} \label{sec:piston}
	In the current section, the augmented partitioned scheme is evaluated on a model FSI problem. The model problem involves the deformation of a linear elastic solid under the load applied by an incompressible, inviscid fluid column moving under the influence of gravity.
	\subsubsection{Problem definition}
	\begin{figure}[pos=H]
	\centering
	\begin{subfigure}[b]{0.32\textwidth}
		\centering
		\begin{tikzpicture}
			\fill[cyan] (0,0) rectangle (3.0,2.0);

			\draw (0,0) rectangle (3.0,6.0);

			\draw (0.0,2.0) -- (3.0,2.0);
			\node at (2.5,1.0) {$\Omega^\mathrm{s}$};
			\node at (2.5,4.0) {$\Omega^\mathrm{f}$};

			\node at (2.5,2.3) {$\Gamma^\mathrm{fs}$};

			\node[anchor=north] at (0,0)   {$(0,-H^\mathrm{s})$};
			\node[anchor=north] at (3,0)   {$(L,-H^\mathrm{s})$};
			\node[anchor=south] at (3,6) {$(L,H^\mathrm{f})$};
			\node[anchor=south] at (0,6) {$(0,H^\mathrm{f})$};

			\node[anchor=east] at (0,2)   {$(0,0)$};
			\node[anchor=west] at (3,2)   {$(L,0)$};
		\end{tikzpicture}
		\caption{}
		\label{fig:domain_a}
	\end{subfigure}
	\begin{subfigure}[b]{0.32\textwidth}
		\centering
		\includegraphics[width=0.8\textwidth]{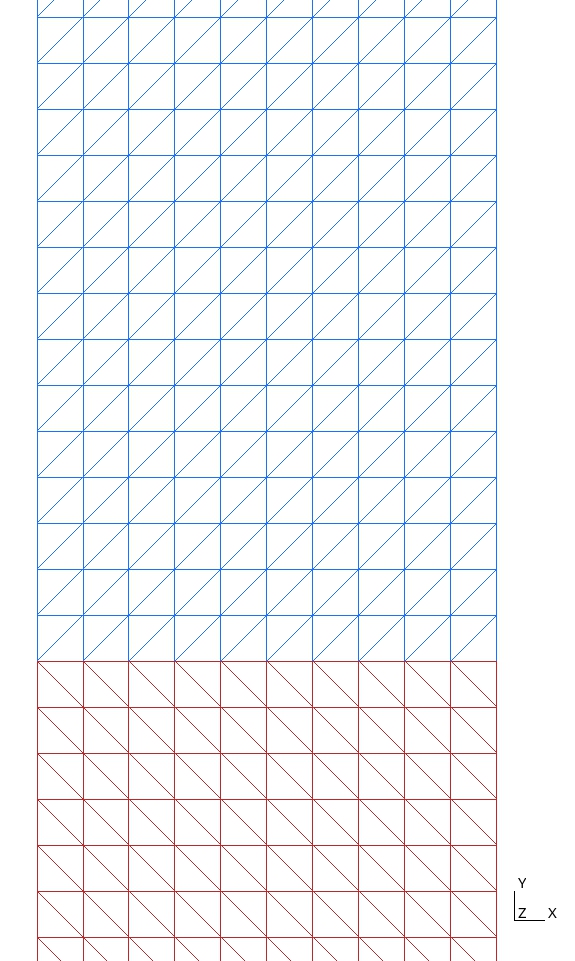}
		\caption{}
		\label{fig:domain_b}
	\end{subfigure}
	\begin{subfigure}[b]{0.32\textwidth}
		\centering
		\includegraphics[width=0.8\textwidth]{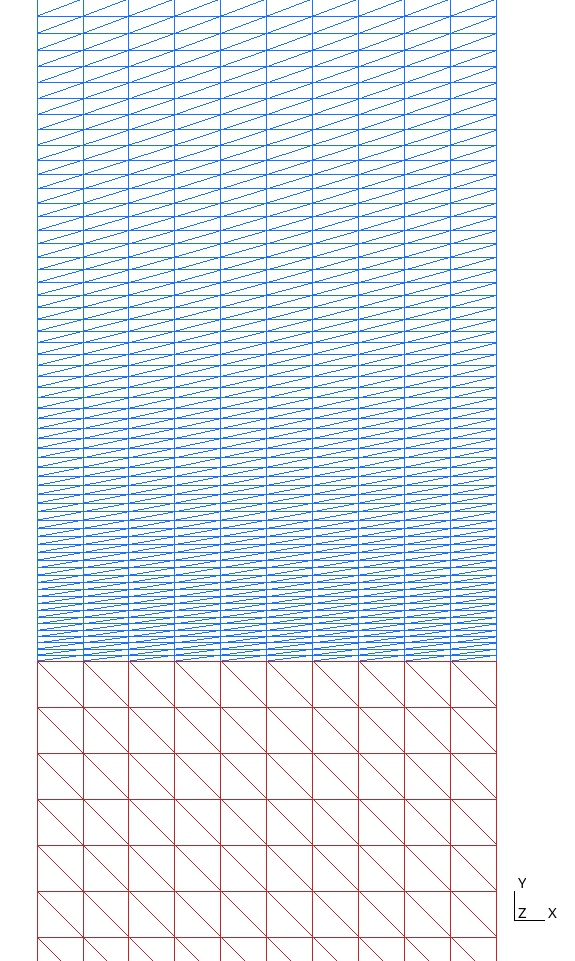}
		\caption{}
		\label{fig:domain_c}
	\end{subfigure}
	\caption{(a) Schematic of fluid column interacting with a compressible linear elastic solid under the effect of gravity and the mesh used for the model problem with (b) $r_\mathrm{fs} = 1.0$, (c) $r_\mathrm{fs} = 0.125$. The meshes are matching at the interface, with the blue elements forming the fluid mesh, and the red elements forming the structural mesh. }
	\label{fig:IncompFluidSolid}
\end{figure}
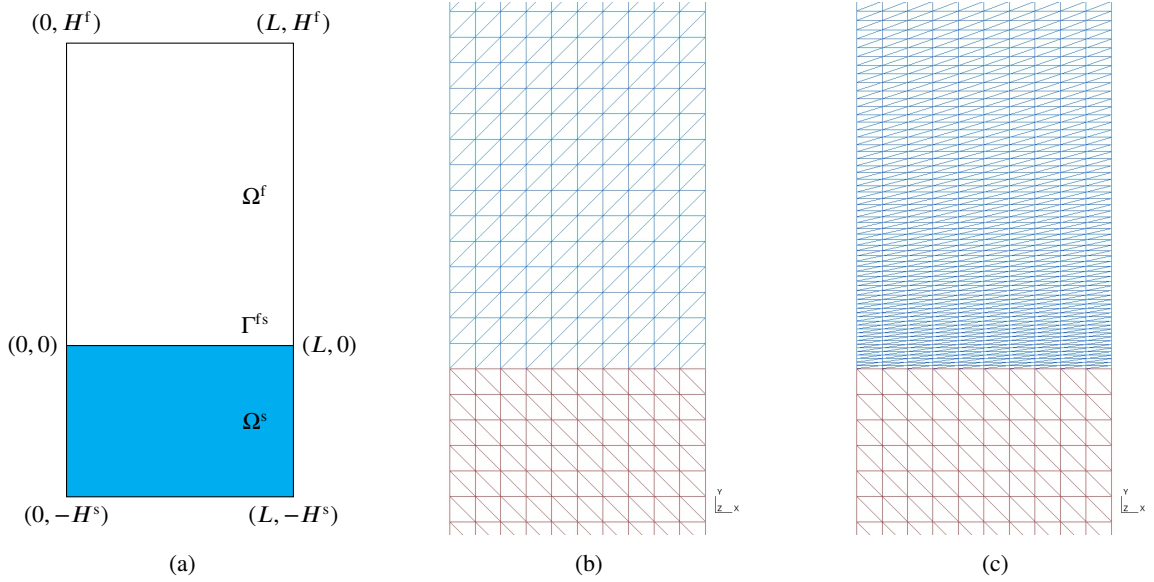
    The computational domain is a two-dimensional rectangular region of length $L = 0.05$~m, with a fluid column of height $H^\mathrm{f} = 0.75$~m, and a solid column of height $H^\mathrm{s} = 0.25$~m. The base of the solid is fixed, with the lateral boundaries located at $x = 0$ and $x = L$ free to deform in the vertical direction. The pressure is set to $p = 0$ on the top boundary of the fluid domain, and a gravitational body force $g = -10~$~m\,s$^{-2}$ is applied in the downward direction. The fluid density is $\rho^\mathrm{f} = 1000$~kg\,m$^{-3}$, while the structure is characterized by a Young's modulus $E^\mathrm{s} = 2.3\times10^5$~Pa and a
    Poisson's ratio $\nu^\mathrm{s} = 0.25$. The structural density is defined in terms
	of the fluid density as,
	\begin{equation}
		\rho^\mathrm{s} = m^\star \rho^\mathrm{f},
		\label{eq:massRatio}
	\end{equation}
	where, $m^\star$ denotes the solid-to-fluid mass ratio.
	
	The structural domain is uniformly discretized using $500$ elements, with the element length normal to the interface fixed at $h_\mathrm{s} = 0.005$~m, and the interface is discretized using $10$ elements. The fluid and structural meshes match at the interface, and therefore the horizontal direction is equivalently resolved. In the direction normal to the interface, the fluid mesh is graded with $h_\mathrm{f}$ being the normal length of the first fluid element adjacent to $\Gamma^\mathrm{fs}$, as shown in Figs.~\ref{fig:domain_b} and \ref{fig:domain_c}. The relative 
	length scales of the two meshes normal to the interface are characterized by the aspect ratio
	\begin{equation}
		r_\mathrm{fs} = \frac{h_\mathrm{f}}{h_\mathrm{s}}.
		\label{eq:meshRatio}
	\end{equation} 
    
    For $r_\mathrm{fs} < 1$, the element size is geometrically grown to $y = 0.25 \text{ m}$, beyond which the mesh size is maintained at $h_\mathrm{f}= 0.005\text{ m}$. For $r_\mathrm{fs}> 1$, a uniform mesh resolution of $h_\mathrm{f}= r_\mathrm{fs}h_\mathrm{s}$ is used for the fluid domain. 
	
	In this configuration, the approximation of the interface operator \eqref{eq:tildeCfApprox} reduces to
	\begin{equation}
		\tilde{D}^\mathrm{fs} =
		\beta_\mathrm{MWIS}I,
	\end{equation} 
	where
	\begin{equation}
		\beta_\mathrm{MWIS} = \frac{\rho^\mathrm{f}h_\mathrm{f}}{\rho^\mathrm{s}h_\mathrm{s}} =
		\frac{r_\mathrm{fs}}{m^\star}.
		\label{eq:betaitf}
	\end{equation}
    
	The parameter $\beta_\mathrm{MWIS}$ represents the relative fluid and structural 
    inertia in the mesh elements straddling the fluid-structure interface and, for a 
    given mass ratio $m^\star$, governs the extent of stabilization applied. 
    Consequently, the following studies are conducted to evaluate the performance of 
    the proposed MWIS scheme.
	\begin{itemize}
		\item \textbf{Spectral assessment}: The mass ratio $m^\star$ and the aspect ratio $r_\mathrm{fs}$ are varied 
        independently, and for each combination, the interfacial eigenvalue problem 
        \eqref{eq:nlevp} is solved directly with $\tilde{C}$ given by 
        \eqref{eq:tildeCfApprox}. The smallest and largest real parts of the eigenvalues 
        $\mu$ are recorded, along with the spectral radius $\rho_c$ of the iteration 
        matrix given by \eqref{eq:condSpec}. This study thereby verifies the analysis of 
        Sec.~\ref{sec:MWIS} against the operators assembled by the solver.
		
		\item \textbf{Iterative performance}: The coupled problem is advanced in
		time at fixed $m^\star$, with $r_\mathrm{fs}$ varied as above. The measured
		quantities are the interface pressure history, the $L_2$ norm of the
		velocity divergence, and the number of coupling iterations per time step.
		The momentum-weighted parameter $\beta_\mathrm{MWIS}$ is compared with the
		interface parameter proposed by \citet{idelsohn_2009},
		\begin{equation}
			\beta_\mathrm{RTIS} = \frac{\rho^\mathrm{f}}{\rho^\mathrm{s}} =
			\frac{1}{m^\star},
			\label{eq:betaRTIS}
		\end{equation}
		which is referred to as relative time-scale interface stabilization (RTIS),
		since it is constructed from the relative time scales of the two domains alone
		and therefore carries no dependence on the interface mesh metrics. The two
		parameters are related by $\beta_\mathrm{MWIS} = r_\mathrm{fs}\,\beta_\mathrm{RTIS}$ and 
		are identical when $r_\mathrm{fs} = 1$, so the comparison isolates the contribution of the
		interface length scales to the stabilization.
	\end{itemize}
	
	The two parameters are varied independently rather than through $\beta_\mathrm{MWIS}$ alone.
	This choice is deliberate, as the same value of $\beta_\mathrm{MWIS}$ can be reached either by
	a heavy structure on a graded fluid mesh or by a light structure on a uniform
	one. Sweeping the two separately therefore distinguishes the contribution of the
	mass ratio, which is the parameter governing the added-mass effect, from that of
	the interface mesh metrics, which enter the stabilization only through the
	interfacial operator \eqref{eq:tildeCfApprox}. The mass ratios considered span three
	orders of magnitude, extending well into the regime where the unstabilized
	partitioned scheme fails, and the aspect ratios span a range
	representative of the interface resolution encountered when a boundary layer
	mesh meets a structural mesh. In all the studies, a time step size of $\Delta t = 6.25 \times 10^{-4} \text{ s}$ is used.
	
	\subsubsection{Spectral assessment}
	\label{sec:monoEigNumerical}
    
	To elucidate the robustness offered by the augmented iterative scheme,
	$\tilde{C}$ is evaluated for the 
	mass ratios $m^\star \in \{10^{-3},10^{-2},
	10^{-1}, 1\}$ on the model problem. For a given mass ratio,
	varying the relative mesh ratio $r_\mathrm{fs}$ impacts the spectral radius of
	the iteration matrix, and the eigenvalue distribution $\mu_i$
	corresponding to the nonlinear eigenvalue problem given by \eqref{eq:nlevp}.
	
	\begin{table}[pos=H]
		\centering
		\begin{tabular}{ c c c c c c c }
			\toprule
			\multirow{2}{*}{$m^\star$} & \multirow{2}{*}{$r_\mathrm{fs}$}
			& \multirow{2}{*}{$\beta_\mathrm{MWIS}$}
			& \multirow{2}{*}{$\min\,\mathscr{R}(\mu)$} &
            \multirow{2}{*}{$\max\,\mathscr{R}(\mu)$} &\multirow{2}{*}{Spectral radius $(\rho_c)$} \\
			& &	& & \\
			\midrule
			\multirow{4}{*}{$10^{-3}$}
			& $1$      & $1000$ &  $-0.91781$ & $13.05$ &$11.1665$ \\
			& $0.25$      & $250$ &  $-0.73272$ & $46.87$ & $2.74145$  \\
			& $0.0625$      & $62.5$ &  $-0.36035$ & $180.2$ & $0.99448$ \\
			& $0.03125$      & $31.25$ &  $-0.20673$ & $370.7$ & $0.99731$ \\
			\hline
			\multirow{4}{*}{$10^{-2}$}
			& $1$      & $100$ & $-0.80426$ & $12.22$ &$4.10888$\\
			& $0.25$      & $25$ &  $-0.47609$ & $45.96$ & $0.97871$ \\
			& $0.0625$      & $6.25$ &  $-0.15739$ & $191.4$ & $0.99480$\\
			& $0.03125$      & $3.125$ &  $-0.08362$ & $371.0$& $0.99731$\\
			\hline
			\multirow{4}{*}{$10^{-1}$}
			& $1$     & $10$ &  $-0.56752$ & $12.40$ & $1.31225$\\
			& $0.25$   & $2.5$ & $-0.22648$ & $46.09$ & $0.97877$\\
			& $0.0625$   & $0.625$ & $-0.06707$ & $190.0$ & $0.99477$\\
			& $0.03125$      & $0.3125$ & $-0.03810$& $359.3$ & $0.99722$\\
			\hline
			\multirow{4}{*}{$10^{0}$}
			& $1$      & $1.0$  & $-0.31236$ & $13.37 $& $0.93043$ \\
			& $0.25$   & $0.25$  & $-0.09793$ & $48.85$ &$0.97994$ \\
			& $0.0625$ & $0.0625$ & $-0.02486$ & $198.0$ & $0.99497$ \\
			& $0.03125$& $0.03125$ & $-0.01985$ & $392.7$ & $0.99746$ \\
			\bottomrule
		\end{tabular}
		\caption{Spectral radius of the interfacial operator $S_\mathrm{S}^{-1}A_\mathrm{if}$ and
			the minimum real part of the eigenvalues $\mu$ of the nonlinear problem
			\eqref{eq:nlevp}, for the interface pressure operator given by
			\eqref{eq:tildeCfApprox}. Convergence of the partitioned iteration requires 
			$\inf \mathscr{R}(\mu) > -\tfrac{1}{2}$.}
		\label{tab:muSweep}
	\end{table} 

    {{Table~\ref{tab:muSweep} illustrates the effectiveness of
    $\tilde{C} = B_\Sigma ({A_\Sigma^\mathrm{f}}^{-1} +
    {A_\Sigma^\mathrm{s}}^{-1})B_\Sigma ^T$ for a particular mass ratio $m^\star$,
    given that the relative mesh ratio $r_\mathrm{fs}$ is adjusted to achieve the
    conditions necessary for convergence.} For a given mass ratio $m^\star$, $\tilde{C}$ shifts the spectrum to the right along the real axis as $r_\mathrm{fs}$ decreases. Furthermore, it can be observed that $\max \mathcal{R}(\mu)$ varies proportionally to $r_\mathrm{fs}$ independently of $m^\star$. On the other end of the spectrum, the migration of $\min \mathcal{R}(\mu)$ is simultaneously governed by $m^\star$ and $r_\mathrm{fs}$. Consequently, different values of the parameter $r_\mathrm{fs}$ render the iterations stable as $m^\star$ is varied. Under the conditions where $\min \mathcal{R}(\mu)$ is adequately shifted to render stable iterations, the maximal eigenvalue begins to govern the spectral radius as defined by \eqref{eq:condSpec}. Thus, with the current choice of $\tilde{C}$, interfacial refinement through $r_\mathrm{fs}$ renders the iterations stable up to a mass ratio of $m^\star = 10^{-3}$, while a degradation is observed in the rate of convergence due to the migration of the maximal eigenvalue. Simultaneously, the control over $r_\mathrm{fs}$ supplements the requirements on mesh sizing for adequate flow resolution.  }

\begin{remark}
	The choice of $\tilde{C}$ presented here is canonical, and demonstrate that
	operators with different spectral characteristics can be designed to meet the
	requirement of stabilizing the fixed point iterations. The dependence of
	stabilization on the parameters $(m^\star, r_\mathrm{fs})$ along with the
	conditions established by \eqref{eq:nlevp} allows for the design of an
	arbitrary $\tilde{C}$.
\end{remark}

\subsubsection{Iterative performance}

{To demonstrate the enhanced accuracy characteristics of the
	augmented iterative scheme, the time evolution of the interface pressure is
	contrasted in Fig. \ref{fig:ItfPresPiston}. Within the model problem, the mass
	ratio $m^\star = 1$, and for the relative mesh ratio $r_\mathrm{fs} = 1$,
	$\beta_\mathrm{MWIS} = \beta_\mathrm{RTIS}$. Fig. \ref{fig:intIdel} illustrates the interface pressure
	when $\beta_\mathrm{RTIS}$ is used as the metric for stabilization. With a reduction in
	$r_\mathrm{fs}$ and refinement of the fluid mesh, variations in the spectral
	content of the interface pressure emerge, prominently in the form of high-
	frequency dispersive components. Within this choice of stabilization, the notion
	of grid-converged results is compromised. Alternately, with $\beta_\mathrm{MWIS}$ as the
	metric for stabilization, the spectral content of the interface pressure is
	preserved with fluid mesh refinement as demonstrated by Fig. \ref{fig:intNew}. 
	Furthermore, the case of $r_\mathrm{fs} > 1$ is unstable when
	$\beta_\mathrm{RTIS}$ is used, whereas the use of $\beta_\mathrm{MWIS}$ maintains mesh independent
	stability even at $r_\mathrm{fs} = 4.0$. 

\begin{figure}[pos=H]
	\centering
	\begin{subfigure}[b]{0.49\textwidth}
		\includegraphics[trim={2cm 1.5cm 0.8cm 0cm}, clip,width=\textwidth]{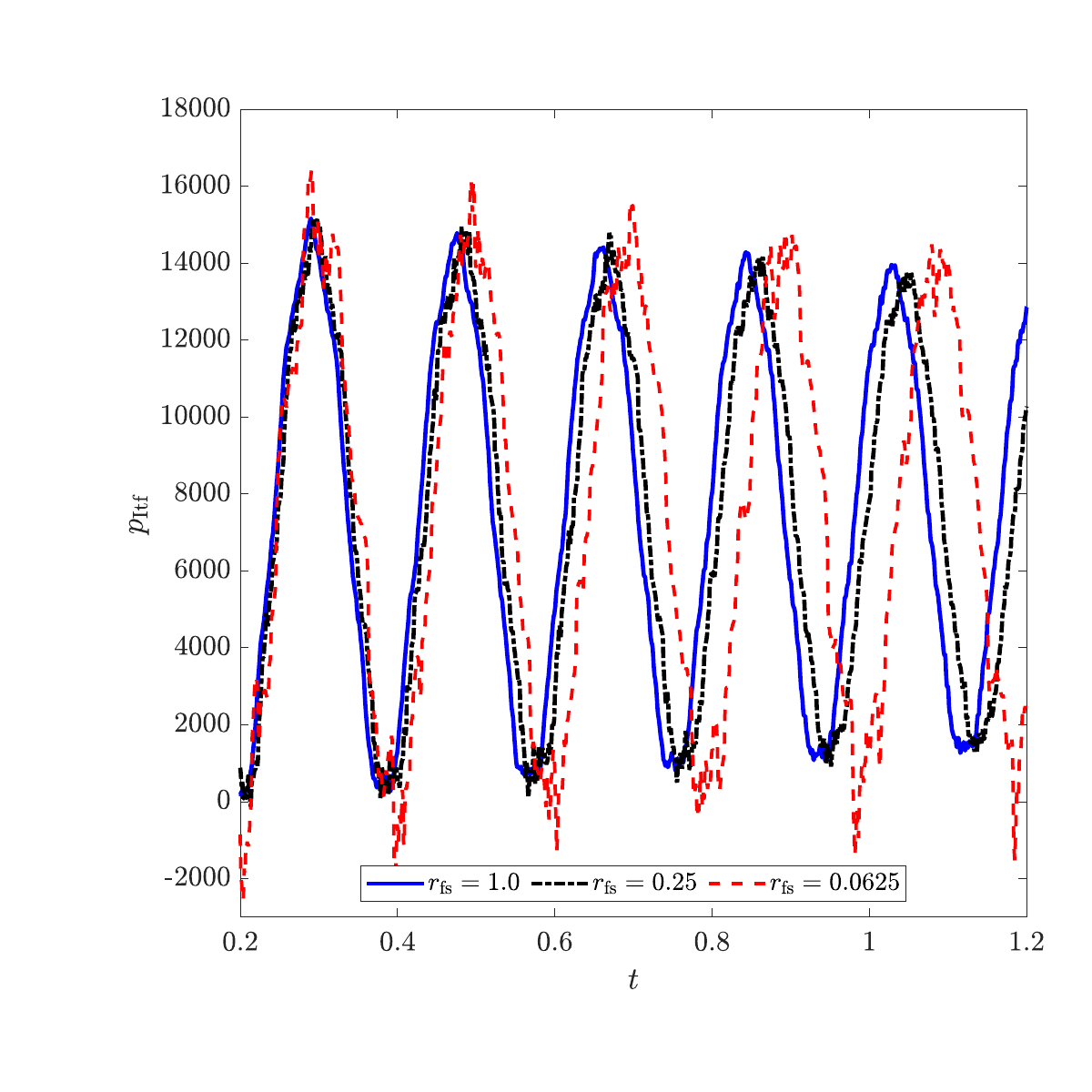}
		\caption{}\label{fig:intIdel}
	\end{subfigure}
	\hfill
	\begin{subfigure}[b]{0.49\textwidth}
		\centering
		\includegraphics[trim={2cm 1.5cm 0.8cm 0cm}, clip,width=\textwidth]{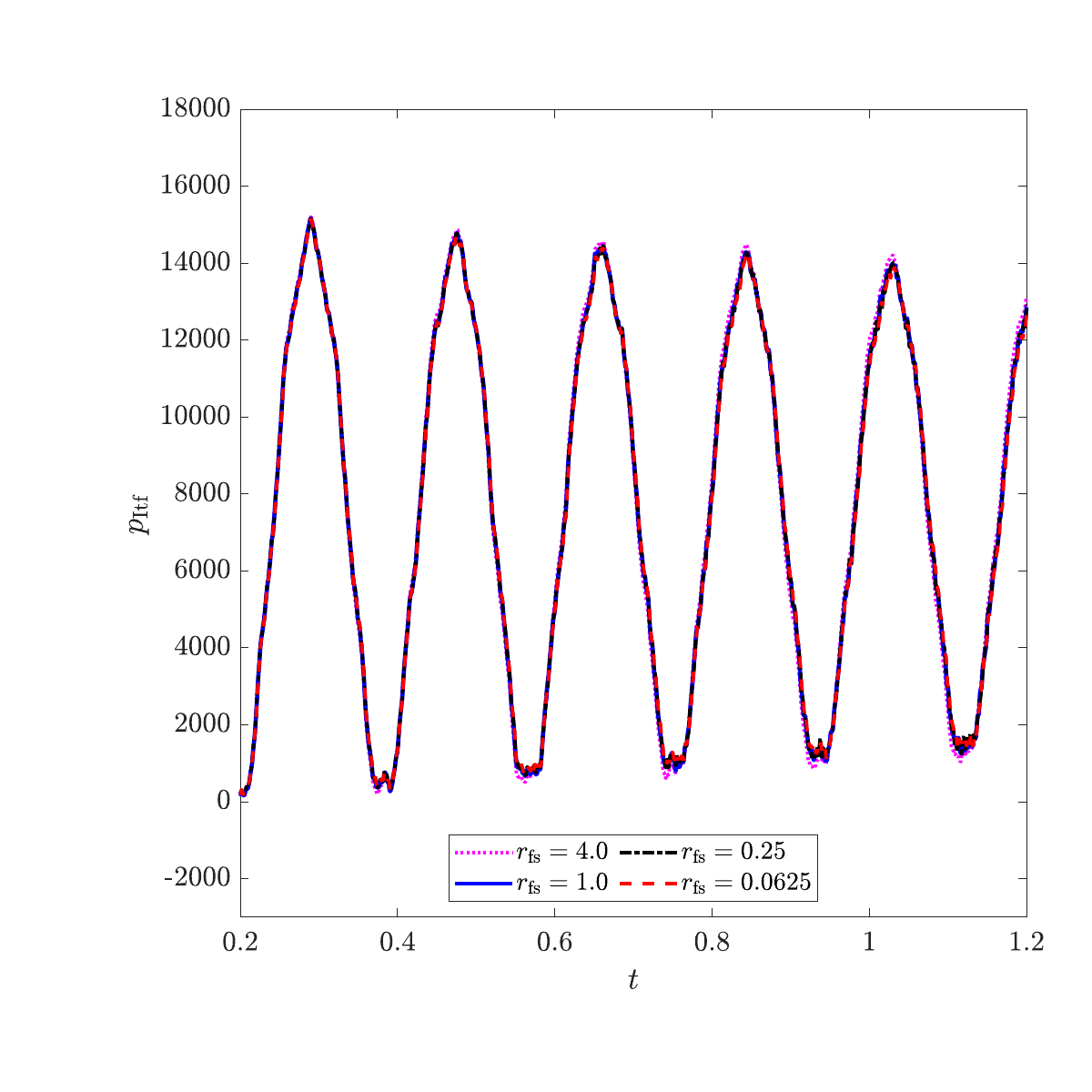}
		\caption{}
		\label{fig:intNew}
	\end{subfigure}
	
	\caption{Temporal variation of the pressure at the fluid-structure interface using (a) $\beta_\mathrm{RTIS}$ and (b) $\beta_\mathrm{MWIS}$ in the approximation of $\tilde{D}^\mathrm{fs}$.}
	\label{fig:ItfPresPiston}
\end{figure}
	
	The enhancement in accuracy and robustness is further reinforced by the
	divergence of the velocity field presented in Fig. \ref{fig:intDivVel}. As the
	relative mesh ratio $r_\mathrm{fs}$ is reduced---a common occurrence in viscous
	boundary layer driven flows---the difficulty in enforcing the divergence-free
	condition on the fluid flow through $\beta_\mathrm{RTIS}$ becomes challenging. Under flow states associated with pressure peaks, the enforcement of the divergence-free
	condition becomes challenging and is evident through the spikes in the iteration
	counts presented in Fig. \ref{fig:intIter}. However, when $\beta_\mathrm{MWIS}$ is used for
	stabilization, the divergence-free constraint is more easily enforced. }

	\begin{figure}[pos=H]
		\centering
		\begin{subfigure}[b]{0.49\textwidth}
			\centering
			\includegraphics[trim={2cm 1.5cm 0.8cm 0cm}, clip,width=\textwidth]{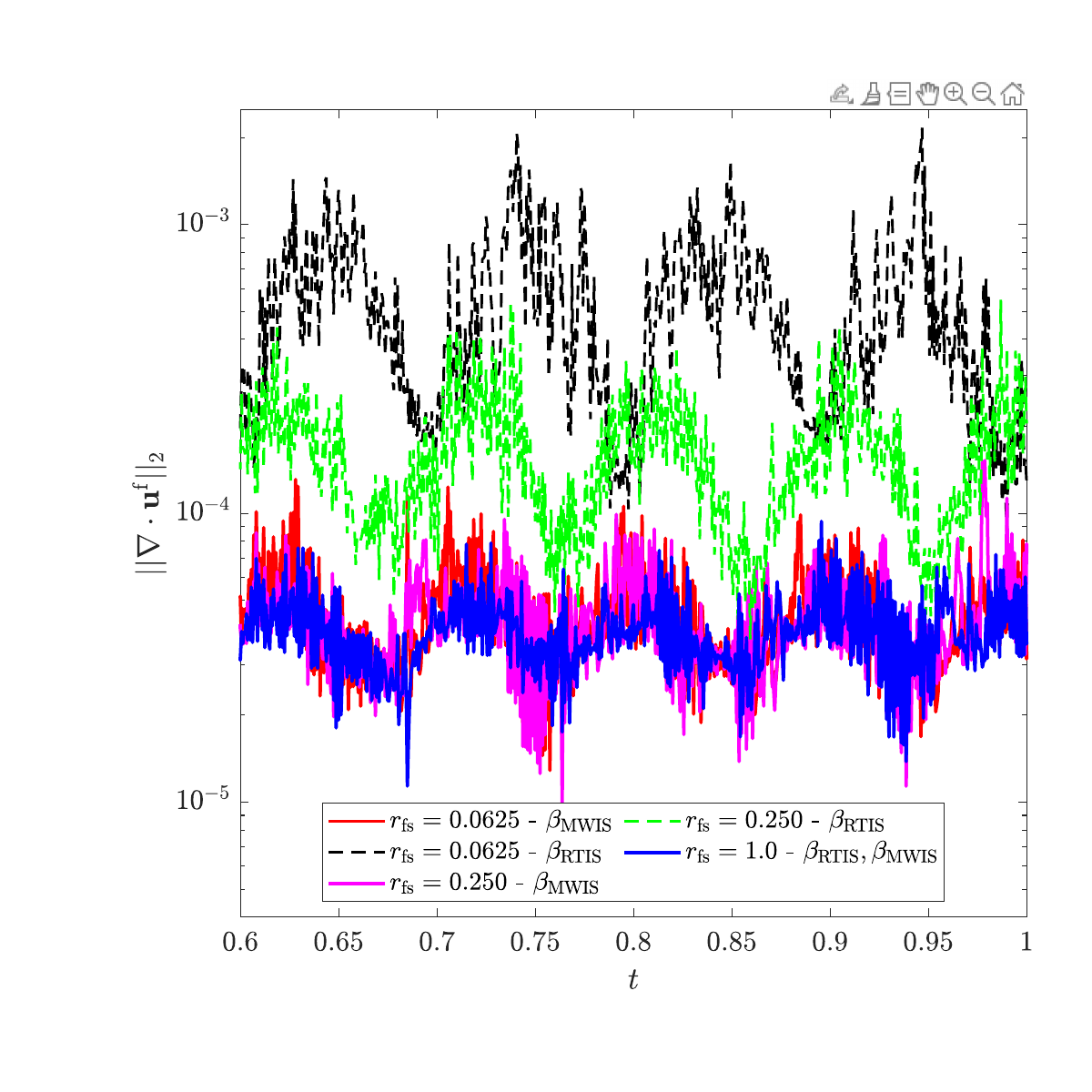}
			\caption{}
			\label{fig:intDivVel}
		\end{subfigure}
		\hfill
		\begin{subfigure}[b]{0.49\textwidth}
			\centering
			\includegraphics[trim={2cm 1.5cm 0.8cm 0cm}, clip,width=\textwidth]{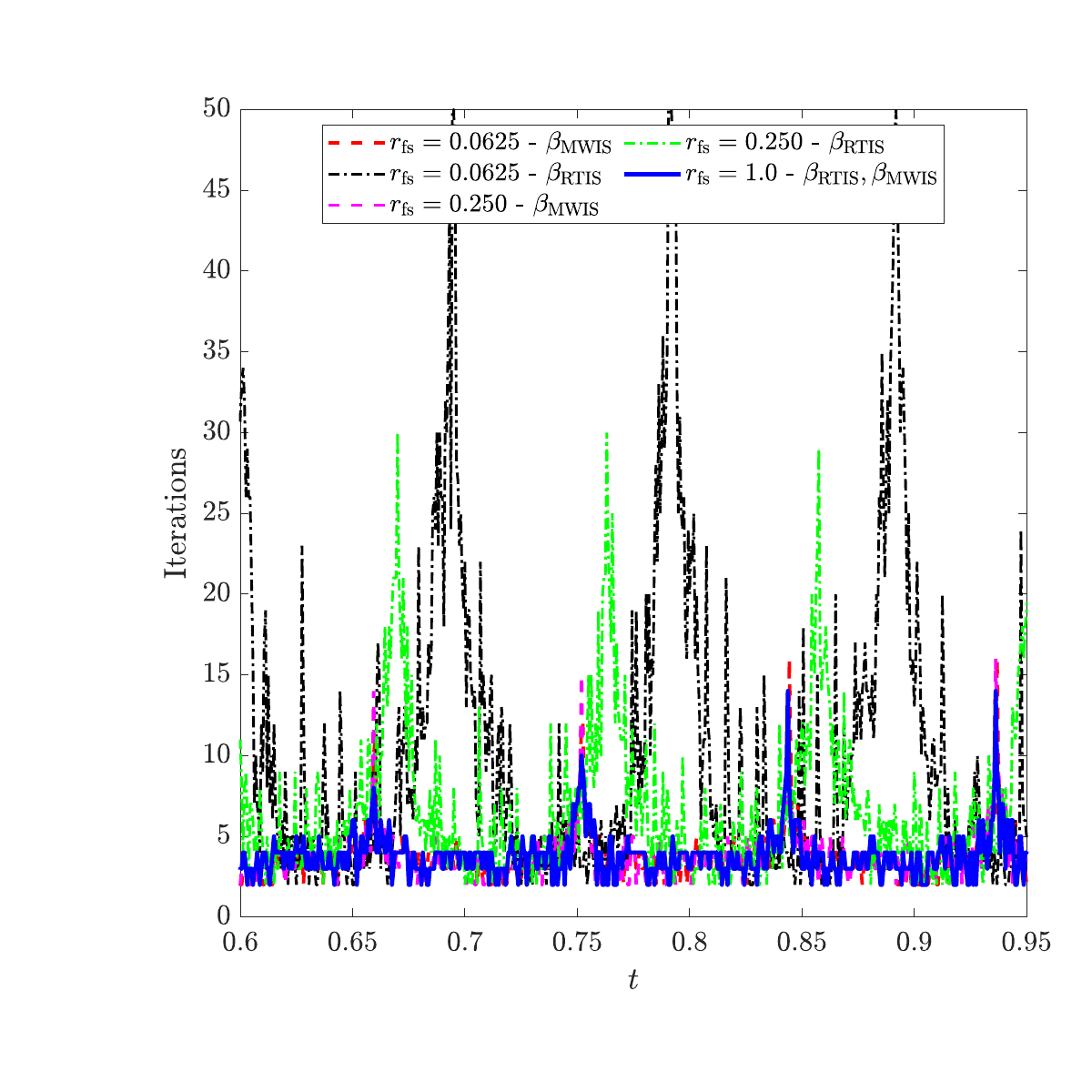}
			\caption{}
			\label{fig:intIter}
		\end{subfigure}
		
		\caption{A comparison of the temporal evolution of (a) the $L_2$ norm of the velocity divergence and (b) the number of iterations required for convergence per time step, using $\beta_\mathrm{RTIS}$ and $\beta_\mathrm{MWIS}$ for interface stabilization. }
        
		\label{fig:divIter}
	\end{figure}

	\subsection{Elastically deforming bar in the unsteady wake of a circular cylinder} 
	\label{sec:fsiStabEval}

	Following the model problem of Sec.~\ref{sec:piston}, the MWIS scheme is assessed 
	on an elastically deforming bar oscillating in the unsteady wake of a circular 
	cylinder. Reference data available for this configuration are first used to 
	validate the coupled fluid and structural response, after which the mass ratio 
	$m^\star$ and the relative mesh ratio $r_\mathrm{fs}$ are varied to assess the 
	robustness of the stabilization under strong added-mass effects.

	\subsubsection{Problem definition}


\begin{figure}[pos=H]
	\centering
\def\thDomL{24}      
\def\thDomH{10}      
\def\thDdx{1.6}      
\def\thDdy{1.0}      
\def\thExt{0.3}      
\def\thCylX{5.0}       
\def\thCylY{4.8}       
\def\thCylR{1.2}       
\def\thFlagL{7}      
\def\thFlagT{0.4}    
\def\thInMax{1.8}    
\def\thLabL{25D}   
\def\thLabH{4.1D}   
\def\thLabCx{2D}     
\def\thLabCy{2D}     
\def\thLabFl{3.5D}   
\def\thLabFt{0.2D}   

\scalebox{1.0}{
\begin{tikzpicture}[scale=0.45, every node/.style={font=\small}, depth/.style={dash pattern=on 1.5pt off 1pt}]

	\pgfmathsetmacro{\thDnorm}{sqrt(\thDdx*\thDdx+\thDdy*\thDdy)}
	\pgfmathsetmacro{\thNx}{-\thDdy/\thDnorm*\thCylR}                      
	\pgfmathsetmacro{\thNy}{\thDdx/\thDnorm*\thCylR}                       
	\pgfmathsetmacro{\thXroot}{\thCylX+sqrt(\thCylR*\thCylR-\thFlagT*\thFlagT/4)}  
	\pgfmathsetmacro{\thXtip}{\thXroot+\thFlagL}                         

	\coordinate (A)  at (0,0);
	\coordinate (B)  at (\thDomL,0);
	\coordinate (C)  at (\thDomL,\thDomH);
	\coordinate (Ff) at (0,\thDomH);
	\coordinate (Ap) at (\thDdx,\thDdy);
	\coordinate (Bp) at (\thDomL+\thDdx,\thDdy);
	\coordinate (Cp) at (\thDomL+\thDdx,\thDomH+\thDdy);
	\coordinate (Fp) at (\thDdx,\thDomH+\thDdy);

	\draw[dashed, thin] (Ap) -- (Bp);
	\draw[dashed, thin] (A) -- (Ap);
	\draw[dashed, thin] (Ap) -- (Fp);

	\draw[thick] (A) -- (B) -- (C) -- (Ff) -- cycle;
	\draw[thick] (Ff) -- (Fp) -- (Cp) -- (C);
	\draw[thick] (B) -- (Bp) -- (Cp);

	\node[above] at (\thDomL/2, \thDomH-0.1) {$\Gamma_\mathrm{top}$};
	\node[below] at (\thDomL/2, -0.3) {$\Gamma_\mathrm{bottom}$};
	\node[left]  at (-0.3, \thDomH/2) {$\Gamma_\mathrm{in}$};
	\node at (\thDomL+\thDdx/2, \thDomH/2+\thDdy/2) {$\Gamma_\mathrm{out}$}; 

	\foreach \y in {0.5,1.5,...,9.5}{
		\pgfmathsetmacro{\thUx}{\thInMax*4*\y*(\thDomH-\y)/(\thDomH*\thDomH)}
		\draw[-{Latex[length=1.5mm]}] (0,\y) -- (\thUx,\y);
	}
	\draw[thin] plot[domain=0:\thDomH, samples=40, smooth]
		({\thInMax*4*\x*(\thDomH-\x)/(\thDomH*\thDomH)}, \x);
	\node[above right] at (\thInMax, \thDomH/2+0.1) {$1.5\bar{U}$};

	\coordinate (Cc) at (\thCylX,\thCylY);
	\coordinate (Cb) at ($(Cc)+(\thExt*\thDdx,\thExt*\thDdy)$);
	\fill[gray!35] (Cb) circle (\thCylR);                                  
	\draw[depth] (Cb) circle (\thCylR);
	\fill[gray!45] ($(Cc)+(\thNx,\thNy)$) -- ($(Cb)+(\thNx,\thNy)$)
		-- ($(Cb)-(\thNx,\thNy)$) -- ($(Cc)-(\thNx,\thNy)$) -- cycle;            
	\draw[depth] ($(Cc)+(\thNx,\thNy)$) -- ($(Cb)+(\thNx,\thNy)$);            
	\draw[depth] ($(Cc)-(\thNx,\thNy)$) -- ($(Cb)-(\thNx,\thNy)$);
	\filldraw[fill=gray!60, thick] (Cc) circle (\thCylR);                 
	\draw[thin] (\thCylX-0.25,\thCylY) -- (\thCylX+0.25,\thCylY)
		(\thCylX,\thCylY-0.25) -- (\thCylX,\thCylY+0.25);                                

	\coordinate (P1) at (\thXroot,\thCylY-\thFlagT/2);
	\coordinate (P2) at (\thXtip,\thCylY-\thFlagT/2);
	\coordinate (P3) at (\thXtip,\thCylY+\thFlagT/2);
	\coordinate (P4) at (\thXroot,\thCylY+\thFlagT/2);
	\coordinate (P2p) at ($(P2)+(\thExt*\thDdx,\thExt*\thDdy)$);
	\coordinate (P3p) at ($(P3)+(\thExt*\thDdx,\thExt*\thDdy)$);
	\coordinate (P4p) at ($(P4)+(\thExt*\thDdx,\thExt*\thDdy)$);
	\fill[cyan] (P1) -- (P2) -- (P3) -- (P4) -- cycle;                 
	\fill[cyan!70] (P4) -- (P3) -- (P3p) -- (P4p) -- cycle;             
	\fill[cyan!55] (P2) -- (P3) -- (P3p) -- (P2p) -- cycle;             
	\draw[thick] (P1) -- (P2) -- (P3) -- (P4);
	\draw[depth] (P4p) -- (P3p) -- (P2p);                               
	\draw[depth] (P4) -- (P4p);                                        
	\draw[depth] (P3) -- (P3p);
	\draw[depth] (P2) -- (P2p);

	\fill (\thXtip,\thCylY) circle (0.12);
	\node[below left] at (\thXtip,\thCylY-\thFlagT/2) {$A$};

	\draw[{Latex[length=1.5mm]}-{Latex[length=1.5mm]}]
		($(Cc)+(300:\thCylR)$) -- ($(Cc)+(120:\thCylR)$);
	\draw ($(Cc)+(120:\thCylR)$) -- ($(Cc)+(120:\thCylR+1.6)$)
		-- ++(-1.0,0) node[pos=0.5, above] {$\varnothing D$};

	\draw[{Latex[length=1.5mm]}-{Latex[length=1.5mm]}] (\thCylX,0) -- (\thCylX,\thCylY);
	\node[right] at (\thCylX, {(\thCylY-\thCylR)/2}) {$\thLabCy$};

	\draw[thin] (\thCylX,\thCylY) -- (\thCylX,-0.95);
	\draw[{Latex[length=1.5mm]}-{Latex[length=1.5mm]}] (0,-0.75) -- (\thCylX,-0.75);
	\node[below] at (\thCylX/2, -0.75) {$\thLabCx$};

	\draw[thin] (\thCylX+\thCylR,\thCylY-\thFlagT/2) -- (\thCylX+\thCylR,\thCylY-\thFlagT/2-1.5);
	\draw[thin] (\thXtip,\thCylY-\thFlagT/2) -- (\thXtip,\thCylY-\thFlagT/2-1.5);        
	\draw[{Latex[length=1.5mm]}-{Latex[length=1.5mm]}]
		(\thCylX+\thCylR,\thCylY-\thFlagT/2-1.3) -- (\thXtip,\thCylY-\thFlagT/2-1.3);
	\node[below] at ({(\thCylX+\thCylR+\thXtip)/2}, \thCylY-\thFlagT/2-1.3) {$\thLabFl$};

	\draw[thin] (\thXtip+0.6,\thCylY+\thFlagT/2) -- (\thXtip+1.3,\thCylY+\thFlagT/2);
	\draw[thin] (\thXtip+0.6,\thCylY-\thFlagT/2) -- (\thXtip+1.3,\thCylY-\thFlagT/2);
	\draw[-{Latex[length=1.5mm]}] (\thXtip+1.0,\thCylY+\thFlagT/2+0.8) -- (\thXtip+1.0,\thCylY+\thFlagT/2);
	\draw[-{Latex[length=1.5mm]}] (\thXtip+1.0,\thCylY-\thFlagT/2-0.8) -- (\thXtip+1.0,\thCylY-\thFlagT/2);
	\node[right] at (\thXtip+1.3, \thCylY) {$\thLabFt$};

	\draw[thin] (0,0) -- (0,-2.0);                                    
	\draw[thin] (\thDomL,0) -- (\thDomL,-2.0);                 
	\draw[{Latex[length=1.5mm]}-{Latex[length=1.5mm]}] (0,-1.8) -- (\thDomL,-1.8);
	\node[below] at (\thDomL/2, -1.8) {$\thLabL$};

	\draw[thin] (Bp) -- ++(1.2,0);                                      
	\draw[thin] (Cp) -- ++(1.2,0);
	\draw[{Latex[length=1.5mm]}-{Latex[length=1.5mm]}] (\thDomL+\thDdx+1.0,\thDdy) -- (\thDomL+\thDdx+1.0,\thDdy+\thDomH);
	\node[right] at (\thDomL+\thDdx+1.0, \thDomH*0.5+\thDdy) {$\thLabH$};

	\draw[thin] (Bp) -- ++(0,-0.9);
	\draw[{Latex[length=1.5mm]}-{Latex[length=1.5mm]}] ($(B)+(0,-0.7)$) -- ($(Bp)+(0,-0.7)$);
	\node[below right] at ($(B)!0.5!(Bp)+(0,-0.7)$) {$D$};

\end{tikzpicture}}
\caption{Schematic of the computational domain and boundaries for the elastic bar in the wake of a circular cylinder \cite{2006_turek_FSIBench}.}
\label{fig:schematicTurekHron}
\end{figure}
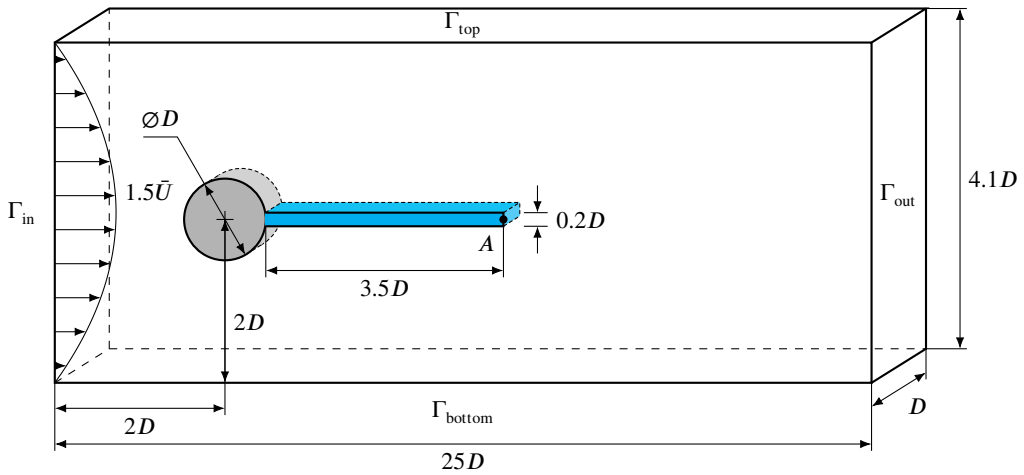

	This configuration was introduced by \citet{2006_turek_FSIBench} as a benchmark for 
	FSI solvers and has since been widely used to assess coupling schemes. The geometric 
	and material properties adopted in the present study largely follow their setup, 
	with the domain extended spanwise over a width $D$. The fluid is modeled as 
	incompressible and Newtonian, governed by the Navier-Stokes equations (see Appendix.~
    \ref{sec:app_governing}), with density 
	$\rho^\mathrm{f} = 1000\,\mathrm{kg\,m^{-3}}$ and dynamic viscosity 
	$\mu^\mathrm{f} = 1.0\,\mathrm{Pa\,s}$. A stationary rigid cylinder of diameter 
	$D = 0.1\,\mathrm{m}$, centered at the origin, with a flexible bar attached to it, together forming the structural domain $\Omega^\mathrm{s}$. Both are 
	enclosed in a rectangular channel bounded by the inlet $\Gamma_\mathrm{in}$, the 
	outlet $\Gamma_\mathrm{out}$, and the top and bottom walls $\Gamma_\mathrm{top}$ and 
	$\Gamma_\mathrm{bottom}$. The 
	remaining geometric parameters are given in Fig.~\ref{fig:schematicTurekHron}. 
    
    A parabolic velocity profile is prescribed at the inlet as
	\begin{align}
		\boldsymbol{u}^\mathrm{f} &= (u_0, 0), \nonumber\\
		u_0 &= \frac{6\bar{U}g(t)}{(4.1D)^2}(y+2D)(2.1D - y),
	\end{align}
	where $\bar{U}$ is the mean inlet velocity and 
    \begin{equation}
        g(t) = \begin{cases}
            \dfrac{1-\cos\left(\pi t/2\right)}{2} & t \leq 2~\text{s} \\
            1.0 & \text{otherwise.}
        \end{cases}
    \end{equation} 
	ramps the inflow from rest. No-slip conditions are applied on the channel walls and 
	the cylinder surface, while the outlet is traction-free. Periodic conditions are 
	imposed on the spanwise boundaries, so that the three-dimensional mesh represents a 
	nominally two-dimensional flow.

	The bar is clamped to the cylinder, while its other boundaries are free to deform 
	and exposed to the fluid domain $\Omega^\mathrm{f}(t)$, forming the fluid-structure 
	interface $\Gamma^\mathrm{fs}$. The deforming bar's material is characterized by a Young's modulus 
	$E^\mathrm{s}$, a Poisson's ratio $\nu^\mathrm{s}$ and a density $\rho^\mathrm{s}$, 
	so that the problem is governed by the non-dimensional parameters
	\begin{equation}
		Re = \frac{\rho^\mathrm{f}\bar{U} D}{\mu^\mathrm{f}}, \quad 
		m^\star = \frac{\rho^\mathrm{s}}{\rho^\mathrm{f}}, \quad 
		Ae = \frac{E^\mathrm{s}}{\rho^\mathrm{f}\bar{U}^2},
	\end{equation}
	where $Re$ is the Reynolds number, $m^\star$ is the structure-to-fluid mass ratio, 
	and $Ae$ is the aeroelastic number. The structural response is monitored at the tip 
	of the bar, located at $(4D, 0)$ and marked as probe A in 
	Fig.~\ref{fig:schematicTurekHron}. A time step of 
	$\Delta t = 5\times 10^{-4}\,\mathrm{s}$ is used in all simulations.

	\subsubsection{Mesh convergence}\label{sec:meshConv}

    The coupled fluid-structure interaction simulation is conducted on three successively 
	refined meshes, $M1$, $M2$ and $M3$, summarized in Table~\ref{tab:fsi-iii}. The 
	relevant non-dimensional parameters used for all levels of mesh refinement are
	\begin{equation}
		Re = 200, \quad 
		m^\star = 1.0, \quad 
		Ae = 1.4\times10^{3},
	\end{equation}
	with the structural Poisson's ratio set to $\nu^\mathrm{s} = 0.4$. The number of 
	elements in both the fluid and structural domains approximately doubles between 
	successive mesh levels. In the fluid domain, a boundary layer mesh surrounds the 
	cylinder and the bar to resolve the near-wall gradients, while the wake is resolved 
	with in-plane isotropic elements, both extruded across the spanwise direction as prismatic elements. The 
    bar is discretized with a structured mesh that matches the fluid mesh along 
    the interface.

    \begin{table}[pos=H]
        \centering
        \begin{tabular}{l c c c c c c c c}
        \toprule
        \multirow{2}{*}{Mesh} & \multicolumn{2}{c}{Solid} & \multicolumn{2}{c}{Fluid} & $\Delta x$ & $\Delta y$ & $f_x$ & $f_y$ \\
         &     DOFs      &     Elements     &      DOFs     &     Elements     & $(\times 10^{-3}\,\mathrm{m})$ & $(\times 10^{-3}\,\mathrm{m})$ &
			$(\mathrm{Hz})$ & $(\mathrm{Hz})$ \\ 
        \midrule
        $M1$ & $2688$ & $385$ & $52816$ & $12860$ & $-2.57 \pm 2.42$ & $0.86 \pm 32.79$ &
			$11$ & $5.5$ \\
		$M2$ & $5214$ & $780$ & $94488$ & $23137$ & $-2.89 \pm 2.68$ & $0.68 \pm 34.94$ &
			$11$ & $5.5$ \\
		$M3$ & $10752$ & $1665$ & $179192$ & $44104$ & $-2.90 \pm 2.74$ & $1.12 \pm 35.14$ &
			$11$ & $5.5$ \\
		Reference \citep{2006_turek_FSIBench} & $98820$ & $5120$ & $705152$ & $36864$ &
			$-2.69 \pm 2.53$ & $1.48 \pm 34.38$ & $10.9$ & $5.3$ \\
        \bottomrule
        \end{tabular}
		\caption{Comparison of the displacement and frequency response at probe A against the FSI3 benchmark. Displacements are reported as mean $\pm$ amplitude.}
		\label{tab:fsi-iii}
    \end{table}

    Fig.~\ref{fig:meshConv} shows the displacement histories for the three mesh 
	refinements, recorded at the structural probe A. While all three meshes capture the 
	qualitative dynamics and the oscillatory response of the elastic bar, meshes $M2$ 
	and $M3$ yield nearly identical amplitudes, whereas $M1$ underpredicts the 
	cross-stream amplitude (see Fig.\ref{fig:dxconvFSI3}). The dominant frequencies, 
	$f_x = 11\,\mathrm{Hz}$ and $f_y = 5.5\,\mathrm{Hz}$, are insensitive to the 
	refinement, with the streamwise oscillation occurring at twice the frequency of the 
	cross-stream oscillation. 
    
	\begin{figure}[pos=H]
		\centering
		\begin{subfigure}[b]{0.49\linewidth}
			\centering
			\includegraphics[width=\linewidth]{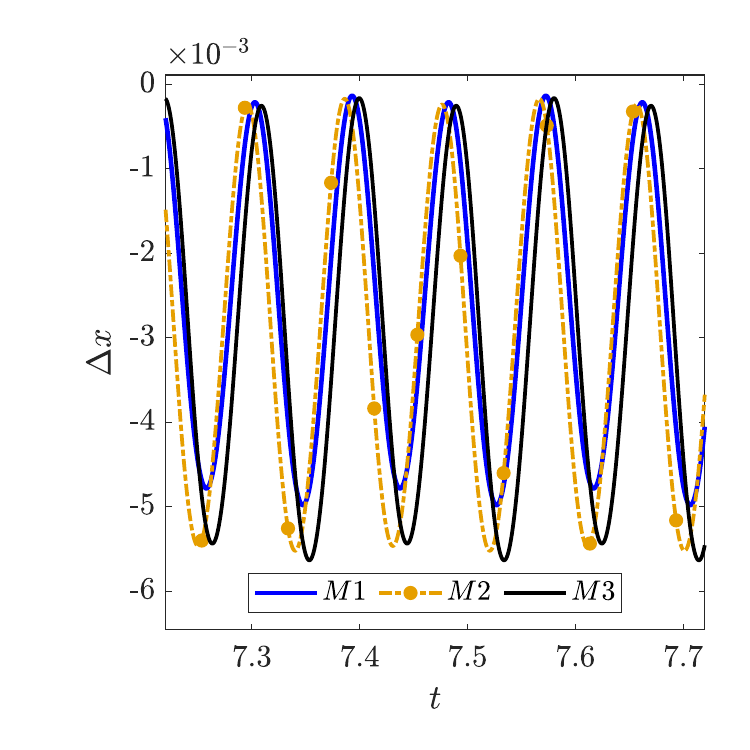}
			\caption{}
			\label{fig:dxconvFSI3}
		\end{subfigure}
		\hfill
		\begin{subfigure}[b]{0.49\linewidth}
			\centering
			\includegraphics[width=\linewidth]{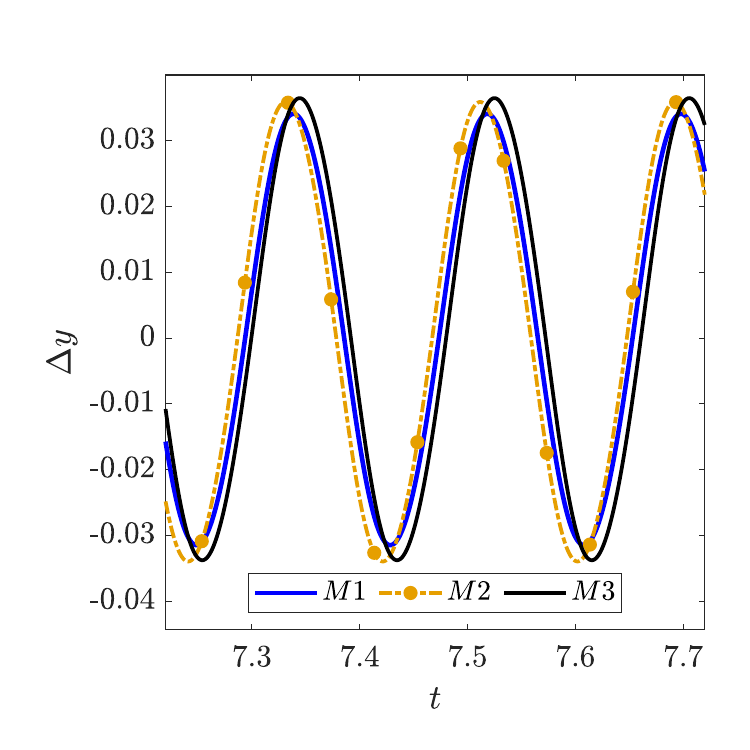}
			\caption{}
			\label{fig:dyconvFSI3}
		\end{subfigure}
		\caption{Temporal variation of (a) x-displacement and (b) y-displacement of the structural probe A (marked in Fig.~\ref{fig:schematicTurekHron}) evaluated for three
        distinct meshes.}
        
		\label{fig:meshConv}
	\end{figure}

    As summarized in Table~\ref{tab:fsi-iii}, the results on 
	$M2$ and $M3$ agree well with the reference solution of \citet{2006_turek_FSIBench},
    wherein a monolithic finite-element framework was used to numerically simulate
    the coupled problem. The y-displacement amplitude measured with $M3$ mesh lies within 
    approximately $2\%$ of the reference value. Consequently, the $M3$ mesh is 
    adopted in what is to follow.
	
	\subsubsection{Mass ratio variation}

    Next, the robustness of the proposed MWIS method is evaluated with respect to the 
	mass ratio. On mesh $M3$, the fluid and structural element sizes normal to the 
	interface give an aspect ratio of $r_\mathrm{fs} = 1.112$. Based on the spectral 
	characteristics established for the model problem in Table~\ref{tab:muSweep}, this 
	choice of $r_\mathrm{fs}$ is expected to challenge the stability of the iterations 
	once $m^\star$ reaches $\mathcal{O}(10^{-1})$. Accordingly, five simulations are 
	conducted with $m^\star \in \{0.05, 0.1, 0.2, 0.5, 1.0\}$, while all other 
	non-dimensional parameters are kept the same as in Sec.~\ref{sec:meshConv}.
	
	The pressure residual ratio at the $k^\mathrm{th}$ iteration is defined as
	\begin{equation}
		r_p^{(k)} = \frac{\displaystyle \sum_{n \in N_p} \left| R_n^{(k)}
			\right|}{\displaystyle \sum_{e = 1}^{n_{el}}\sum_{j =
				1}^{n_{n,el}}\left|r_j^{e,(k)}\right|},
	\end{equation} where $r_j^{e,(k)}$ is the contribution to the continuity
	equation to the local node $j$ of element $e$, and $R_n^{(k)}$ is the global
	residual at the pressure degree of freedom $n$. For the experiments in the
	current section, the nonlinear iterations within a time step are terminated when
	
	\begin{equation}
		r_p^{(k)} \;<\; \epsilon_{\text{tol}}
		\quad \text{and} \quad k \leq k_{\max} ,
	\end{equation}
	where $\epsilon_{\text{tol}}$ is the prescribed tolerance and $k_{\max}$ is the
	maximum number of nonlinear iterations allowed per time-step, which are set to
    $10^{-4}$ and $50$, respectively.
	
	As shown in Fig.~\ref{fig:itersMstarFSI3}, the coupled iterations converge rapidly 
	for $m^\star > 0.1$, requiring on average 4, 6 and 9 iterations per time step 
	for $m^\star = 0.2$, $0.5$ and $1.0$, respectively, compared with 18 at 
	$m^\star = 0.1$. The spectral radius of the iteration is estimated from the last 
	four iterates before convergence, and its temporal variation is presented in 
	Fig.~\ref{fig:rhoMstarFSI3}. The estimate remains well below unity for 
	$m^\star \geq 0.5$, but approaches unity for $m^\star = 0.2$ and $0.1$, which
    is expected. At $m^\star = 0.1$, the augmentation provided by $\tilde{C}$ can
    be observed to be close to the stability limit, restricting both the convergence rate 
    and its capacity to mitigate the added-mass effect. At $m^\star = 0.05$, the 
    iteration cap $k_{\max}$ is reached at each time step and the residual reduction is 
    insufficient to prevent error accumulation, leading to slow divergence.	
	\begin{figure}[pos=H]
		\centering
		\begin{subfigure}[b]{0.49\linewidth}
			\centering
			\includegraphics[trim={2cm 1.5cm 0.0cm 0cm}, clip,width=\textwidth]{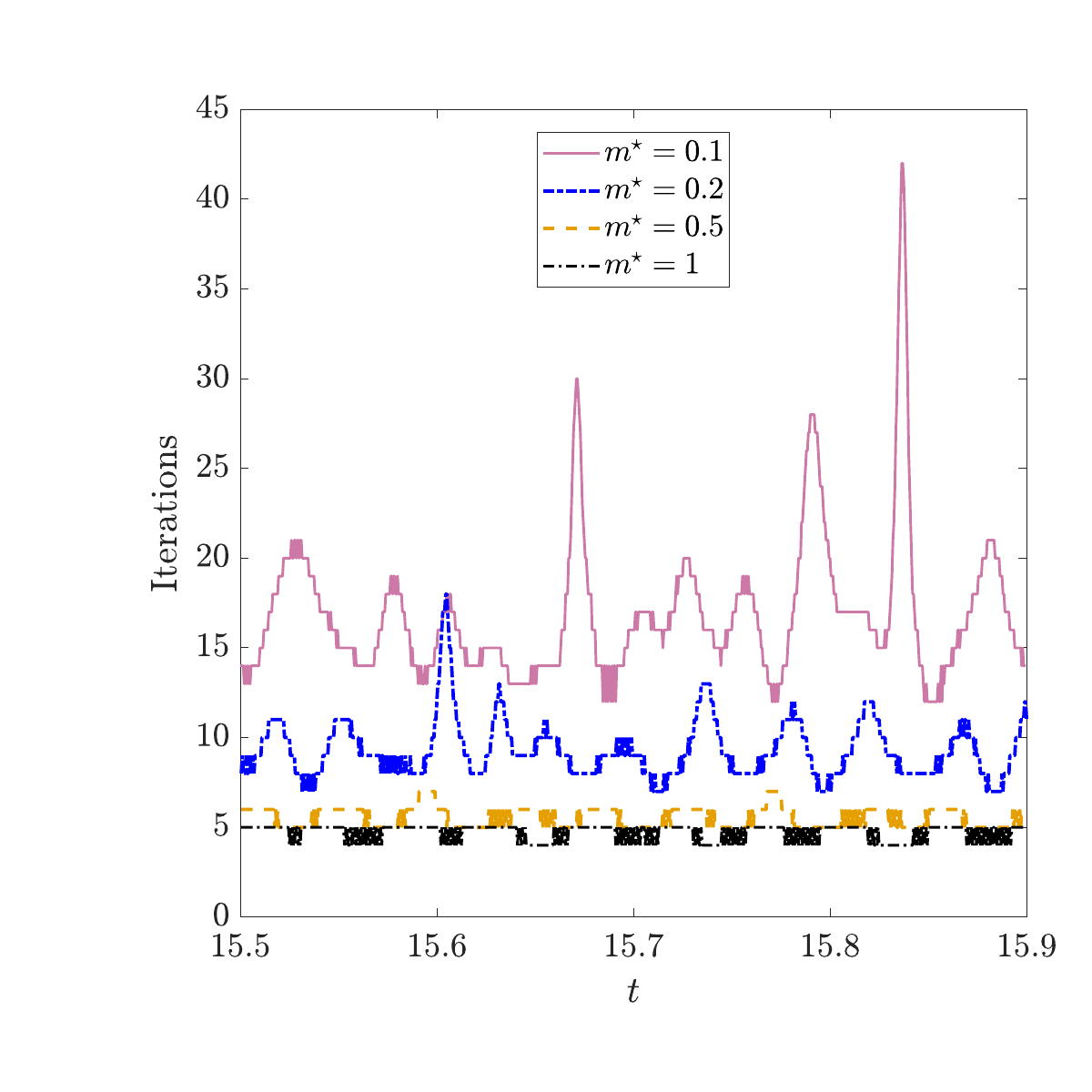}
			\caption{}
			\label{fig:itersMstarFSI3}
		\end{subfigure}
		\hfill
		\begin{subfigure}[b]{0.49\linewidth}
			\centering
			\includegraphics[trim={2cm 1.5cm 0.0cm 0cm}, clip,width=\textwidth]{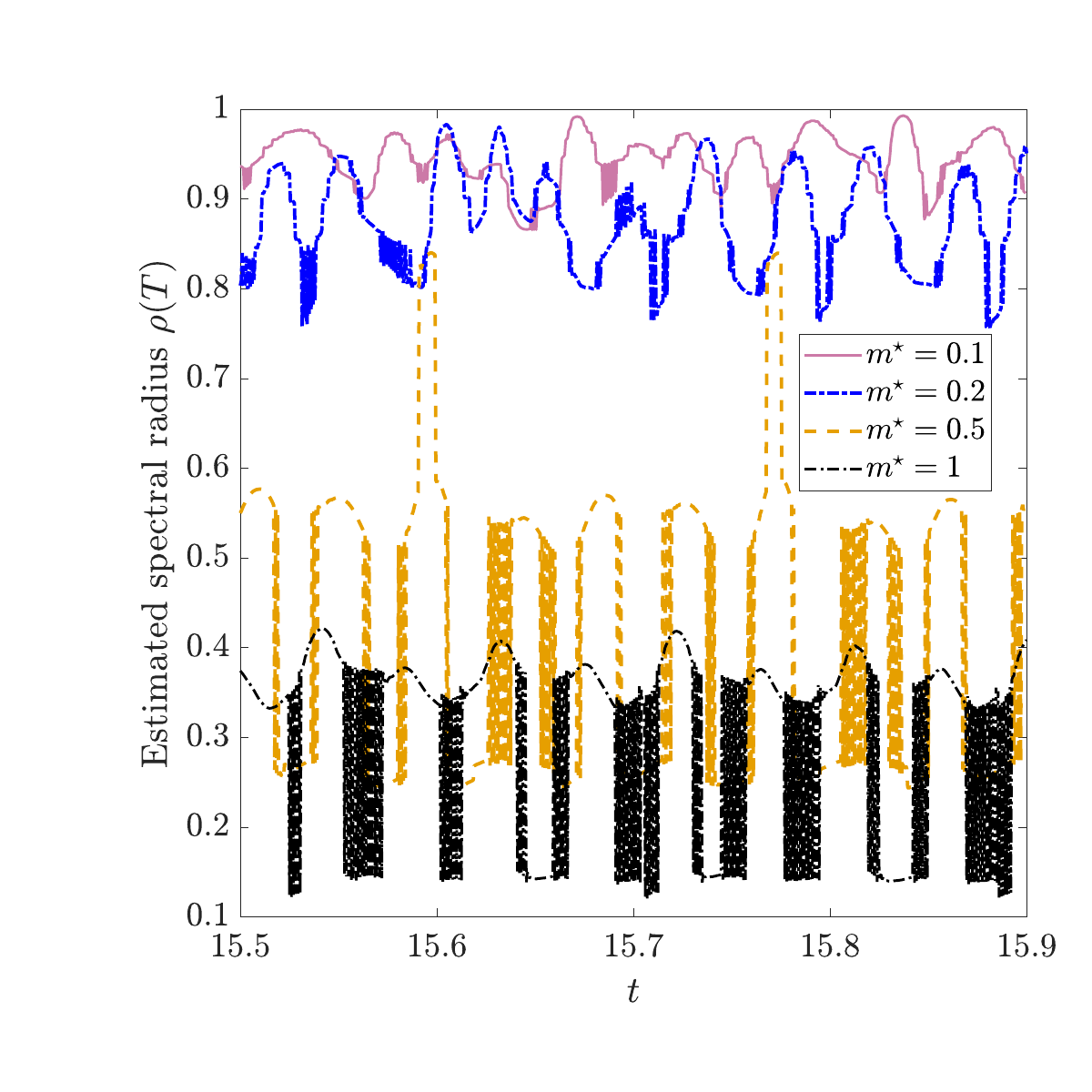}
			\caption{}
			\label{fig:rhoMstarFSI3}
		\end{subfigure}
		\caption{Convergence of the nonlinear iterations over a representative interval $t \in [15.5, 15.9]$ for mass ratios $m^\star = 0.1$, $0.2$, $0.5$ and $1$, showing (a) the number of iterations required per time step and (b) the estimated spectral radius of the nonlinear iterations.}
		\label{fig:iterdata}
	\end{figure} 
	

    The dynamical characteristics for these lower mass ratios are depicted in 
	Fig.~\ref{fig:disp}. From Figs.~\ref{fig:disp}a and \ref{fig:disp}c, regular, nearly 
    sinusoidal oscillations can be observed for $0.2 \leq m^\star \leq 1$, with the 
	$x$-displacement oscillating at twice the frequency of the $y$-displacement. As 
	$m^\star$ decreases, the reduced structural inertia relative to the fluid lowers the 
	oscillation amplitude, most noticeably between $m^\star = 1$ and $0.5$, while the 
	dominant frequency remains nearly unchanged. The corresponding spectra in 
	Figs.~\ref{fig:disp}b and \ref{fig:disp}d confirm a single dominant peak in each 
	direction for $m^\star \geq 0.2$. At $m^\star = 0.1$, additional sub-harmonic and 
	super-harmonic components emerge in the $x$-displacement spectrum, indicating a 
    departure from the purely sinusoidal response while the motion remains periodic.
	
	\begin{figure}[pos=H]
		\centering
		\begin{subfigure}[b]{0.45\linewidth}
			\centering
			\includegraphics[trim={2cm 1.5cm 0.0cm 0cm}, clip,width=\textwidth]{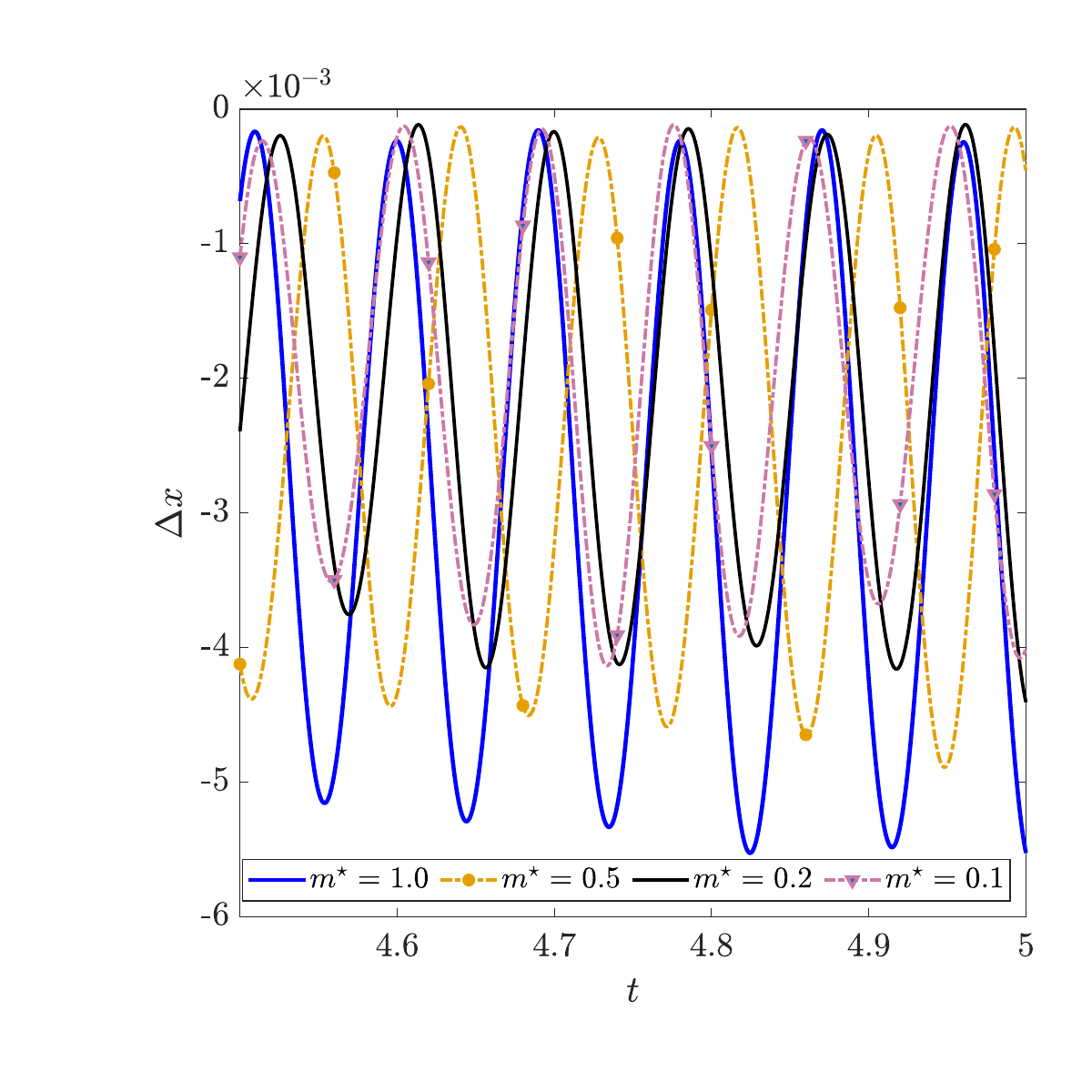}
			\caption{}
			\label{fig:dxMstarFSI3}
		\end{subfigure}
		\hfill
		\begin{subfigure}[b]{0.45\linewidth}
			\centering
			\includegraphics[trim={2cm 1.5cm 0.0cm 0cm}, clip,width=\textwidth]{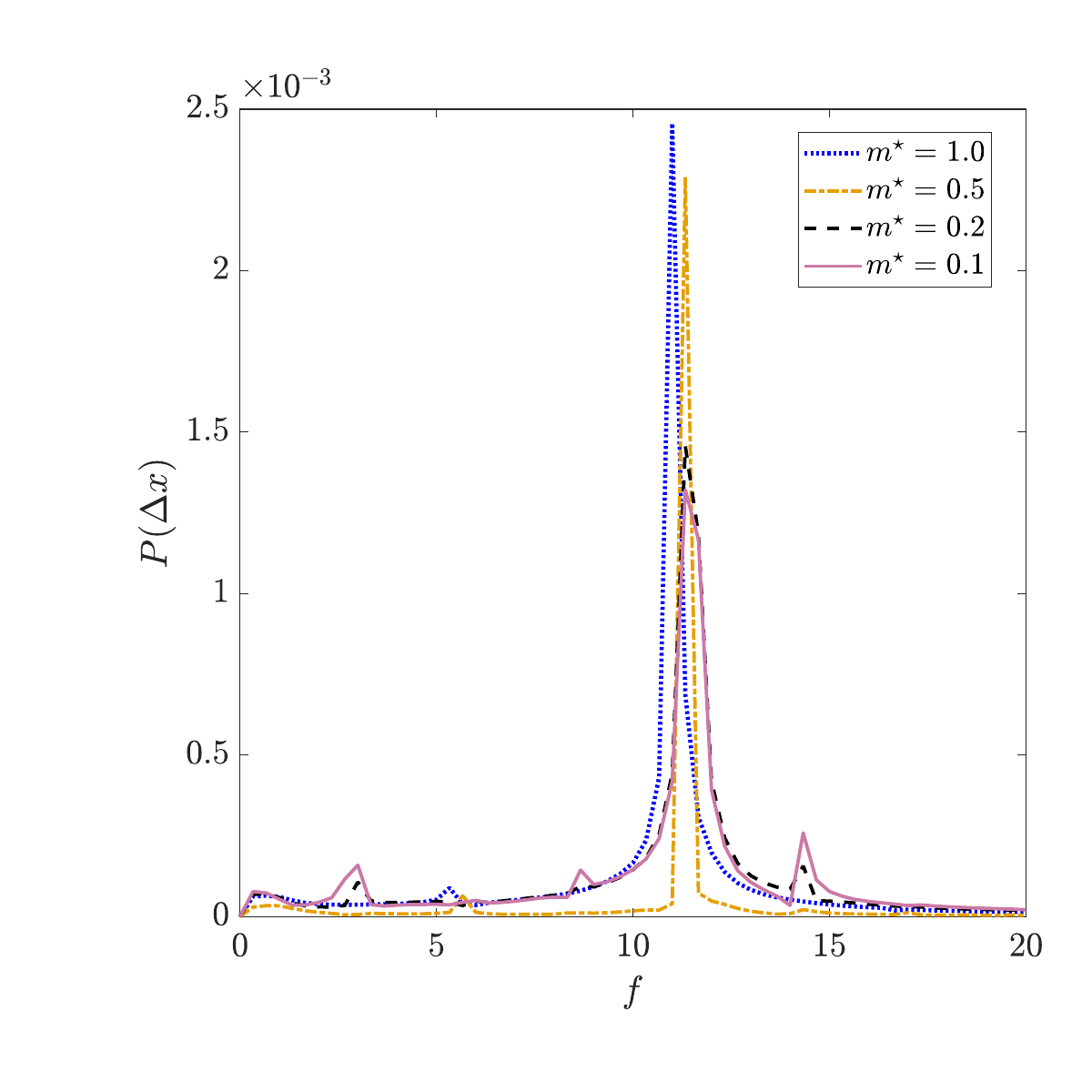}
			\caption{}
			\label{fig:fdxMstarFSI3}
		\end{subfigure}
		\begin{subfigure}[b]{0.45\linewidth}
			\centering
			\includegraphics[trim={2cm 1.5cm 0.0cm 0cm}, clip,width=\textwidth]{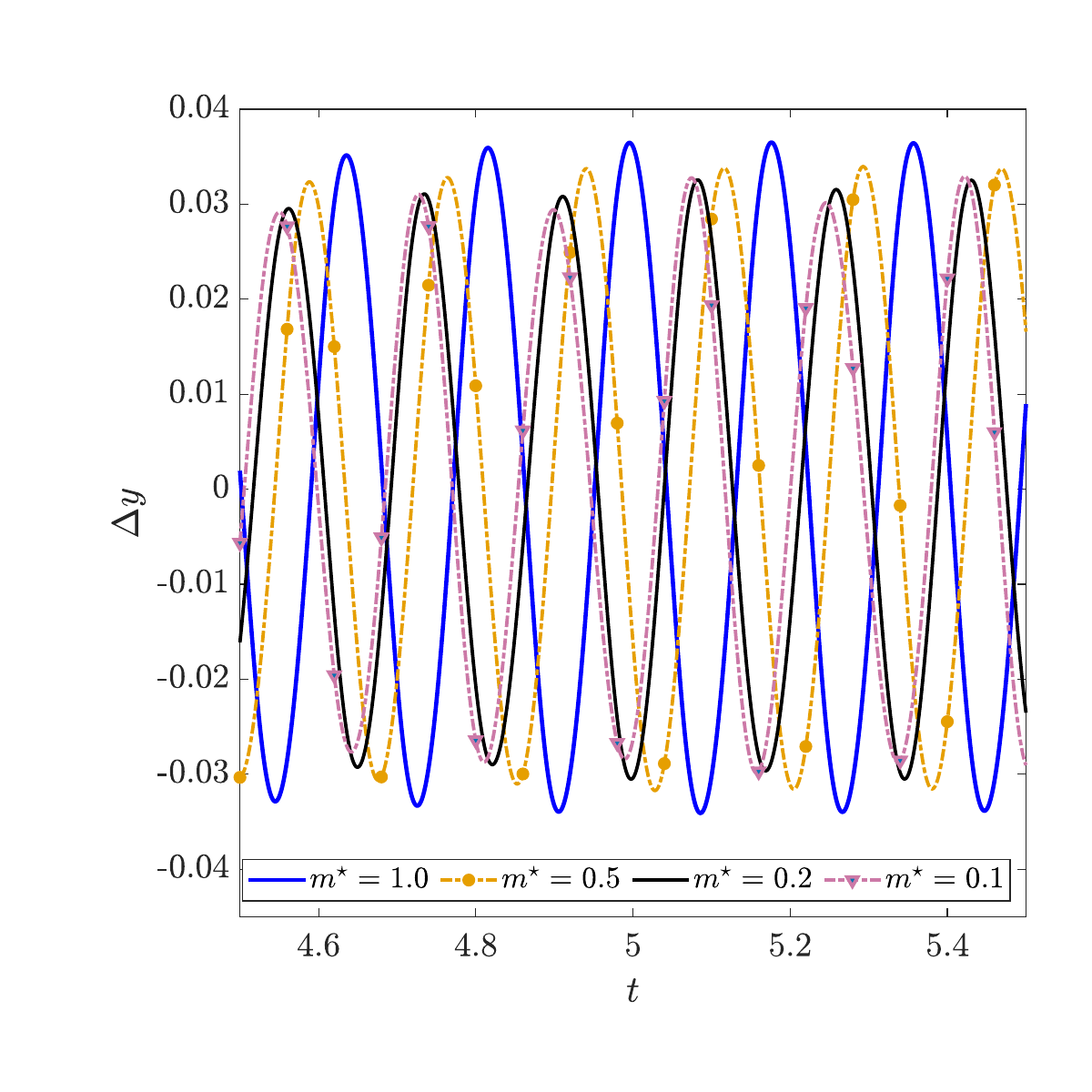}
			\caption{}
			\label{fig:dyMstarFSI3}
		\end{subfigure}
		\hfill
		\begin{subfigure}[b]{0.45\linewidth}
			\centering
			\includegraphics[trim={2cm 1.5cm 0.0cm 0cm}, clip,width=\textwidth]{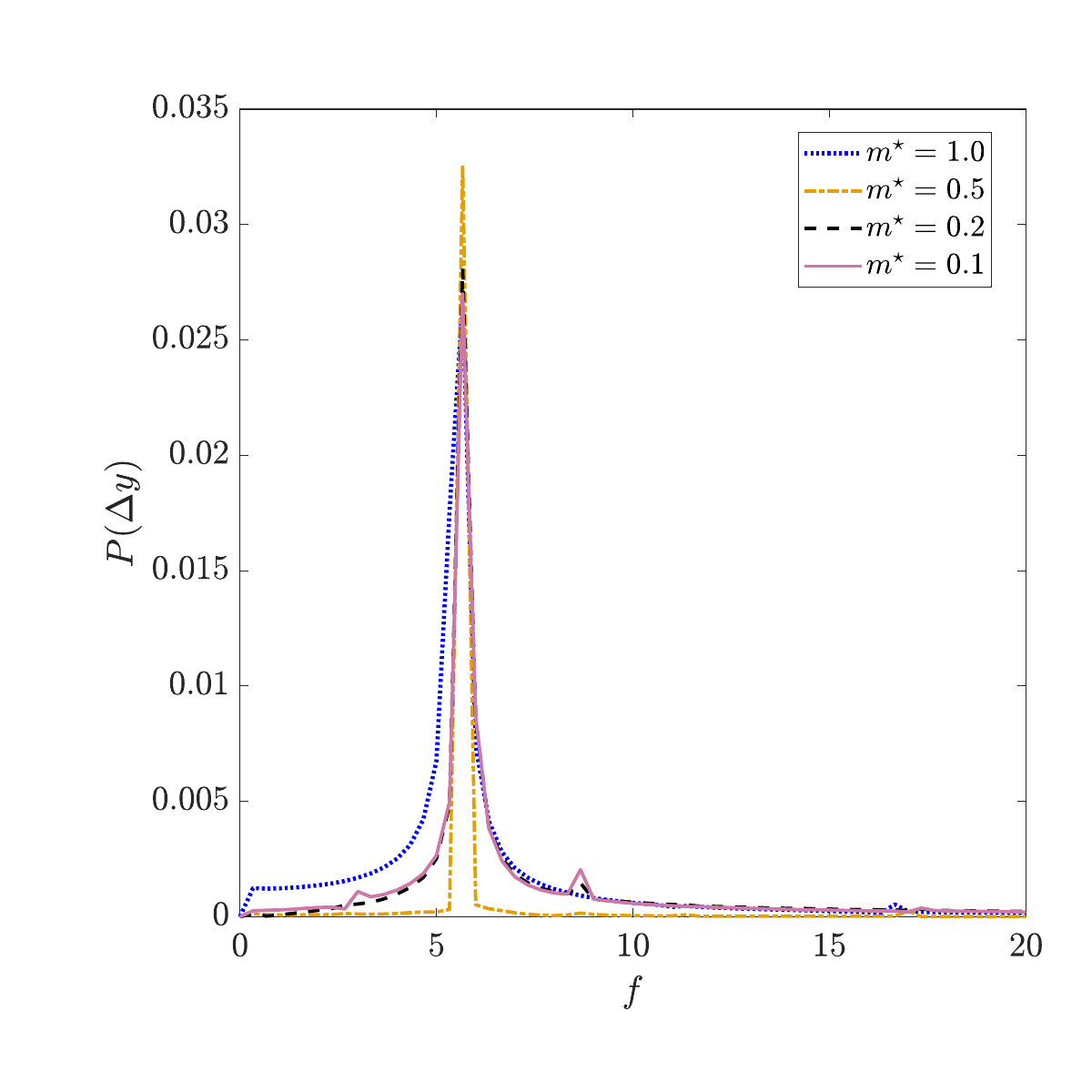}
			\caption{}
			\label{fig:fdyMstarFSI3}
		\end{subfigure}
        \caption{Tip displacement response of the structural probe A (see Fig.~\ref{fig:schematicTurekHron}) for different mass ratios $m^\star$, showing the (a) $x$-displacement and (c) $y$-displacement over a representative interval, with (b) and (d) the corresponding frequency spectra. The dominant frequency of the $x$-displacement is twice that of the $y$-displacement, and sub-harmonic and super-harmonic components emerge as $m^\star$ is reduced.}
		\label{fig:disp}
	\end{figure}
	
	\subsubsection{Relative mesh ratio variation}\label{sec:AR}

    Finally, the robustness of the MWIS scheme is assessed with respect to the interface 
	discretization. Since $\tilde{C}$ is constructed from the element metrics at the 
	interface, changing $r_\mathrm{fs}$ alters the stabilization introduced by the 
	augmentation, while the physical added-mass effect remains unchanged. This effect is 
	isolated by fixing the mass ratio at $m^\star = 1$ and adjusting the fluid element 
	size normal to the bar to obtain $r_\mathrm{fs} \in \{0.25, 0.5, 1.0, 2.0\}$. 
    
	\begin{figure}[pos=H]
		\centering
		\begin{subfigure}[b]{0.45\linewidth}
			\centering
			\includegraphics[trim={2cm 1.5cm 0.0cm 0cm}, clip,width=\textwidth]{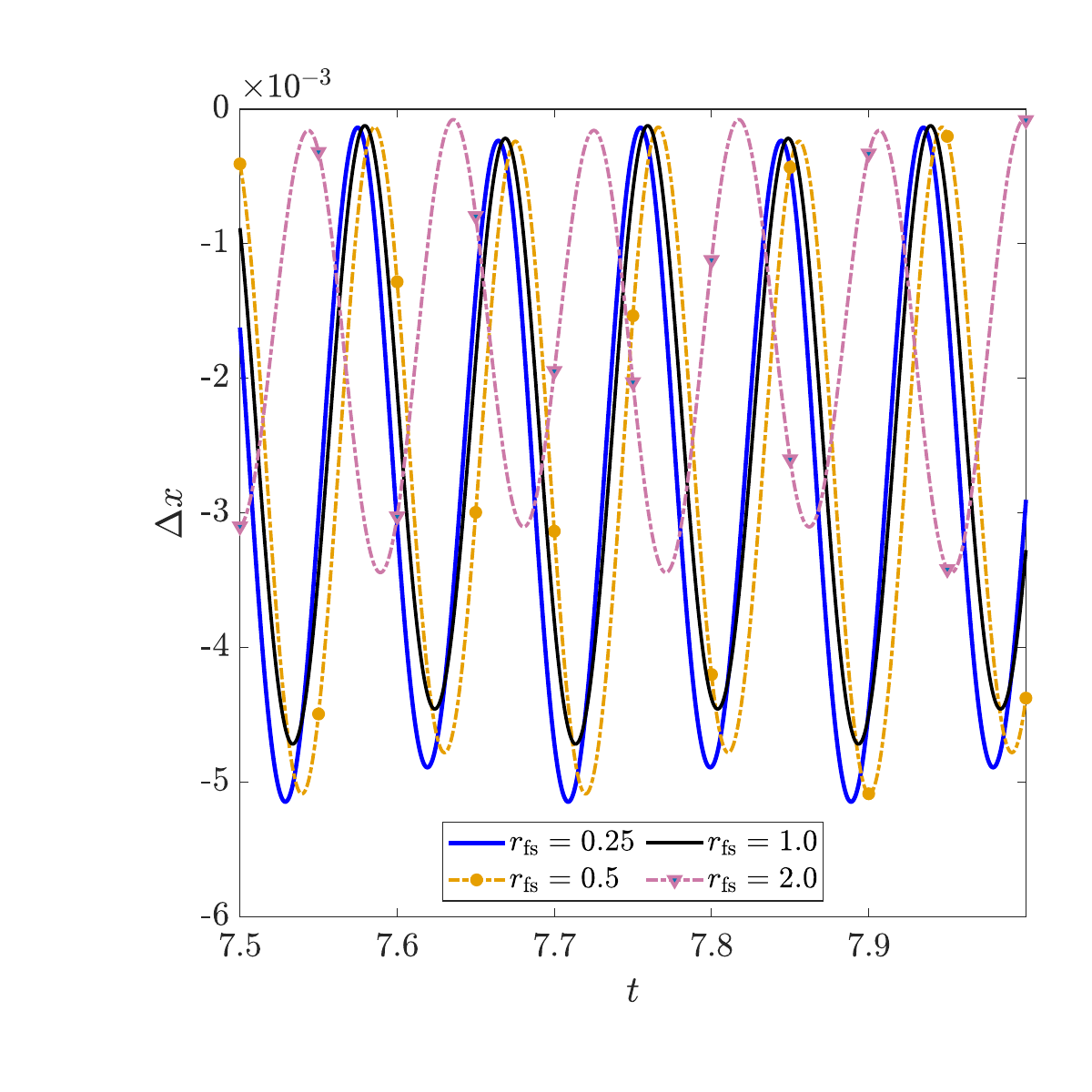}
			\caption{}
			\label{fig:dxARFSI3}
		\end{subfigure}
		\hfill
		\begin{subfigure}[b]{0.45\linewidth}
			\centering
			\includegraphics[trim={2cm 1.5cm 0.0cm 0cm}, clip,width=\textwidth]{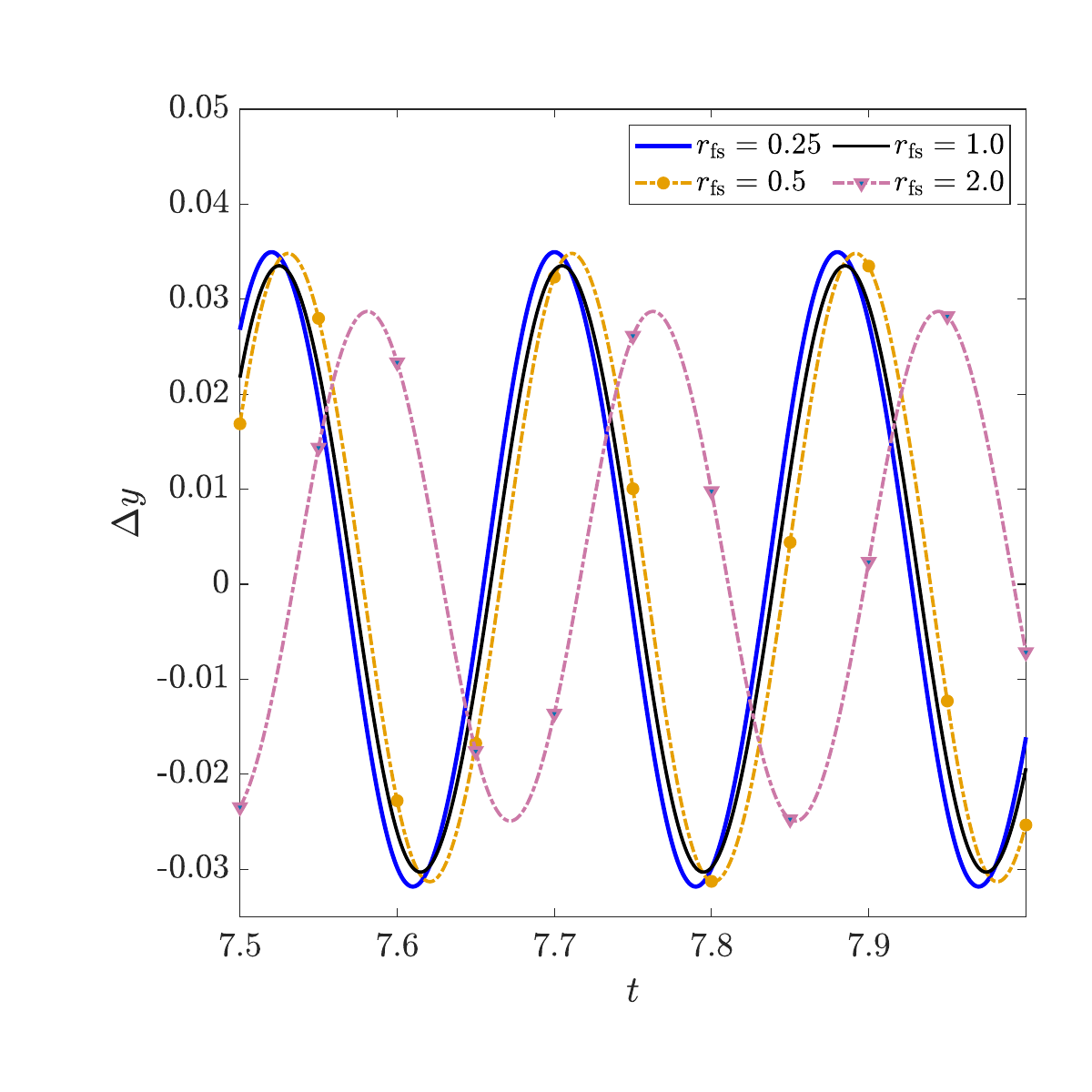}
			\caption{}
			\label{fig:dyARFSI3}
		\end{subfigure}
		\caption{Temporal variation of (a) x-displacement and (b) y-displacement of the structural probe A (marked in Fig.~\ref{fig:schematicTurekHron}).}
        
		\label{fig:dispAR}
	\end{figure}
    
    Fig.~\ref{fig:dispAR} illustrates the displacement histories recorded at the
    structural probe A in the developed periodic regime. For $r_\mathrm{fs} \leq 1$, 
    the responses are nearly identical in amplitude and frequency. The cross-stream 
    oscillation occurs at approximately $5.5$~Hz and the streamwise oscillation at 
    twice that frequency, consistent with Table~\ref{tab:fsi-iii}. Refining 
    the fluid mesh at the interface therefore leaves the computed response unchanged, 
    indicating that the augmentation does not alter the converged solution. 
    
    For $r_\mathrm{fs} = 2.0$, the coupling remains stable and the oscillation frequency 
    is preserved. However, the cross-stream amplitude decreases and the mean streamwise
    displacement shifts toward zero. This loss of 
	accuracy is attributed to the coarser fluid elements adjacent to the bar, which 
	under-resolve the boundary layer and hence the fluid loading on the structure. The 
	choice of $r_\mathrm{fs}$ is therefore governed jointly by the flow resolution 
	required near the interface and by the constraints it imposes on the stabilization.
    
	
	\subsection{Flapping dynamics of a thin flexible plate} \label{sec:fsiFlexFlag}
    
	In the current section, the flapping phenomenon associated with the deformation of a thin flexible plate is evaluated. This application further demonstrates the robustness of the MWIS scheme for tackling problems where structural slenderness drives the added-mass effect. At low mass ratios, the strong added-mass effect governs the stability regime of the plate as illustrated in \citet{LIU2014687}. For this FSI problem, the non-dimensional parameters governing the problem are
	\begin{equation}
		Re = \frac{\rho^\mathrm{f}U_\infty L}{\mu^\mathrm{f}}, m^\star =
		\frac{\rho^\mathrm{s}t^\mathrm{s}}{\rho^\mathrm{f}L}, K_\mathrm{B}
		=\frac{E^\mathrm{s}I}{\rho^\mathrm{f}U_\infty^2L^3b}, H =
		\frac{t^\mathrm{s}}{L},
	\end{equation} where $Re$ denotes the Reynolds number, $m^\star$ is the mass
	ratio of the flexible plate with the surrounding fluid,
	$K_\mathrm{B}$ is the modulus of bending rigidity, and $H$ is the ratio of 
thickness $t^\mathrm{s}$ to length $L$ of the plate.
	$I$ is the area moment of inertia of the plate defined as $bh^3/12(1-\nu^2)$,
	where $b$ is the width of the plate, $\nu$ is the Poisson's ratio
	and $E^\mathrm{s} $ is the Young's modulus of the material constituting the
	flexible plate. 
	

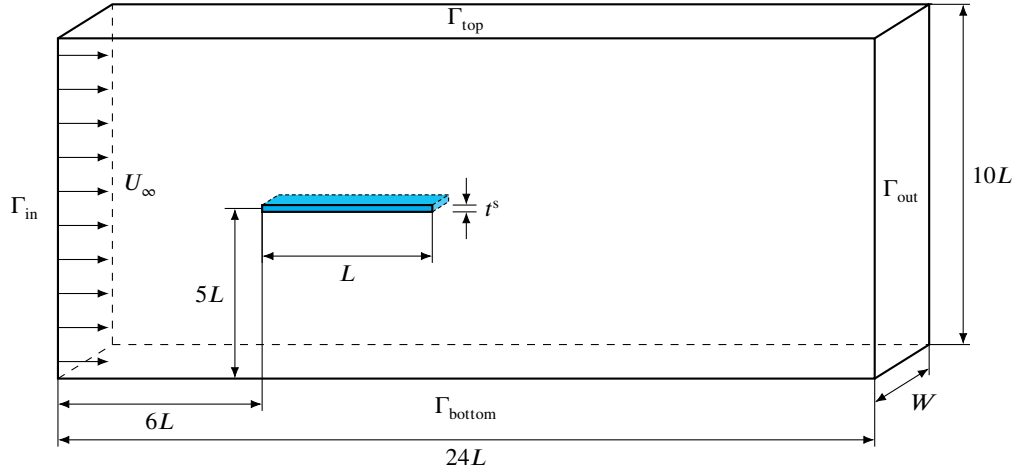
\begin{figure}[pos=H]
	\centering
\def\fpDomL{24}     
\def\fpDomH{10}     
\def\fpDdx{1.6}     
\def\fpDdy{1.0}     
\def\fpExt{0.3}     
\def\fpXc{6}        
\def\fpYc{5}        
\def\fpPlL{5}       
\def\fpPlT{0.2}     
\def\fpInLen{1.5}   
\scalebox{1.0}{
\begin{tikzpicture}[scale=0.45, every node/.style={font=\small},
	depth/.style={dash pattern=on 1.5pt off 1pt},
	dim/.style={{Latex[length=1.5mm]}-{Latex[length=1.5mm]}}]

	\pgfmathsetmacro{\fpXt}{\fpXc+\fpPlL}   

	\coordinate (A)  at (0,0);
	\coordinate (B)  at (\fpDomL,0);
	\coordinate (C)  at (\fpDomL,\fpDomH);
	\coordinate (Ff) at (0,\fpDomH);
	\coordinate (Ap) at (\fpDdx,\fpDdy);
	\coordinate (Bp) at (\fpDomL+\fpDdx,\fpDdy);
	\coordinate (Cp) at (\fpDomL+\fpDdx,\fpDomH+\fpDdy);
	\coordinate (Fp) at (\fpDdx,\fpDomH+\fpDdy);

	\draw[dashed, thin] (Ap) -- (Bp);
	\draw[dashed, thin] (A) -- (Ap);
	\draw[dashed, thin] (Ap) -- (Fp);

	\draw[thick] (A) -- (B) -- (C) -- (Ff) -- cycle;
	\draw[thick] (Ff) -- (Fp) -- (Cp) -- (C);
	\draw[thick] (B) -- (Bp) -- (Cp);

	\node[above] at (\fpDomL/2, \fpDomH-0.1) {$\Gamma_\mathrm{top}$};
	\node[below] at (\fpDomL/2, -0.3) {$\Gamma_\mathrm{bottom}$};
	\node[left]  at (-0.3, \fpDomH/2) {$\Gamma_\mathrm{in}$};
	\node at (\fpDomL+\fpDdx/2, \fpDomH/2+\fpDdy/2) {$\Gamma_\mathrm{out}$}; 

	\foreach \y in {0.5,1.5,...,9.5}{
		\draw[-{Latex[length=1.5mm]}] (0,\y) -- (\fpInLen,\y);
	}
	\node[above right] at (\fpInLen+0.2, \fpDomH/2+0.1) {$U_\infty$};

	\coordinate (P1) at (\fpXc,\fpYc-\fpPlT/2);
	\coordinate (P2) at (\fpXt,\fpYc-\fpPlT/2);
	\coordinate (P3) at (\fpXt,\fpYc+\fpPlT/2);
	\coordinate (P4) at (\fpXc,\fpYc+\fpPlT/2);
	\coordinate (P2p) at ($(P2)+(\fpExt*\fpDdx,\fpExt*\fpDdy)$);
	\coordinate (P3p) at ($(P3)+(\fpExt*\fpDdx,\fpExt*\fpDdy)$);
	\coordinate (P4p) at ($(P4)+(\fpExt*\fpDdx,\fpExt*\fpDdy)$);
	\fill[cyan] (P1) -- (P2) -- (P3) -- (P4) -- cycle;                 
	\fill[cyan!70] (P4) -- (P3) -- (P3p) -- (P4p) -- cycle;             
	\fill[cyan!55] (P2) -- (P3) -- (P3p) -- (P2p) -- cycle;             
	\draw[thick] (P1) -- (P2) -- (P3) -- (P4) -- cycle;                 
	\draw[depth] (P4p) -- (P3p) -- (P2p);                               
	\draw[depth] (P4) -- (P4p);                                         
	\draw[depth] (P3) -- (P3p);
	\draw[depth] (P2) -- (P2p);

	\draw[thin] (P1) -- ++(0,-1.5);
	\draw[thin] (P2) -- ++(0,-1.5);
	\draw[dim] (\fpXc,\fpYc-\fpPlT/2-1.3) -- (\fpXt,\fpYc-\fpPlT/2-1.3);
	\node[below] at (\fpXc+\fpPlL/2, \fpYc-\fpPlT/2-1.3) {$L$};

	\draw[thin] (\fpXt+0.6,\fpYc+\fpPlT/2) -- (\fpXt+1.3,\fpYc+\fpPlT/2);
	\draw[thin] (\fpXt+0.6,\fpYc-\fpPlT/2) -- (\fpXt+1.3,\fpYc-\fpPlT/2);
	\draw[-{Latex[length=1.5mm]}] (\fpXt+1.0,\fpYc+\fpPlT/2+0.8) -- (\fpXt+1.0,\fpYc+\fpPlT/2);
	\draw[-{Latex[length=1.5mm]}] (\fpXt+1.0,\fpYc-\fpPlT/2-0.8) -- (\fpXt+1.0,\fpYc-\fpPlT/2);
	\node[right] at (\fpXt+1.3, \fpYc) {$t^{\mathrm{s}}$};

	\draw[thin] (\fpXc,\fpYc) -- (\fpXc-1.1,\fpYc);
	\draw[dim] (\fpXc-0.8,0) -- (\fpXc-0.8,\fpYc);
	\node[left] at (\fpXc-0.85, \fpYc/2) {$\fpYc L$};

	\draw[thin] (P1) -- (\fpXc,-0.95);
	\draw[dim] (0,-0.75) -- (\fpXc,-0.75);
	\node[below] at (\fpXc/2, -0.75) {$\fpXc L$};

	\draw[thin] (0,0) -- (0,-2.0);                                    
	\draw[thin] (\fpDomL,0) -- (\fpDomL,-2.0);                 
	\draw[dim] (0,-1.8) -- (\fpDomL,-1.8);
	\node[below] at (\fpDomL/2, -1.8) {$\fpDomL L$};

	\draw[thin] (Bp) -- ++(1.2,0);                                      
	\draw[thin] (Cp) -- ++(1.2,0);
	\draw[dim] (\fpDomL+\fpDdx+1.0,\fpDdy) -- (\fpDomL+\fpDdx+1.0,\fpDdy+\fpDomH);
	\node[right] at (\fpDomL+\fpDdx+1.0, \fpDomH*0.5+\fpDdy) {$\fpDomH L$};

	\draw[thin] (Bp) -- ++(0,-0.9);
	\draw[dim] ($(B)+(0,-0.7)$) -- ($(Bp)+(0,-0.7)$);
	\node[below right] at ($(B)!0.5!(Bp)+(0,-0.7)$) {$W$};

\end{tikzpicture}}
\caption{Schematic and computational domain of the elastically deforming thin flexible plate in a free flow.}
\label{fig:schematicFlexPlate}
\end{figure}
	
	Fig.\ref{fig:schematicFlexPlate} illustrates the domain and boundary conditions
	associated with the FSI problem. The structural domain $\Omega^\mathrm{s}$
	consists of the thin flexible plate of length $L$ with the leading edge clamped at $(-L,0)$, and thickness $t^\mathrm{s} = 0.01L$. With the leading edge clamped, the other faces of the plate are treated as part of the fluid-structure interface, which is in contact with the fluid domain. The fluid domain is rectangular with dimensions $[-4L, 20L]\times[-5L, 5L]$, with the channel walls located symmetrically about the flexible plate and denoted by $\Gamma^\mathrm{f}_\mathrm{top}$ and $\Gamma^\mathrm{f}_\mathrm{bottom}$ respectively. On the side walls of the channel, a symmetric boundary condition is applied. The inlet is located at $x = -4L$, where a freestream velocity $\boldsymbol{u}^\mathrm{f} = (U_\infty, 0)$ is prescribed. A traction-free outflow
	boundary condition is applied at the channel exit ($\Gamma_\mathrm{out}$). For all results in the current section, the dimensional quantities are $L = 0.1\text{ m}$, $\rho^\mathrm{f} = 1000\text{ kgm}^{-3}$, and the transverse width of the plate is $b = 0.1\text{ m}$. 

    \begin{figure}[pos=H]
		\centering
		\begin{subfigure}[b]{0.49\linewidth}
			\centering
			\includegraphics[trim={2cm 1.5cm 0.0cm 0cm}, clip,width=\textwidth]{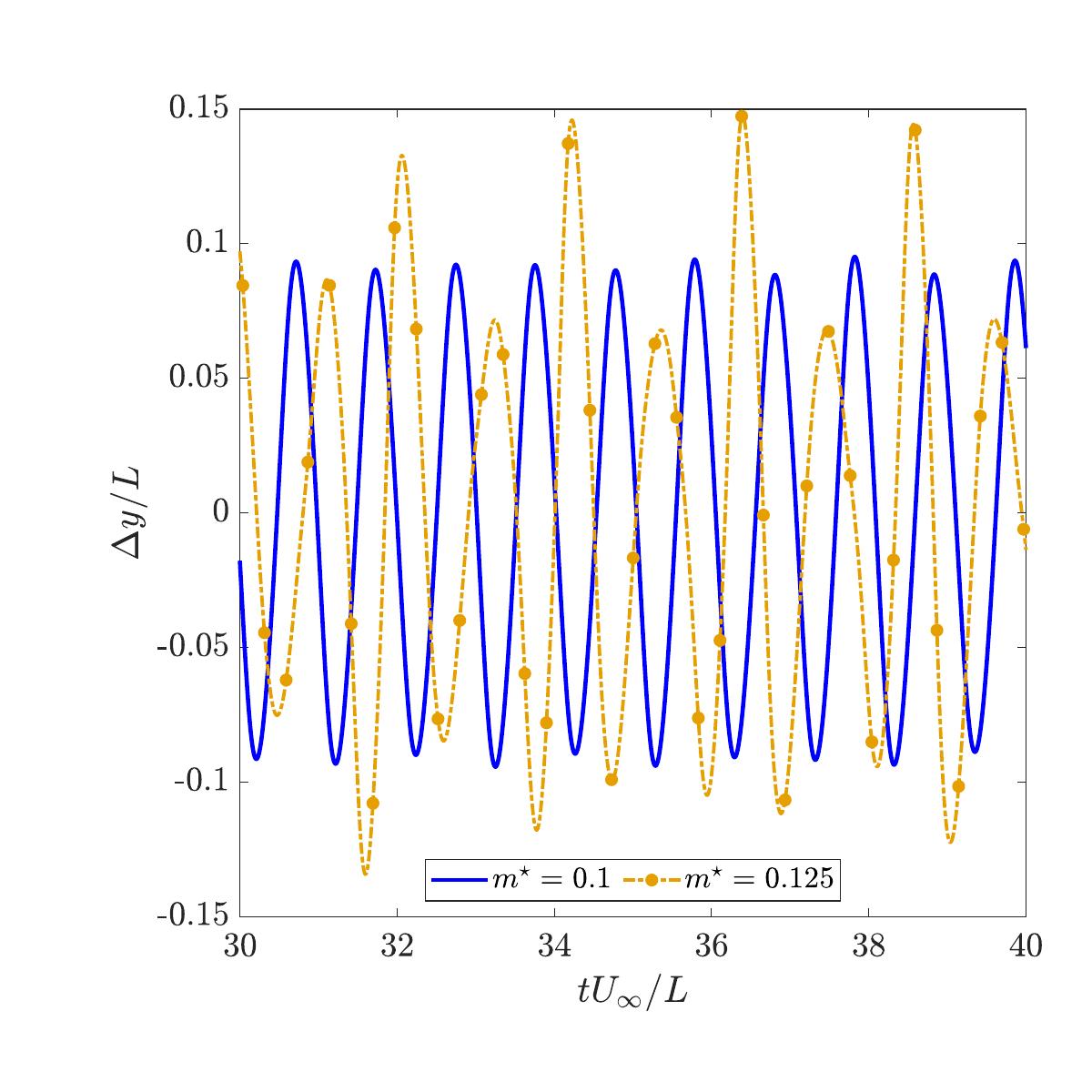}
			\caption{}
			\label{fig:itersMstarflag}
		\end{subfigure}
		\hfill
		\begin{subfigure}[b]{0.49\linewidth}
			\centering
			\includegraphics[trim={2cm 1.5cm 0.0cm 0cm}, clip,width=\textwidth]{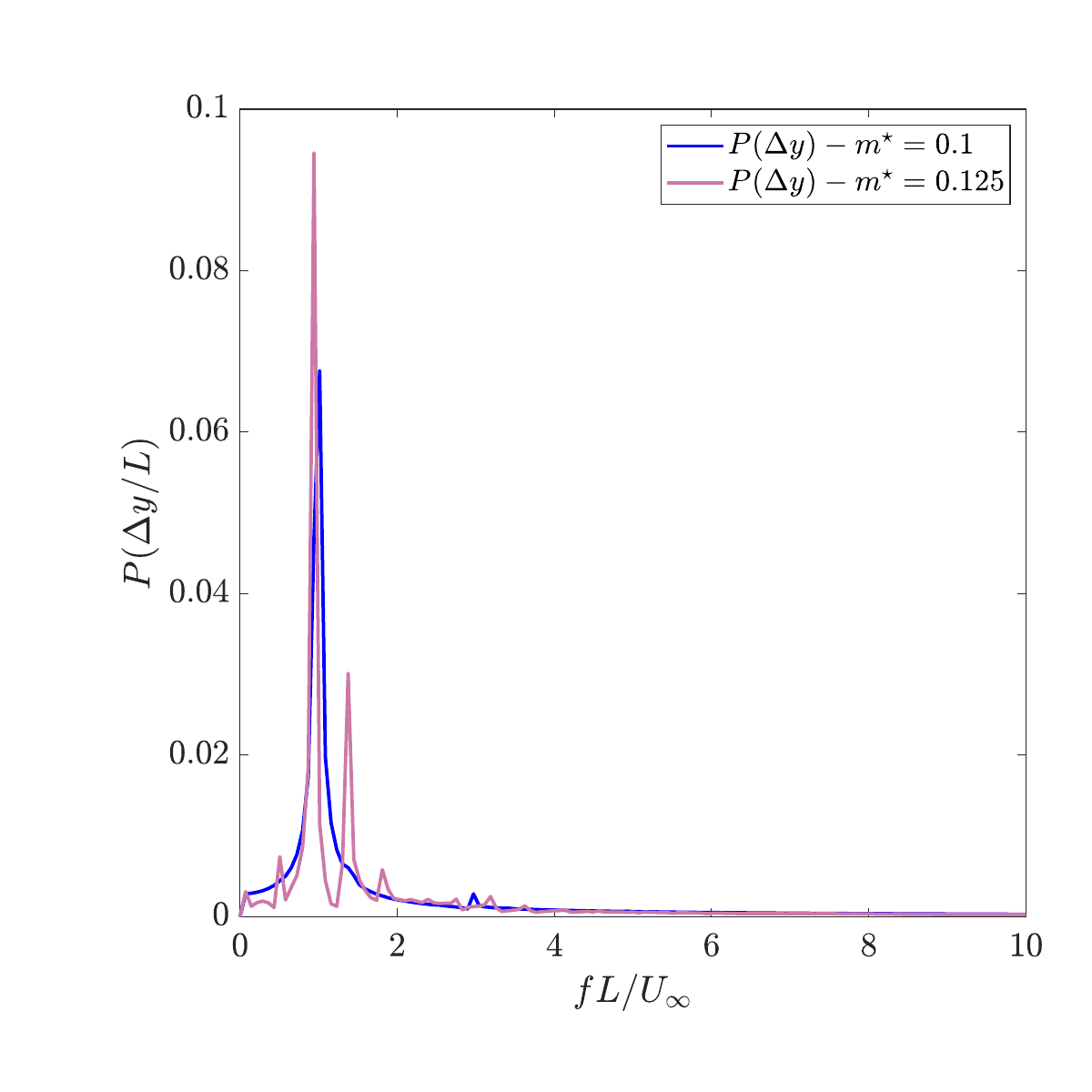}
			\caption{}
			\label{fig:rhoMstarflag}
		\end{subfigure}
		\caption{(a) Temporal variation of the cross-stream displacement of the trailing edge of the slender plate. (b) Frequency spectrum of the cross-stream displacement at $Re= 1000$ for $K_B = 10^{-4}$ and varying the mass ratio $m^\star$ }
		\label{fig:iterdata2}
	\end{figure}

    The effects of the mass ratio on the flapping dynamics of the slender plate are illustrated at $Re = 1000$ through two different values of $m^\star$. The chosen mass ratios correspond to the unsteady flapping response regime \citep{LIU2014687}. As illustrated in Fig.\ref{fig:itersMstarflag}, for $m^\star = 0.1$ regular sinusoidal oscillations are observed, wherein the dominant frequency is synchronized with the flow characteristic frequency. In contrast, the response observed when $m^\star = 0.125$ corresponds to the chaotic flapping regime. While the dominant frequency synchronizes with the flow characteristic frequency, a range of higher harmonics are excited. 
    
    As illustrated in Fig. \ref{fig:timeseries}, the slender plate exhibits a distinct multimodal response, which is reflected in the von K\'arm\'an vortices shed into the wake. In the first half of the cycle, $L/6$ of the plate in proximity to the trailing edge deforms, and the detachment of the shear layer results in a clockwise vortex being shed. Consequently, the deformed section of the plate rapidly moves downward, where the pairing counter-clockwise vortex is shed. In the second half of the cycle the trailing $2L/3$ section of the plate is excited wherein a travelling wave appears to propagate along the plate resulting in the detachment of a highly elongated shear layer. Thus the multimodal motion of the plate can be tied to these two distinct mechanisms of vortex shedding. Further analysis pertaining to the addded-mass coefficient and correlation to the wetted structural modes is necessary to characterize the origin of the higher harmonics, and is beyond the scope of the current work. A comprehensive review of flapping bodies interacting with fluid flows can be found in \citep{2011_ARFM}.

    \begin{figure}[pos=H]
		\centering
		\includegraphics[width=\linewidth]{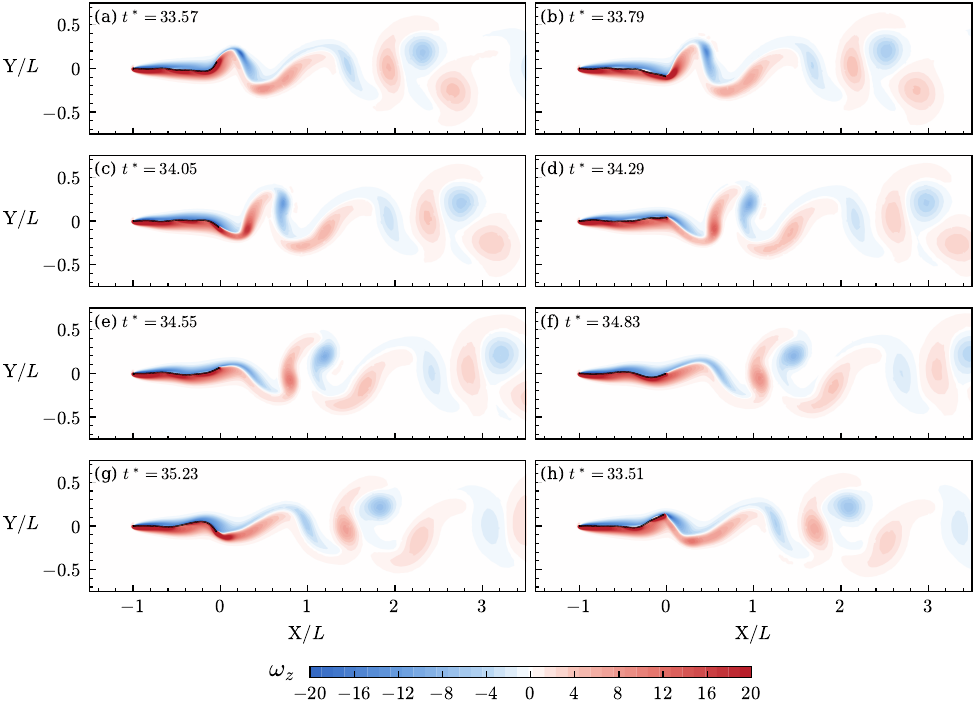}
		\hfill
		\caption{Temporal evolution of the von K\'arm\'an vortices in the wake of the elastically deforming thin flexible plate for $Re = 1000$, $m^\star = 0.125$ and $K_B = 10^{-4}$.  }
		\label{fig:timeseries}
	\end{figure}

\section{Conclusions}\label{sec:conclusion}
In this work, a momentum-weighted interface stabilization (MWIS) method is developed for Dirichlet-Neumann partitioned fluid-structure interaction (FSI) schemes to mitigate added-mass-induced instabilities. An eigenvalue analysis of a simplified linear FSI problem reveals that the instability originates from the omission of the structural interface response in the partitioned fluid subproblem and demonstrates how an appropriate interface augmentation can improve convergence. The analysis is subsequently extended to the general linearized FSI problem, where a comparison of the interface pressure Schur complements of the monolithic and partitioned systems identifies the source of the instability and establishes spectral criteria for the augmentation operator. The augmented scheme is consistent by construction, as the stabilization contribution vanishes at the coupled fixed point, recovering the monolithic solution. Guided by this analysis, a diagonal approximation of the momentum-weighted stabilization operator is developed using fluid and structural interface inertia and element metrics. The operator has a memory footprint comparable to the data exchanged at the interface, enabling implementation with limited additional storage and minimal modifications to existing fluid and structural solvers. For the simplified linear model, the idealized augmentation predicts single-iteration convergence in the inertia-dominated regime under the stated assumptions. Numerical FSI benchmarks demonstrate improved coupling robustness at low structural-to-fluid mass ratios, where conventional Dirichlet-Neumann iterations may diverge. Future work will examine the stability and accuracy of MWIS for nonmatching interface discretizations, three-dimensional configurations, and strongly nonlinear FSI problems. Further opportunities include the design of efficient interface preconditioners and adaptive stabilization operators to reduce coupling iterations while preserving consistency with the monolithic formulation.

\section*{Acknowledgments}
The authors gratefully acknowledge the financial support of the Natural Sciences and Engineering Research Council of Canada (NSERC) through the Alliance program, Mitacs, and ANSYS Canada. Computational resources were provided by the Shared Hierarchical Academic Research Computing Network (SHARCNET) and the Digital Research Alliance of Canada (formerly Compute Canada).

\printcredits

\section*{Competing interests}
The authors declare no competing financial interests.
    
	\appendix
	\section{Appendix A. Variational Formulation and Linearization Scheme}
	\label{app1}
\subsection{Governing equations}
\label{sec:app_governing}
As detailed in Sec.~\ref{sec:simplifiedFSI}, the fluid occupies a physical domain $\Omega^\mathrm{f}(\mathbf{x}^\mathrm{f},t)$ confined within the fluid boundary $\Gamma^\mathrm{f}(t)$, where $\mathbf{x}^\mathrm{f}$ and $t$ denote the spatial and temporal coordinates respectively. A Newtonian fluid with density $\rho^\mathrm{f}$ and dynamic viscosity $\mu^\mathrm{f}$ is assumed to occupy the fluid domain, whose motion is governed by the incompressible Navier--Stokes equations. In an Arbitrary Lagrangian-Eulerian (ALE) frame, the equations are as follows:
\begin{align}
	\rho^\mathrm{f}\left(\left.\frac{\partial\boldsymbol{u}^\mathrm{f}}{\partial t}\right|_{\hat{\mathbf{x}}} + \left(\boldsymbol{u}^\mathrm{f}-\boldsymbol{w}^\mathrm{f}\right)\cdot\nabla\boldsymbol{u}^\mathrm{f}\right) &= \nabla\cdot\boldsymbol{\sigma}^\mathrm{f} + \boldsymbol{f}^\mathrm{f}, && \text{in }\Omega^\mathrm{f}(t), \label{eq:app_fluidmom}\\
	\nabla\cdot\boldsymbol{u}^\mathrm{f} &= 0, && \text{in }\Omega^\mathrm{f}(t), \label{eq:app_fluidcont}
\end{align} where $\hat{\mathbf{x}}$ denotes the referential coordinate system, $\boldsymbol{u}^\mathrm{f} = \boldsymbol{u}^\mathrm{f}(\mathbf{x}^\mathrm{f},t)$ is the fluid velocity at each spatial location $\mathbf{x}^\mathrm{f} \in \Omega^\mathrm{f}$ and $\boldsymbol{w}^\mathrm{f}$ is the relative velocity of the spatial coordinates $\mathbf{x}^\mathrm{f}$ with respect to the referential coordinate system $\hat{\mathbf{x}}$. $\boldsymbol{f}^\mathrm{f}$ is the body force applied on the fluid, and $\boldsymbol{\sigma}^\mathrm{f}$ is the Cauchy stress tensor defined for a Newtonian fluid as
\begin{equation}
	\boldsymbol{\sigma}^\mathrm{f}=-p^\mathrm{f}\boldsymbol{I}+\mu^\mathrm{f}\left(\nabla\boldsymbol{u}^\mathrm{f}+(\nabla\boldsymbol{u}^\mathrm{f})^T\right)
\end{equation} where $p^\mathrm{f}$ denotes the fluid pressure. 
Similarly, the structure occupies the reference domain $\Omega^\mathrm{s}(\mathbf{X}^\mathrm{s})$, bounded within the structural boundary $\Gamma^\mathrm{s}$. The deformation map $\boldsymbol{\varphi}^\mathrm{s}(\mathbf{X}^\mathrm{s},t)$ expresses the deformation history of a material point $\mathbf{X}^\mathrm{s} \in \Omega^\mathrm{s}$. As a consequence the structural velocity can be expressed as
\begin{equation}
	\boldsymbol{u}^\mathrm{s}(\mathbf{X}^\mathrm{s},t) = \frac{\partial \boldsymbol{\varphi}^\mathrm{s}(\mathbf{X}^\mathrm{s},t)}{\partial t},
\end{equation} and its evolution is governed by the momentum equation
\begin{equation}
	\rho^\mathrm{s}\frac{\partial\boldsymbol{u}^\mathrm{s}}{\partial t} = \nabla\cdot\boldsymbol{\sigma}^\mathrm{s} + \boldsymbol{f}^\mathrm{s} \quad \text{in }\Omega^\mathrm{s}, \label{eq:app_solidmom}
\end{equation} where $\rho^\mathrm{s}$ is the density of the material constituting the structure, $\boldsymbol{\sigma}^\mathrm{s}$ is the first Piola-Kirchhoff stress tensor and $\boldsymbol{f}^\mathrm{s}$ is the body force per unit reference volume \citep{LIU2014687}. The displacement of a material point $\mathbf{X}^\mathrm{s}$ is defined as
\begin{equation}
	\boldsymbol{\eta}^\mathrm{s}(\mathbf{X}^\mathrm{s},t)=\boldsymbol{\varphi}^\mathrm{s}(\mathbf{X}^\mathrm{s},t)-\mathbf{X}^\mathrm{s},
\end{equation} using which infinitesmal strain tensor $\boldsymbol{\epsilon}^\mathrm{s}$, deformation gradient tensor $\mathbf{F}$, and the Green-Lagrange strain tensor $\mathbf{E}$ are respectively defined as
\begin{align}
	\boldsymbol{\epsilon}^\mathrm{s} = \frac12(\nabla \boldsymbol{\eta}^\mathrm{s} + (\nabla \boldsymbol{\eta}^\mathrm{s})^T),\quad
	\mathbf{F} &= \nabla \boldsymbol{\varphi}^\mathrm{s},\quad
	\mathbf{E} = \frac12 \left(\mathbf{F}\mathbf{F}^T - \mathbf{I}\right).
\end{align} For a linear elastic material, the stress tensor takes the form 
\begin{equation}
	\boldsymbol{\sigma}^\mathrm{s}(\boldsymbol{\varphi}^\mathrm{s}) = \mu^\mathrm{s}\left(\nabla\boldsymbol{\eta}^\mathrm{s}+(\nabla\boldsymbol{\eta}^\mathrm{s})^T\right) + \lambda^\mathrm{s}\left(\nabla\cdot\boldsymbol{\eta}^\mathrm{s}\right)\boldsymbol{I}, \label{eq:app_linelastic}
\end{equation} where $\lambda^\mathrm{s}$ and $\mu^\mathrm{s}$ are the Lam\'{e}'s parameters for the material constituting the structure. For a St. Venant-Kirchhoff material, the stress tensor is 
\begin{equation}
	\boldsymbol{\sigma}^\mathrm{s}(\boldsymbol{\varphi}^\mathrm{s}) = 2\mu^\mathrm{s}\boldsymbol{F}\boldsymbol{E} + \lambda^\mathrm{s}\left(\mathrm{tr}\,\boldsymbol{E}\right)\boldsymbol{F}. \label{eq:app_stvk}
\end{equation}
At $t = 0$, the fluid-structure interface is defined as $\Gamma^\mathrm{fs}_0 = \Gamma^\mathrm{f}(0) \cap \Gamma^\mathrm{s}$. At time $t$, the interface is located using the deformation map as $\Gamma^\mathrm{fs} = \boldsymbol{\varphi}^\mathrm{s}(\Gamma^\mathrm{fs}_0,t)$. The velocity and traction conditions couple the FSI system on $\Gamma^\mathrm{fs}$ as
\begin{align}
	\boldsymbol{u}^\mathrm{f}(\boldsymbol{\varphi}^\mathrm{s}(\mathbf{z},t),t) &= \boldsymbol{u}^\mathrm{s}(\mathbf{z},t), && \forall\,\mathbf{z}\in\Gamma^\mathrm{fs}_0, \label{eq:app_kin}\\
	\int_{\boldsymbol{\varphi}^\mathrm{s}(\gamma,t)}\boldsymbol{\sigma}^\mathrm{f}(\mathbf{x}^\mathrm{f},t)\,\mathbf{n}^\mathrm{f}\,dS(\mathbf{x}^\mathrm{f}) &=- \int_{\gamma}\boldsymbol{\sigma}^\mathrm{s}(\mathbf{z},t)\,\mathbf{n}^\mathrm{s}\,dS(\mathbf{z}), && \forall\,\gamma\subset\Gamma^\mathrm{fs}_0. \label{eq:app_dyn}
\end{align} where $\mathbf{n}^\mathrm{s}$ and $\mathbf{n}^\mathrm{f}$ are the unit outward normal vectors on $\gamma$ and $\boldsymbol{\varphi}^\mathrm{s}(\gamma,t)$, respectively. $\gamma$ is a part of the undeformed fluid-structure interface $\Gamma^\mathrm{fs}_0$, and $dS$ is a differential surface area. On the boundaries $\Gamma^\mathrm{f}(t)\ \Gamma^\mathrm{fs}$ and $\Gamma^\mathrm{s}\ \Gamma^\mathrm{fs}_0$, appropriate boundary condtions are applied on the fluid and structure boundaries to render the problem well posed. 

The deformation of the fluid domain $\Omega^\mathrm{f}(t)$ required by the ALE description is governed by modeling the coordinates $\boldsymbol{x}^\mathrm{f}$ to deform as an elastic material in equilibrium equation \citep{LIU2014687}
\begin{equation}
	\nabla\cdot\boldsymbol{\sigma}^\mathrm{m} = \boldsymbol{0}, \qquad \boldsymbol{\sigma}^\mathrm{m} = (1+\tau_\mathrm{m})\left[2\mu^\mathrm{m}\left(\nabla\boldsymbol{\eta}^\mathrm{f}+(\nabla\boldsymbol{\eta}^\mathrm{f})^T\right) + \lambda^\mathrm{m}\left(\nabla\cdot\boldsymbol{\eta}^\mathrm{f}\right)\boldsymbol{I}\right], \label{eq:app_meshelastic}
\end{equation}
subject to
\begin{equation}
	\boldsymbol{\eta}^\mathrm{f} = \boldsymbol{\varphi}^\mathrm{s}(\mathbf{z},t)-\mathbf{z} \ \text{ on }\Gamma^\mathrm{fs}_0, \qquad \boldsymbol{\eta}^\mathrm{f} = \boldsymbol{0} \ \text{ on }\partial\Omega^\mathrm{f}_0\setminus\Gamma^\mathrm{fs}, \label{eq:app_meshbc}
\end{equation}
where $\boldsymbol{\sigma}^\mathrm{m}$ is the stress experienced by the fluid mesh, $\boldsymbol{\eta}^\mathrm{f}$ is the displacement vector of the fluid spatial coordinates, $(\lambda^\mathrm{m}, \mu^\mathrm{m})$ are the mesh Lam\'{e}'s parameters, and $\tau_\mathrm{m}$ is a local stiffness parameter used to limit the distortion of small elements \citep{Stein_2002_MeshMotion}. Solving \eqref{eq:app_meshelastic}--\eqref{eq:app_meshbc} in each time step yields the deformed configuration $\Omega^\mathrm{f}(t)$ and the mesh velocity $\boldsymbol{w}^\mathrm{f}=\partial_t\boldsymbol{\eta}^\mathrm{f}$ that appears in \eqref{eq:app_fluidmom}.
	
\subsection{Variational formulation}
\label{sec:app_weak}
\subsubsection{Galerkin variational formulation for structural dynamics}
The structural domain $\Omega^\mathrm{s}$ is discretized into $n_\mathrm{el}^\mathrm{S}$ non-overlapping finite elements $\Omega_e^\mathrm{s}$. For the current study, the HHT-$\alpha$ method \citep{Hilber1977} is used to evolve the temporal dynamics of the structure. The structural acceleration $\boldsymbol{a}_h^{\mathrm{s},n+1}$, and displacement $\boldsymbol{\eta}_h^{\mathrm{s},n+1}$ at $t^{n+1}$ are expressed in terms of the structural velocity $\boldsymbol{u}_h^{\mathrm{s},n+1}$ and the state at $t^n$ through the Newmark-$\beta$ relations as 
\begin{align}
	\boldsymbol{a}_h^{\mathrm{s},n+1} &= \frac{\boldsymbol{u}_h^{\mathrm{s},n+1}-\boldsymbol{u}_h^{\mathrm{s},n}}{\gamma\,\Delta t} - \frac{1-\gamma}{\gamma}\,\boldsymbol{a}_h^{\mathrm{s},n}, \label{eq:app_newmark_a}\\
	\boldsymbol{\eta}_h^{\mathrm{s},n+1} &= \boldsymbol{\eta}_h^{\mathrm{s},n} + \Delta t\left(1-\frac{\beta}{\gamma}\right)\boldsymbol{u}_h^{\mathrm{s},n} + \frac{\beta\Delta t}{\gamma}\,\boldsymbol{u}_h^{\mathrm{s},n+1} + \Delta t^2\left(\frac{1}{2}-\frac{\beta}{\gamma}\right)\boldsymbol{a}_h^{\mathrm{s},n}, \label{eq:app_newmark_eta}
\end{align} where $\gamma=\tfrac12-\alpha_\mathrm{H}$, $\beta=\tfrac14(1-\alpha_\mathrm{H})^2$, and $\alpha_\mathrm{H}\in[-1/3,0]$ is the HHT-$\alpha$ parameter which controls the numerical dissipation of the high-frequency structural response.  $\boldsymbol{\eta}_h^{\mathrm{s},n}, \boldsymbol{u}_h^{\mathrm{s},n}, \boldsymbol{a}_h^{\mathrm{s},n}$ are the structural state variables known from the previous time step. Given the test and trial spaces $\mathcal{P}_h$ and $\mathcal{Q}_h$ respectively defined on $\Omega^\mathrm{s}$, the variational form of the structural equations is expressed as: Find $\boldsymbol{u}_h^{\mathrm{s},n+1} \in \mathcal{Q}_h$ such that $\forall~\boldsymbol{v}^\mathrm{s}_h\in\mathcal{P}_h$
\begin{equation}
	\int_{\Omega^\mathrm{s}}\boldsymbol{v}^\mathrm{s}_h\cdot\rho^\mathrm{s}\,\boldsymbol{a}_h^{\mathrm{s},n+1}\,d\Omega
	+ \int_{\Omega^\mathrm{s}}\nabla\boldsymbol{v}^\mathrm{s}_h:\boldsymbol{\sigma}^\mathrm{s}\!\left(\boldsymbol{\eta}_h^{\mathrm{s},n+1}\right)\,d\Omega
	= \int_{\Omega^\mathrm{s}}\boldsymbol{v}^\mathrm{s}_h\cdot\boldsymbol{f}^{\mathrm{s},n+1}\,d\Omega
	+ \int_{\Gamma^\mathrm{s}_\mathrm{N}}\boldsymbol{v}^\mathrm{s}_h\cdot\mathbf{t}^{\mathrm{s},n+1}\,dS
	+ \int_{\Gamma^\mathrm{fs}_0}\boldsymbol{v}^\mathrm{s}_h\cdot\mathbf{t}_\Sigma^{n+1}\,dS.
	\label{eq:app_solidsemi}
\end{equation} The traction $\mathbf{t}_\Sigma^{n+1}$ is given by \eqref{eq:app_dyn}, and considers the fluid force exerted on the structure within a partitioned framework. Furthermore, the deformation response $\boldsymbol{\eta}_h^\mathrm{s}$ of the structure governs the deformation of the fluid domain $\Omega^\mathrm{f}(t)$. 

\subsubsection{Petrov-Galerkin variational formulation for fluid flow}
Through the finite element method, the fluid spatial domain $\Omega^\mathrm{f}$ is discretized into $n_\mathrm{el}^\mathrm{F}$ non-overlapping finite elements $\Omega_e^\mathrm{f}$. For the current study, the second order Backward-Difference (BDF2) temporal discretization is adopted. Furthermore, given the test and trial spaces $\mathcal{V}_h$ and $\mathcal{S}_h$ respectively defined on $\Omega^\mathrm{f}$, the variational form of the flow equations is expressed as: Find $[\boldsymbol{u}^{\mathrm{f},n+1}_h, p^{\mathrm{f},n+1}_h] \in \mathcal{S}_h$ such that $\forall~[\boldsymbol{v}^\mathrm{f}, q] \in \mathcal{V}_h$
\begin{equation}
	\begin{split}
		&\int_{\Omega^{\mathrm{f}}} \boldsymbol{v}^\mathrm{f}\cdot\rho^\mathrm{f}\left(\frac{3\boldsymbol{u}_h^{\mathrm{f},n+1}-4\boldsymbol{u}_h^{\mathrm{f},n}+\boldsymbol{u}_h^{\mathrm{f},n-1}}{2\Delta t} + \left(\boldsymbol{u}_h^{\mathrm{f},n+1}-\boldsymbol{w}_h^\mathrm{f}\right)\cdot\nabla\boldsymbol{u}_h^{\mathrm{f},n+1}\right)\,d\Omega \\
		&\quad + \int_{\Omega^{\mathrm{f}}} \nabla\boldsymbol{v}^\mathrm{f}:\boldsymbol{\sigma}^\mathrm{f}\,d\Omega
		+ \int_{\Omega^{\mathrm{f}}} q\,\left(\nabla\cdot\boldsymbol{u}_h^{\mathrm{f},n+1}\right)\,d\Omega \\
		&\quad + \sum_{e=1}^{n_\mathrm{el}^\mathrm{F}}\int_{\Omega_e^\mathrm{f}}\tau_\mathrm{m}\left(\rho^\mathrm{f}\left((\boldsymbol{u}_h^{\mathrm{f},n+1}-\boldsymbol{w}_h^\mathrm{f})\cdot\nabla\right)\boldsymbol{v}^\mathrm{f} + \nabla q\right)\cdot\boldsymbol{\mathcal{R}}_\mathrm{m}\,d\Omega
		+ \sum_{e=1}^{n_\mathrm{el}^\mathrm{F}}\int_{\Omega_e^\mathrm{f}} \tau_\mathrm{c} \left(\nabla \cdot \boldsymbol{v}^\mathrm{f}\right) \mathcal{R}_\mathrm{c}\,d\Omega \\
		&\;=\; \int_{\Omega^\mathrm{f}}\boldsymbol{v}^\mathrm{f}\cdot\boldsymbol{f}^\mathrm{f}\,d\Omega
		+ \int_{\Gamma^\mathrm{f}_\mathrm{N}}\boldsymbol{v}^\mathrm{f}\cdot\left(\boldsymbol{\sigma}^\mathrm{f}\mathbf{n}^\mathrm{f}\right)\,dS,
	\end{split}
	\label{eq:app_fluidweak}
\end{equation} where $\boldsymbol{\mathcal{R}}_\mathrm{m}$ and $\mathcal{R}_c$ are the momentum and continuity residuals defined as
\begin{align}
	\boldsymbol{\mathcal{R}}_\mathrm{m} &= \rho^\mathrm{f}\left(\frac{3\boldsymbol{u}_h^{\mathrm{f},n+1}-4\boldsymbol{u}_h^{\mathrm{f},n}+\boldsymbol{u}_h^{\mathrm{f},n-1}}{2\Delta t}+\left(\boldsymbol{u}_h^{\mathrm{f},n+1}-\boldsymbol{w}_h^\mathrm{f}\right)\cdot\nabla\boldsymbol{u}_h^{\mathrm{f},n+1}\right) - \nabla\cdot\boldsymbol{\sigma}^\mathrm{f}\!\left(\boldsymbol{u}_h^{\mathrm{f},n+1},p_h^{\mathrm{f},n+1}\right) - \boldsymbol{f}^\mathrm{f}, \label{eq:app_Rm}\\
	\mathcal{R}_\mathrm{c} &= \nabla\cdot\boldsymbol{u}_h^{\mathrm{f},n+1}.   \label{eq:app_Rc}
\end{align} The stabilization parameters $\tau_m$ and $\tau_c$ are the least-squares metrics which for the momentum equation are defined as 
\begin{align}
	\tau_m = \left[\left(\frac{2\rho^\mathrm{f}}{\Delta t}\right)^2 + {(\rho^\mathrm{f})}^2(\boldsymbol{u}^\mathrm{f}_h - \boldsymbol{w}^\mathrm{f}_h)\cdot \boldsymbol{G}(\boldsymbol{u}^\mathrm{f}_h - \boldsymbol{w}^\mathrm{f}_h) + C_I(\mu^\mathrm{f})^2\boldsymbol{G}:\boldsymbol{G} \right]^{-\frac12}
\end{align} where $C_I$ is the element-wise inverse estimate and $\boldsymbol{G}$ is the size of element contravariant metric tensor as defined in \eqref{eq:metricTensor}. The least-squares metric $\tau_c$ is defined as
\begin{equation}
	\tau_c = (\text{tr}(\boldsymbol{G})\tau_m)^{-1}, 
\end{equation} where $\text{tr}(\boldsymbol{G})$ denotes the trace of the contravariant metric tensor. Within the partitioned fluid-structure interaction framework, the structural velocity is imposed as a Dirichlet boundary condition on the interface $\Gamma^\mathrm{fs}$.

\subsection{Constructing the linearized FSI problem} \label{sec:app_linearization}
The variational forms \eqref{eq:app_solidsemi} and \eqref{eq:app_fluidweak} are
nonlinear in their respective unknowns, and a Newton--Raphson technique is employed
to iteratively solve the individual subproblems in the partitioned framework. At the start 
of the $k^\mathrm{th}$ nonlinear iteration, let the state of the system be defined by 
$\boldsymbol{u} = [\boldsymbol{u}^{\mathrm{s},k}_h~ \boldsymbol{u}^{\mathrm{f},k}_h~ p_h^k]^T$, with the corresponding
increments defined as $\delta \boldsymbol{u} = [\delta\boldsymbol{u}^\mathrm{s}_h~\delta\boldsymbol{u}^\mathrm{f}_h~\delta p_h]^T$. 
Furthermore, let $\boldsymbol{R}^\mathrm{s}$ and $\boldsymbol{R}^\mathrm{f}$ be the variational forms \eqref{eq:app_solidsemi}--\eqref{eq:app_fluidweak} evaluated using the current state of the FSI system. 
In order to construct the interfacial operators, as well as enforce consistency of boundary conditions in the
partitioned framework, the partitioning of the degrees of freedom is first established before linearization.

\subsubsection{Partitioning degrees of freedom}
Define the boolean restriction operators $\Pi_\mathrm{S}$, $\Pi_\Sigma^\mathrm{s}$, $\Pi_\mathrm{F}$,
$\Pi_\Sigma^\mathrm{f}$ and $\Pi_p$ which extract respectively, 
the interior and interface rows of the structural velocity vector,
the interior and interface rows of the fluid velocity vector, and the pressure rows. This leads to
the definition of the increments
\begin{equation}
	\delta \boldsymbol{u}_\mathrm{S} = \Pi_\mathrm{S}\delta \boldsymbol{u}^\mathrm{s}_h,\quad
	\delta \boldsymbol{u}_\mathrm{F} = \Pi_\mathrm{F}\delta \boldsymbol{u}^\mathrm{f}_h.\label{eq:app_restrict}
\end{equation} The kinematic condition \eqref{eq:app_kin} forms a Dirichlet boundary condition on the fluid subproblem in the
partitioned framework, and requires the interface trace
\begin{equation}
	\Pi_\Sigma^\mathrm{s}\delta\boldsymbol{u}^\mathrm{s}
	= \Pi_\Sigma^\mathrm{f}\delta\boldsymbol{u}^\mathrm{f}
	=: \delta\boldsymbol{u}_\Sigma,
	\label{eq:app_kin_disc}
\end{equation} which, within the monolithic framework allows the implicit enforcement of continuity at the fluid-structure interface.
Similarly, the balance of momentum across $\Gamma^\mathrm{fs}$ is reccovered variationally, by summing the momentum equations tested with the shared interface basis functions \citep{LIU2014687}. This allows for constructing the interface residual within the monolithic system as
\begin{equation}
	R_{m,\Sigma} = \Pi_\Sigma^\mathrm{s} \boldsymbol{R}^\mathrm{s} + \Pi_\Sigma^\mathrm{f}\boldsymbol{R}^\mathrm{f}.
\end{equation} In a partitioned framework, the term $\Pi_\Sigma^\mathrm{f}\boldsymbol{R}^\mathrm{f}$ forms the fluid traction $\mathbf{t}_\Sigma$ acting on the structure.

Given the partitioning of the variables, for a global linear operator $X$ and the partitions $\alpha, \beta \in \{S, \Sigma, F\}$, the following shorthand
\begin{equation}
	X_{\alpha\beta} = \Pi_\alpha X\,\Pi_\beta^T,
	\qquad
	X_\alpha = X_{\alpha\alpha},
	\label{eq:app_blockshorthand}
\end{equation} is defined. This convention is used to develop the interface operators in Sec.~\ref{sec:itfoperator}, and will be useful in constructing the necessary restrictions for handling the monolithic and partitioned linearized systems.
	
\subsubsection{Structural subproblem}\label{sec:app_structlin}
Using the state of the FSI system $\boldsymbol{u}^k$, the variational form of the structural governing equations \eqref{eq:app_solidsemi} is linearized. The Newmark relations \eqref{eq:app_newmark_a}--\eqref{eq:app_newmark_eta} allow the construction of the Jacobians $\partial \boldsymbol{a}^\mathrm{s}_h/\partial \boldsymbol{u}^\mathrm{s}_h$ and $\partial \boldsymbol{\eta}^\mathrm{s}_h/\partial \boldsymbol{u}^\mathrm{s}_h$, which result in the expression of \eqref{eq:app_solidsemi} in terms of the structural velocity $\boldsymbol{u}^\mathrm{s}_h$. These Jacobians lead to the time-integration scaling factors $\theta^\mathrm{s} = \gamma \Delta t$ and $\theta^\mathrm{s}_K = \beta \Delta t/\gamma$. Using the increment in structural velocity $\delta \boldsymbol{u}^\mathrm{s}$, the Newton linearization of \eqref{eq:app_solidsemi} takes the form
\begin{equation}
	\frac{\rho^\mathrm{s}}{\theta^\mathrm{s}}M^\mathrm{s}\delta\boldsymbol{u}^\mathrm{s}
	+ \theta_K^\mathrm{s}\left(
	(\lambda^\mathrm{s}+2\mu^\mathrm{s})K^\mathrm{vol}
	+ \mu^\mathrm{s}K^\mathrm{dev}\right)\delta\boldsymbol{u}^\mathrm{s}
	= {R}^\mathrm{s}(\boldsymbol{u}^k).
	\label{eq:app_solidlin}
\end{equation} where $M^\mathrm{s}$ is the mass matrix, $K^\mathrm{vol}$ and $K^\mathrm{dev}$ are the dilatation and shear parts of the tangent stiffness matrix, and $R^\mathrm{s}$ is the structural momentum residual at the current iterate. The tangent stiffness matrix is constant for the linear elastic material \eqref{eq:app_linelastic} and depends on the deformation tensor $\mathbf{F}(\boldsymbol{\eta}^\mathrm{s})$ for the St. Venant-Kirchhoff material \eqref{eq:app_stvk}. For brevity in what follows, the operator acting on $\delta\boldsymbol{u}^\mathrm{s}$
in \eqref{eq:app_solidlin} is denoted
\begin{equation}
A^\mathrm{s} = \frac{\rho^\mathrm{s}}{\theta^\mathrm{s}}M^\mathrm{s}
+ \theta_K^\mathrm{s}\left((\lambda^\mathrm{s}+2\mu^\mathrm{s})K^\mathrm{vol}
+ \mu^\mathrm{s}K^\mathrm{dev}\right),
\label{eq:app_As}
\end{equation} so that \eqref{eq:app_solidlin} reads
$A^\mathrm{s}\delta\boldsymbol{u}^\mathrm{s} = {R}^\mathrm{s}(\boldsymbol{u}^k)$.

\subsubsection{Fluid subproblem}\label{sec:app_fluidlin}
Similar to the structural subproblem, the variational form of the fluid governing equations \eqref{eq:app_fluidweak} is linearized about the current state of the FSI system $\boldsymbol{u}^\mathrm{k}$. Using the increment in fluid velocity and pressure $[\delta \boldsymbol{u}^\mathrm{f}~\delta p]^T$, the linearization of \eqref{eq:app_fluidweak} leads to 
\begin{align}
	\frac{\rho^\mathrm{f}}{\theta^\mathrm{f}}M^\mathrm{f}\delta\boldsymbol{u}^\mathrm{f}
	+ \rho^\mathrm{f}K^\mathrm{conv}\delta\boldsymbol{u}^\mathrm{f}
	+ \mu^\mathrm{f}K^\mathrm{diff}\delta\boldsymbol{u}^\mathrm{f}
	+ \left(B^T + G_\tau\right)\delta p
	&= \boldsymbol{R}^\mathrm{f}_m(\boldsymbol{u}^k),
	\label{eq:app_fluidlin_mom}\\
	B\,\delta\boldsymbol{u}^\mathrm{f} - C_\tau\,\delta p
	&= {R}_{c}(\boldsymbol{u}^k).
	\label{eq:app_fluidlin_cont}
\end{align}  Here, $\mathrm{M}^\mathrm{f}$ is the mass matrix, $K^\mathrm{conv}$ is the convection matrix,
$K^\mathrm{diff}$ is the diffusion matrix. The discrete pressure gradient operator is denoted by $B^T$, and SUPG
stabilization contributes the operator $G_\tau$ to act on the pressure field. $B$ is the discrete divergence operator,
and the PSPG stabilization acts on the pressure field within the fluid domain through the operator $C_\tau$. As in the structural subproblem, the operator acting on
$\delta\boldsymbol{u}^\mathrm{f}$ in \eqref{eq:app_fluidlin_mom} is denoted by
\begin{equation}
	A^\mathrm{f} = \frac{\rho^\mathrm{f}}{\theta^\mathrm{f}}M^\mathrm{f}
	+ \rho^\mathrm{f}K^\mathrm{conv} + \mu^\mathrm{f}K^\mathrm{diff},
	\label{eq:app_Af}
\end{equation}
so that \eqref{eq:app_fluidlin_mom} reads
$A^\mathrm{f}\delta\boldsymbol{u}^\mathrm{f} + (B^T+G_\tau)\delta p
= \boldsymbol{R}^\mathrm{f}_m(\boldsymbol{u}^k)$. The quantities
$\boldsymbol{R}^\mathrm{f}_m$ and $\boldsymbol{R}^\mathrm{f}_c$ are the momentum and continuity residuals arising from evaluating \eqref{eq:app_fluidweak} at 
the current iterate. The quantity $\theta^\mathrm{f} = 2\Delta t/3$
is a scalar associated with the BDF2 temporal discretization. More details on the construction of linearization of the fluid subproblem can be found in \citep{JAIMAN2016187, li_2018}.
\subsubsection{Assembled block system}

The kinematic condition \eqref{eq:app_kin_disc} eliminates the fluid interface
velocities as independent unknowns, so that the subproblem increments of
Sec.~\ref{sec:app_structlin} and Sec.~\ref{sec:app_fluidlin} are reconstructed
from the partitioned unknowns as
\begin{equation}
	\delta\boldsymbol{u}^\mathrm{s} = \Pi_\mathrm{S}^T\delta\boldsymbol{u}_\mathrm{S}
	+ (\Pi_\Sigma^\mathrm{s})^T\delta\boldsymbol{u}_\Sigma,
	\qquad
	\delta\boldsymbol{u}^\mathrm{f} = \Pi_\mathrm{F}^T\delta\boldsymbol{u}_\mathrm{F}
	+ (\Pi_\Sigma^\mathrm{f})^T\delta\boldsymbol{u}_\Sigma.
	\label{eq:app_prolong}
\end{equation}
Substituting \eqref{eq:app_prolong} into the linearized subproblems
\eqref{eq:app_solidlin} and
\eqref{eq:app_fluidlin_mom}--\eqref{eq:app_fluidlin_cont}, followed by testing each equation
with the corresponding boolean restriction and ordering the rows as
$(\mathrm{S},\Sigma,\mathrm{F},p)$ recovers the monolithic system
\eqref{eq:discBlockMonolithic}. A similar restriction is developed for the residuals with
 $R_\mathrm{S} = \Pi_\mathrm{S}{R}^\mathrm{s}$,
$R_{m,\mathrm{F}} = \Pi_\mathrm{F}\boldsymbol{R}^\mathrm{f}_m$ and
$R_c = \Pi_p{R}_c$.

For the monolithic system, the restriction of the structural and fluid operators $A^\mathrm{s}$
and $A^\mathrm{f}$ leads to the diagonal block $A_\Sigma = A_\Sigma^\mathrm{s} + A_\Sigma^\mathrm{f}$ which
recovers \eqref{eq:itfMonolithic}. Similarly, the residual $R_{m,\Sigma} = \Pi_\Sigma^\mathrm{s}\boldsymbol{R}^\mathrm{s} + \Pi_\Sigma^\mathrm{f}\boldsymbol{R}^\mathrm{f}$, the second term of which forms the traction $\mathbf{t}_\Sigma$ applied on the
structure within the partitioned framework. The symmetry of $A^\mathrm{s}$ makes the $(\mathrm{S},\Sigma)$
block the transpose $A_{\Sigma\mathrm{S}}^T$ of the $(\Sigma,\mathrm{S})$ block,
whereas the non-symmetry of $K^\mathrm{conv}$ requires $A_{\Sigma\mathrm{F}}$
and $A_{\mathrm{F}\Sigma}$ to be retained separately. Within the partitioned framework, the augmentation operator $\tilde{C}$ is
added into the pressure block alongside $C_\tau$.

\begin{remark}[Stabilization blocks]
	The operators $G_\tau$ and $C_\tau$ of
	\eqref{eq:app_fluidlin_mom}--\eqref{eq:app_fluidlin_cont} are retained in the
	solver but omitted from \eqref{eq:discBlockMonolithic}. Dropping $G_\tau$ makes
	the discrete pressure gradient the exact adjoint of the discrete divergence, so
	that the off-diagonal blocks of \eqref{eq:discBlock2} are $\mathcal{B}$ and
	$\mathcal{B}^T$, while dropping the restriction of $C_\tau$ renders the
	$(2,2)$ block null and makes the pressure Schur complement
	\eqref{eq:pressureSchur} the pure velocity-elimination operator
	$-\mathcal{B}\mathcal{A}_\mathrm{m}^{-1}\mathcal{B}^T$. As $C_\tau$ is
	symmetric positive semi-definite and independent of the partitioning, its
	inclusion shifts the monolithic pressure equation
	\eqref{eq:monolithicPressure} and the partitioned update \eqref{eq:ppeFSI}
	identically, and affects neither the added-mass mechanism identified through
	\eqref{eq:S1split} nor the cancellation of $\tilde{C}$ at the fixed point.
\end{remark}
	
	\section{Appendix B. Algebraic details of the monolithic and partitioned solutions}
	
	\subsection{Inversion of the monolithic momentum operator}\label{app2}
	
	The pressure Schur complement of \eqref{eq:discBlock2} requires
	$\mathcal{A}_\mathrm{m}^{-1}$, which is obtained from the $LDU$ factorization
	\eqref{eq:LDU} by forward and backward substitution. Writing
	$\mathcal{A}_\mathrm{m}\,\delta{\boldsymbol{u}} = {R}_\mathrm{m,F}$ and solving
	the three factors in turn yields
	\begin{equation}
		\mathcal{A}_\mathrm{m}^{-1} = \begin{bmatrix}
			A_\mathrm{S}^{-1} + A_\mathrm{S}^{-1}A_{\Sigma\mathrm{S}}^T{S}_3^{-1}A_{\Sigma\mathrm{S}}A_\mathrm{S}^{-1}
			& -A_\mathrm{S}^{-1}A_{\Sigma\mathrm{S}}^T{S}_3^{-1}
			& A_\mathrm{S}^{-1}A_{\Sigma\mathrm{S}}^TS_1^{-1}A_{\Sigma\mathrm{F}}S_2^{-1}\\[2pt]
			-{S}_3^{-1}A_{\Sigma\mathrm{S}}A_\mathrm{S}^{-1}
			& {S}_3^{-1}
			& -S_1^{-1}A_{\Sigma\mathrm{F}}S_2^{-1}\\[2pt]
			S_2^{-1}A_{\mathrm{F}\Sigma}S_1^{-1}A_{\Sigma\mathrm{S}}A_\mathrm{S}^{-1}
			& -S_2^{-1}A_{\mathrm{F}\Sigma}S_1^{-1}
			& S_2^{-1}
		\end{bmatrix},
		\label{eq:app_Aminv}
	\end{equation}
	where $S_1$ and $S_2$ are given by \eqref{eq:schurmonolithic} and
	\begin{equation}
		{S}_3^{-1} = S_1^{-1} +
		S_1^{-1}A_{\Sigma\mathrm{F}}S_2^{-1}A_{\mathrm{F}\Sigma}S_1^{-1}.
		\label{eq:app_S1tilde}
	\end{equation}
	
	Since the velocity-pressure coupling acts only on the interface and interior
	fluid degrees of freedom, with $\mathcal{B} = [\,0 \;\; B_\Sigma \;\;
	B_\mathrm{F}\,]$, the pressure Schur complement follows as
	\begin{align}
		\mathcal{B}\mathcal{A}_\mathrm{m}^{-1}\mathcal{B}^T
		&= B_\Sigma{S}_3^{-1}B_\Sigma^T
		- B_\Sigma S_1^{-1}A_{\Sigma\mathrm{F}}S_2^{-1}B_\mathrm{F}^T
		- B_\mathrm{F}S_2^{-1}A_{\mathrm{F}\Sigma}S_1^{-1}B_\Sigma^T
		+ B_\mathrm{F}S_2^{-1}B_\mathrm{F}^T \nonumber\\
		&= B_\Sigma S_1^{-1}B_\Sigma^T + B_\mathrm{F}S_2^{-1}B_\mathrm{F}^T
		+ B_\Sigma S_1^{-1}A_{\Sigma\mathrm{F}}S_2^{-1}A_{\mathrm{F}\Sigma}S_1^{-1}B_\Sigma^T
		\nonumber\\
		&\quad - B_\Sigma S_1^{-1}A_{\Sigma\mathrm{F}}S_2^{-1}B_\mathrm{F}^T
		- B_\mathrm{F}S_2^{-1}A_{\mathrm{F}\Sigma}S_1^{-1}B_\Sigma^T,
		\label{eq:app_schurfull}
	\end{align}
	where $S_3^{-1}$ has been eliminated using \eqref{eq:app_S1tilde}. In the
	inviscid, pressure-driven limit, the couplings $A_{\Sigma\mathrm{F}}$ and
	$A_{\mathrm{F}\Sigma}$ vanish, reducing $S_2$ to $A_\mathrm{F}$ in
	\eqref{eq:schurmonolithic}, and the last three terms drop out, giving
	\begin{equation}
		\mathcal{B}\mathcal{A}_\mathrm{m}^{-1}\mathcal{B}^T =
		B_\Sigma S_1^{-1}B_\Sigma^T + B_\mathrm{F}A_\mathrm{F}^{-1}B_\mathrm{F}^T,
		\label{eq:app_schurinviscid}
	\end{equation}
	recovering \eqref{eq:Schurinviscid}. Substituting \eqref{eq:app_Aminv} in
	\eqref{eq:solutionMonolithic} and applying the same limit, the monolithic
	interface pressure satisfies
	\begin{equation}
		\left(B_\Sigma S_1^{-1}B_\Sigma^T +
		B_\mathrm{F}A_\mathrm{F}^{-1}B_\mathrm{F}^T\right)\delta p =
		-B_\Sigma S_1^{-1}A_{\Sigma\mathrm{S}}A_\mathrm{S}^{-1}R_\mathrm{S}
		+ B_\Sigma S_1^{-1}R_\mathrm{m,\Sigma}
		+ B_\mathrm{F}A_\mathrm{F}^{-1}R_{\mathrm{m,F}},
		\label{eq:app_monolithicPressure}
	\end{equation}
	which is \eqref{eq:monolithicPressure}.
	
	\subsection{Proof of Proposition \ref{prop:contraction}}\label{app3}
	
	Eliminating $\delta\boldsymbol{u}_\mathrm{F}^{k+1}$ from \eqref{eq:fluidSolveAug}
	and writing $S_\mathrm{F} = \tilde{C} +
	B_\mathrm{F}A_\mathrm{F}^{-1}B_\mathrm{F}^T$ gives
	\begin{align}
		S_\mathrm{F}\,\delta p^{k+1} &=
		B_\mathrm{F}A_\mathrm{F}^{-1}\left(R_\mathrm{m,F} -
		A_{\mathrm{F}\Sigma}\delta\boldsymbol{u}_\Sigma^{k+1}\right)
		+ B_\Sigma\delta\boldsymbol{u}_\Sigma^{k+1} + \tilde{C}\delta p^k \nonumber\\
		&= B_\mathrm{F}A_\mathrm{F}^{-1}R_\mathrm{m,F}
		+ \left(B_\Sigma - B_\mathrm{F}A_\mathrm{F}^{-1}A_{\mathrm{F}\Sigma}\right)
		\delta\boldsymbol{u}_\Sigma^{k+1} + \tilde{C}\delta p^k.
		\label{eq:pressureItf}
	\end{align} Substituting the interface velocity \eqref{eq:itfVel},
	\begin{align}
		S_\mathrm{F}\delta p^{k+1} &= BA_\mathrm{F}^{-1}\boldsymbol{R}^\mathrm{f} +
		(B_\Sigma -
		BA_\mathrm{F}^{-1}A_{\mathrm{F}\Sigma})S_\mathrm{S}^{-1}\left(\boldsymbol{R}^\mathrm{i}
		- A_{\Sigma \mathrm{S}}A_\mathrm{S}^{-1}\boldsymbol{R}^\mathrm{s}\right) \nonumber\\ 
		&- (B_\Sigma -
		BA_\mathrm{F}^{-1}A_{\mathrm{F}\Sigma})S_\mathrm{S}^{-1}A_{\Sigma\mathrm{F}}\boldsymbol{u}_\mathrm{f}^{k}
		- (B_\Sigma -
		BA_\mathrm{F}^{-1}A_{\mathrm{F}\Sigma})S_\mathrm{S}^{-1}A_\Sigma^\mathrm{f}\delta\boldsymbol{u}_\Sigma ^k\nonumber\\
		&+ \left( \tilde{C} - (B_\Sigma -
		BA_\mathrm{F}^{-1}A_{\mathrm{F}\Sigma})S_\mathrm{S}^{-1}B_\Sigma^T\right)\delta p^k,
		\nonumber
	\end{align} which reduces to \eqref{eq:ppeFSI}, when pressure and inertia
	dominant flow-features are considered by setting $A_{\mathrm{F}\Sigma} = 0, 
	A_{\Sigma\mathrm{F}} = 0$.
	
	From Eqs. \eqref{eq:schurmonolithic}, \eqref{eq:itfMonolithic} and
	\eqref{eq:Schurinviscid}, and invoking the Woodbury matrix identity,
	\begin{equation}
		S_1^{-1} = S_\mathrm{S}^{-1} - \underbrace{S_\mathrm{S}^{-1}(I +
			A_\Sigma^\mathrm{f}S_\mathrm{S}^{-1})^{-1}A_\Sigma^\mathrm{f}S_\mathrm{S}^{-1}}_{S_\Sigma ^{-1}}.
		\label{eq:woodbury}
	\end{equation} Substituting \eqref{eq:woodbury} in
	\eqref{eq:monolithicPressure} and using the definition of $\boldsymbol{G}$ from
	\eqref{eq:ppeFSI},
	\begin{equation}
		(B_\Sigma S_\mathrm{S}^{-1}B_\Sigma^T +
		B_\mathrm{F}A_\mathrm{F}^{-1}B_\mathrm{F}^T)\delta p_\mathrm{m} -
		B_\Sigma S_\Sigma ^{-1}B_\Sigma^T\delta p_\mathrm{m} = \boldsymbol{G} +
		B_\Sigma S_\Sigma ^{-1}A_{\Sigma \mathrm{S}}A_\mathrm{S}^{-1}\boldsymbol{R}^\mathrm{s}
		- B_\Sigma S_\Sigma ^{-1}\boldsymbol{R}^\mathrm{i}. \label{eq:monolithicW}
	\end{equation} Eliminating $\boldsymbol{G}$ between Eqs. \eqref{eq:ppeFSI} and
	\eqref{eq:monolithicW},
	\begin{align}
		(B_\Sigma S_\mathrm{S}^{-1}B_\Sigma^T +
		B_\mathrm{F}A_\mathrm{F}^{-1}B_\mathrm{F}^T)\delta p_\mathrm{m} &= \left(\tilde{C} +
		B_\mathrm{F}A_\mathrm{F}^{-1}B_\mathrm{F}^T\right)\delta p^{k+1} - \left(\tilde{C} -
		B_\Sigma S_\mathrm{S}^{-1}B_\Sigma^T\right)\delta p^k\nonumber\\&-B_\Sigma S_\Sigma ^{-1}\underbrace{\left(\boldsymbol{R}^\mathrm{i}
			- A_{\Sigma \mathrm{S}}A_\mathrm{S}^{-1}\boldsymbol{R}^\mathrm{s} -
			B_\Sigma^T\delta p_\mathrm{m}\right)}_{S_1\boldsymbol{u}_{\mathrm{i},\mathrm{m}}}+
		B_\Sigma S_\mathrm{S}^{-1}A_\Sigma^\mathrm{f}\delta\boldsymbol{u}_\Sigma ^k.
		\nonumber
	\end{align} Subtracting $\left(\tilde{C} +
	B_\mathrm{F}A_\mathrm{F}^{-1}B_\mathrm{F}^T\right)\delta p_\mathrm{m}$ from both sides,
	\begin{equation}
		\left(\tilde{C} + B_\mathrm{F}A_\mathrm{F}^{-1}B_\mathrm{F}^T\right)(\delta p^{k+1}-\delta p_\mathrm{m}) =
		\left(\tilde{C} - B_\Sigma S_\mathrm{S}^{-1}B_\Sigma^T\right)(\delta p^k -
		\delta p_\mathrm{m}) +B_\Sigma S_\Sigma ^{-1}S_1\delta\boldsymbol{u}_{\Sigma,\mathrm{m}} -
		B_\Sigma S_\mathrm{S}^{-1}A_\Sigma^\mathrm{f}\delta\boldsymbol{u}_\Sigma ^k.
		\label{eq:iteration1}
	\end{equation}
	The product $S_\Sigma ^{-1}S_1$ simplifies as
	\begin{align}
		S_\Sigma ^{-1}S_1 &= S_\mathrm{S}^{-1}(I +
		A_\Sigma^\mathrm{f}S_\mathrm{S}^{-1})^{-1}A_\Sigma^\mathrm{f}S_\mathrm{S}^{-1}S_\mathrm{S}(I
		+ S_\mathrm{S}^{-1}A_\Sigma^\mathrm{f}) \nonumber\\
		&= S_\mathrm{S}^{-1}(I + A_\Sigma^\mathrm{f}S_\mathrm{S}^{-1})^{-1}A_\Sigma^\mathrm{f}(I +
		S_\mathrm{S}^{-1}A_\Sigma^\mathrm{f}), \nonumber
	\end{align} and using the push-through identity,
	\begin{align}
		S_\Sigma ^{-1}S_1 &= S_\mathrm{S}^{-1}A_\Sigma^\mathrm{f}(I +
		S_\mathrm{S}^{-1}A_\Sigma^\mathrm{f})^{-1}(I + S_\mathrm{S}^{-1}A_\Sigma^\mathrm{f}) =
		S_\mathrm{S}^{-1}A_\Sigma^\mathrm{f}. \label{eq:pushthrough}
	\end{align}
	With \eqref{eq:pushthrough}, \eqref{eq:iteration1} simplifies to
	\begin{equation}
		\left(\tilde{C} + B_\mathrm{F}A_\mathrm{F}^{-1}B_\mathrm{F}^T\right)(\delta p^{k+1}-\delta p_\mathrm{m})
		=\left(\tilde{C} - B_\Sigma S_\mathrm{S}^{-1}B_\Sigma^T\right)(\delta p^k -
		\delta p_\mathrm{m})
		-B_\Sigma S_\mathrm{S}^{-1}A_\Sigma^\mathrm{f}\left(\delta\boldsymbol{u}_\Sigma ^k -
		\delta\boldsymbol{u}_{\Sigma,\mathrm{m}} \right), \nonumber
	\end{equation} which is the second of the equations in \eqref{eq:errBlock}. Finally, subtracting the
	monolithic interface velocity $\delta\boldsymbol{u}_{\Sigma,\mathrm{m}} $ from \eqref{eq:itfVel} with $A_{\Sigma\mathrm{F}} = 0$ gives \eqref{eq:errBlock}.
	
	\subsection{Solving the algebraic system of equations for the partitioned FSI problem}
\label{sec:solvers}
The structural subproblem is solved with the finite element code CalculiX
\citep{dhondt2004calculix}, which is coupled to an in-house solver for the fluid
flow and the mesh deformation. The fluid mesh is partitioned using METIS
\citep{karypis1998metis}, and the fluid solver follows an MPI-based
distributed-memory parallelization paradigm.

Within each coupling iteration, the structural subproblem \eqref{eq:solidSolve}
is governed by the operator
\begin{equation}
	\mathcal{K}^\mathrm{s} = \begin{bmatrix}
		A_\mathrm{S} & A_{\Sigma\mathrm{S}}^T\\[2pt]
		A_{\Sigma\mathrm{S}} & A_\Sigma^\mathrm{s}
	\end{bmatrix},
	\label{eq:Ks}
\end{equation}
which is assembled within the multithreaded framework of CalculiX and solved
directly using the supernodal sparse Cholesky factorization of CHOLMOD from the
SuiteSparse library \citep{chen2008cholmod}.

The linear systems arising on the fluid side are solved using the Krylov
subspace solvers provided by PETSc \citep{petsc-user-ref}. The mesh deformation
equations are solved using the conjugate gradient (CG) method. The augmented
fluid subproblem \eqref{eq:fluidSolveAug} is governed by the saddle-point
operator
\begin{equation}
	\mathcal{K}^\mathrm{f} = \begin{bmatrix}
		A_\mathrm{F} & B_\mathrm{F}^T\\[2pt]
		B_\mathrm{F} & -\tilde{C}
	\end{bmatrix},
	\label{eq:Kf}
\end{equation}
and is solved using the flexible GMRES (FGMRES) method \citep{saad1993flexible}
with 40 Krylov subspace vectors. The system is right-preconditioned with the
upper block-triangular operator
\begin{equation}
	\mathcal{P}^\mathrm{f} = \begin{bmatrix}
		A_\mathrm{F} & B_\mathrm{F}^T\\[2pt]
		0 & -\hat{S}
	\end{bmatrix},
	\label{eq:Pf}
\end{equation}
where $\hat{S}$ approximates the augmented pressure Schur complement
$\tilde{C} + B_\mathrm{F}A_\mathrm{F}^{-1}B_\mathrm{F}^T$ appearing in
\eqref{eq:ppeFSI}. Since $\tilde{C}$ is diagonal, it enters $\hat{S}$ without
altering its sparsity pattern. Applying $(\mathcal{P}^\mathrm{f})^{-1}$ to a
vector $[r_u;\, r_p]$ requires two inner solves, namely
$\hat{S}z_p = -r_p$, followed by
$A_\mathrm{F}z_u = r_u - B_\mathrm{F}^Tz_p$. The pressure block is solved
using CG, and the momentum block using GMRES with 20 Krylov subspace vectors. The inexact inner solves necessitate the use of the flexible variant of GMRES.

	
	
	
	
		
		
		
	
	\bibliographystyle{unsrtnat}
	%
	\bibliography{ref}
	
	
	
\end{document}